\documentclass[a4paper, 11pt]{elsarticle}
\usepackage{mathtools}
\usepackage{extarrows}
\usepackage{amsmath}
\usepackage{amssymb}
\usepackage{hyperref}
\usepackage{lipsum}
\usepackage{arydshln}
\usepackage[T1]{fontenc}
\usepackage{amssymb}
\usepackage{graphicx}
\usepackage{amsfonts}
\usepackage{booktabs}
\usepackage{amsmath}
\usepackage{geometry}
\usepackage{mathtools}
\usepackage{epsfig}
\usepackage{multicol}
\usepackage{amsmath, mathtools}
\usepackage{CJK}
\usepackage{amsmath}
\usepackage{amsthm}
\usepackage{setspace}
\usepackage{graphicx}
\usepackage{epstopdf}
\usepackage{subeqnarray}
\usepackage{mathrsfs}
\usepackage{bbding}
\usepackage{cases}

\usepackage{amsmath, amssymb, amsfonts, mathtools, bm}
\usepackage{amsthm}
\usepackage{graphicx}
\usepackage{subcaption}
\usepackage{geometry}
\usepackage{setspace}
\usepackage{multicol}
\usepackage{booktabs}
\usepackage{arydshln}

\usepackage{hyperref}
\usepackage{lineno}
\modulolinenumbers[5]

\usepackage{extarrows}
\usepackage{lipsum}
\usepackage{cases}
\usepackage{mathrsfs}
\usepackage{float}
\usepackage{tikz}
\usepackage{xcolor}
\usepackage{epstopdf}
\usepackage{bbding}
\usepackage{cleveref}

\journal{the journal}

\newtheorem{theorem}{Theorem}
\newtheorem{lemma}{Lemma}
\numberwithin{lemma}{section}
\newtheorem{rmk}{Remark}
\newtheorem{prop}{Proposition}
\numberwithin{prop}{section}

\numberwithin{corollary}{section}

\numberwithin{definition}{section}
\numberwithin{equation}{section}
\newtheorem{remark}{Remark}
\numberwithin{remark}{section}
\crefname{prop}{Proposition}{Propositions}

\newcounter{mycounter}
\makeatletter
\newcommand{\xlongless}[2][]{\ext@arrow 0099{\longlessfill@}{#1}{#2}}
\newcommand{\xlonggreater}[2][]{\ext@arrow 0099{\longgreaterfill@}{#1}{#2}}
\newcommand{\xlongleq}[2][]{\ext@arrow 0099{\longleqfill@}{#1}{#2}}
\newcommand{\xlonggeq}[2][]{\ext@arrow 0099{\longgeqfill@}{#1}{#2}}
\newcommand{\longlessfill@}{\arrowfill@{\leftarrow}\relbar\relbar}
\newcommand{\longgreaterfill@}{\arrowfill@\relbar\relbar{\rightarrow}}
\newcommand{\longleqfill@}{\arrowfill@{\leftarrow}\relbar\relbar}
\newcommand{\longgeqfill@}{\arrowfill@\relbar\relbar{\rightarrow}}
\makeatother

\biboptions{sort&compress}
\begin{document}
	\hypersetup{pdftitle={Asymptotic Analysis of N-Elliptic Localized Solutions for the Fokas--Lenells Equation},pdfauthor={Wang Tang, Guo-Fu Yu}}
	\title{Asymptotic analysis of N-elliptic localized solutions for the Fokas--Lenells equation}
	\author[1]{Bao-Feng Feng}
	\author[2]{Wang Tang\corref{cor1}}
	\ead{terrencet@scut.edu.cn}
	\author[3]{Guo-Fu Yu}
	\ead{gfyu@sjtu.edu.cn}
	\cortext[cor1]{Corresponding author}
	\address[1]{School of Mathematical and Statistics Sciences,
		The University of Texas Rio Grande Valley, Brownsville, United States}
	\address[2]{School of Mathematics, South China University of Technology, Guangzhou, China}
	\address[3]{School of Mathematics, Shanghai Jiao Tong University, Shanghai, China}
	\begin{abstract}
	This paper investigates the $N$-elliptic localized solutions of the Fokas–Lenells equation. Based on the corresponding Lax pair, the Weierstrass elliptic functions are adopted to construct the elliptic function solutions and the fundamental solution matrix of the equation. The $N$-elliptic localized solutions are further derived via the $N$-fold Darboux–B\"acklund transformation. By virtue of the Cauchy determinant expressed with sigma functions, the asymptotic behaviors of the obtained solutions are systematically analyzed along and between their propagation directions, and the symmetry properties of these solutions are established.
	\end{abstract}

	\maketitle
	
	\section{Introduction}
The nonlinear Schr\"odinger (NLS) equation is the standard description of how a slowly varying wave envelope evolves in a weakly nonlinear, dispersive medium, and it governs phenomena as diverse as wave trains on deep water and signal transmission along optical fibers~\cite{zakharov1972,ablowitz1981}. Zakharov and Shabat established its integrability through the inverse scattering transform~\cite{zakharov1972}, building on the technique originally devised for the Korteweg--de Vries equation~\cite{ggkm1967,lax1968} and later cast into a unified formalism by Ablowitz, Kaup, Newell, and Segur~\cite{akns1974}.

When the carrier pulse is shortened to the sub-picosecond or femtosecond scale, the approximations underlying the NLS model omit contributions such as self-steepening and spatio-temporal dispersion. Once these higher-order effects begin to influence the evolution, one is led to integrable models of a more refined type~\cite{agrawal2019,moll2002}. Two such refinements arise from the Kaup--Newell (KN) hierarchy: the derivative NLS (DNLS) equation~\cite{kaup1978}, whose derivative-type nonlinearity captures self-steepening, and the Fokas--Lenells (FL) equation~\cite{Fokas1995,Lenells2009}, a negative flow of the same hierarchy that further incorporates spatio-temporal dispersion. While DNLS adequately describes self-steepening at sub-picosecond scales, the FL equation extends the regime of validity to even shorter pulses where both self-steepening and spatio-temporal dispersion are of comparable importance.
	
	The present paper is devoted to the FL equation
	\begin{equation}
		\mathrm{i}q_t-\beta q_{tx}+\gamma q_{xx}+\alpha|q|^2\!\left(q+\mathrm{i}\beta q_x\right)=0,\qquad \alpha=\pm 1,
	\end{equation}
	where $q=q(x,t)$ is a complex-valued function on $\mathbb{R}^2$. The model was first obtained by Fokas through a bi-Hamiltonian construction~\cite{Fokas1995} and was subsequently analyzed in depth by Lenells and Fokas~\cite{Lenells2009}. After a suitable change of variables it takes the streamlined form
	\begin{equation}\label{Eq-FL}
		u_{xt}+u-\mathrm{i}|u|^2u_x=0,
	\end{equation}
	which is the first negative member of the KN hierarchy~\cite{kaup1978}. Both the inverse scattering transform and the bi-Hamiltonian structure confirm its integrability~\cite{Lenells2009,Lashkin2021}. On this foundation, a wealth of soliton and rational solutions on constant backgrounds have been produced through Hirota's bilinear scheme and Darboux transformations~\cite{Wang2020FokasLenells,He2012,Xu2013,Xu2014a,Xu2014b,Li2015}, while the algebro-geometric content of the FL hierarchy has been elucidated by means of finite-gap integration~\cite{Sun2012,Zhao2013b} and its vector extensions~\cite{Geng2017}.
	
	Interest in nontrivial backgrounds was spurred by Grinevich and Santini~\cite{Grinevich2018}, who showed that a weakly perturbed NLS Cauchy problem generically produces genus-one or genus-two evolution. In the genus-one regime an elliptic background appears prior to the onset of rogue waves and Akhmediev breathers. Because this kind of backgrounds can be realized experimentally, the explicit description of the solutions they support has developed into an active line of inquiry. The constructions reported to date fall naturally into two complementary groups.
	
	The first group starts from the reduction of algebro-geometric (finite-gap) data. In this spirit, multiphase cnoidal-wave modulations for the defocusing NLS equation were recovered by degenerating finite-gap solutions, while multi-Akhmediev breathers riding dnoidal backgrounds were treated by complex coordinate transformations~\cite{Belokolos1994,Gesztesy2003}. The inverse scattering apparatus was likewise used to render multi-soliton solutions as theta functions over elliptic backgrounds~\cite{Takahashi2012,Takahashi2016}. More recently, a direct-linearization scheme has yielded elliptic function solutions for a class of Boussinesq-type lattice equations~\cite{Nijhoff2023}.
	
	The second group takes elementary seed solutions and, through transformation techniques, dresses them into richer structures on elliptic backgrounds. By Darboux--B\"acklund transformations, Shin generated cnoidal waves~\cite{Shin2003} and solitons superimposed on cnoidal backgrounds~\cite{Shin2004} for the coupled NLS systems, and the same procedure subsequently supplied dark solitons over cnoidal backgrounds for the defocusing NLS equation~\cite{Shin2005}. Solutions on elliptic backgrounds have also been retrieved via nonlocal symmetry localization~\cite{Hu2012} and consistent Riccati expansion~\cite{Lou2015}, while numerical implementations of the Darboux--B\"acklund transformation have produced higher-order breathers and rogue waves over cnoidal and dnoidal backgrounds~\cite{Ashour2022,Kedziora2014}.
	
	Alongside these developments, several systematic frameworks for generating exact solutions over elliptic backgrounds through Darboux--B\"acklund transformations have taken shape. The first is grounded in the nonlinearization theory of Lax pairs~\cite{cao1999relation,Cao1994a,Cao1994b}. Exploiting this idea, Chen and Pelinovsky derived algebraically decaying solitons, rogue waves, and further localized structures over elliptic backgrounds of the NLS equation~\cite{chen2018rogue}, after which the construction was carried over to the modified Korteweg--de Vries (mKdV) and DNLS settings~\cite{chen2019periodic,chen2021rogue,chen2021modulational}. Imposing Bargmann constraints with two eigenfunctions widened the scheme to doubly periodic backgrounds of the NLS equation~\cite{chen2020periodic}, and discrete counterparts, including the Ablowitz--Ladik lattice, followed soon afterwards~\cite{chen2023periodic,chen2024rogue,Chen2019}. Introducing Baker--Akhiezer functions on algebraic curves, Geng, Li, and coworkers devised an inverse algebro-geometric (IAG) approach capable of treating integrable systems carrying $3\times 3$ or larger Lax pairs, among them the vector Geng--Li model~\cite{geng2022vector}, the Yajima--Oikawa system~\cite{li2023rogue}, and the FL equation itself~\cite{li2023roguefl,Li2016,Li2020b}. The IAG construction was afterwards transferred to semidiscrete equations as well~\cite{li2025breather}. In a separate framework, Ling and collaborators replace the spectral parameter by a uniform parameter $z$ and represent multi-breather solutions through theta functions over a range of elliptic backgrounds for the Ablowitz--Kaup--Newell--Segur (AKNS) equations~\cite{ling2024multi,ling2023multi,ling2023elliptic,Feng2019,feng2020multi}. Higher-order rogue waves over elliptic backgrounds were thereafter secured for these systems by the generalized Darboux transformation~\cite{ling2024elliptic,Guo2014,Guo2015}. Within the lattice context, the elliptic solutions of Boussinesq-type equations obtained by Nijhoff, Sun, and Zhang~\cite{Nijhoff2023} further attest to the reach of these direct methods.
	
	A recurring theme underlying this body of work is the long-time disintegration of solutions. The soliton resolution conjecture asserts that, for generic initial data, a weakly nonlinear dispersive flow splits asymptotically into a finite collection of solitons, a dispersive part obeying a linear law, and a remainder that decays as $t\to+\infty$~\cite{tao2008why}. Originating in the 1970s and 1980s from work on the Korteweg--de Vries equation~\cite{miura1976korteweg,segur1973korteweg} and the accompanying numerical experiments~\cite{ivancevic2010quantum}, the conjecture received its first rigorous confirmation through inverse scattering~\cite{segur1976asymptotic1,segur1976asymptotic2,novokshenov1980asymptotic,schuur1986asymptotic}. The nonlinear steepest-descent treatment of Riemann--Hilbert problems due to Deift and Zhou later secured soliton resolution for a broad range of integrable systems~\cite{deift1993steepest,yang2023soliton,jenkins2018soliton}.
	
	Although elliptic-background solutions of the FL equation have already been examined within the IAG approach~\cite{li2023roguefl}, a compact and fully explicit description of the $N$-elliptic localized solutions generated by the $N$-fold Darboux--B\"acklund transformation, together with their long-time behavior, has not yet been provided. The aim of this paper is to fill that gap. Following the uniform-parameter philosophy of~\cite{ling2023elliptic,lingtang2026dnls}, we employ Weierstrass elliptic functions to construct the elliptic function solutions of the FL equation and the associated fundamental solution matrix, apply the $N$-fold Darboux--B\"acklund transformation to obtain the $N$-elliptic localized solutions, and express them in compact form via a sigma-function version of the Cauchy determinant. A detailed asymptotic analysis then shows that, as $t\to\pm\infty$, the $N$-elliptic localized solution separates into first-order elliptic localized waves that travel at distinct velocities over a shared elliptic background and collide elastically, with a symmetry condition upgrading these collisions to strictly elastic ones. This furnishes a verification of the soliton resolution picture for exact solutions on elliptic backgrounds within the KN hierarchy.
	
The remainder of this paper is organized as follows. Section~\ref{sec:elliptic} constructs elliptic function solutions to the FL equation using Weierstrass functions and assembles the fundamental solution matrix for the associated Lax pair. Section~\ref{sec:N-elliptic} derives the $N$-elliptic localized solutions via the $N$-fold Darboux–B\"acklund transformation and presents them in a compact, fully explicit form in terms of the sigma-function Cauchy determinant. Section~\ref{sec:asymptotic} is devoted to the asymptotic analysis of these solutions. We characterize their asymptotic behavior along and between their propagation directions, demonstrate that the $N$-elliptic localized solution decomposes into individual first-order elliptic localized waves with elastic collisions, and derive their symmetry properties. Section~\ref{sec:dynamics} illustrates the dynamics of one- and two-elliptic localized solutions by combining analytical results with graphical plots. Section~\ref{sec:conclusion} contains concluding remarks. The appendix compiles all relevant identities and properties of the Weierstrass elliptic functions utilized in this work.
	
	\section{The elliptic function solutions to the Fokas--Lenells equation}\label{sec:elliptic}
This section serves two purposes. We first construct explicit elliptic function solutions to the FL equation. We then assemble the fundamental solution matrix for the corresponding Lax pair from these solutions, and this matrix will serve as the seed for the Darboux–B\"acklund transformation in Section~\ref{sec:N-elliptic}.

	The FL equation \eqref{Eq-FL} admits the following Lax pair:
	\begin{equation}\label{Eq-Lax pair FL}
		\begin{split}
			\bm{\Psi}_x(x,t;\lambda) &= \mathbf{U}(\mathbf{Q};\lambda)\bm{\Psi}(x,t;\lambda), \quad 
			\mathbf{U}(\mathbf{Q};\lambda) = -\mathrm{i}\bm{\sigma}_3\lambda^{-2} + \mathbf{Q}_x\lambda^{-1},\\[4pt]
			\bm{\Psi}_t(x,t;\lambda) &= \mathbf{V}(\mathbf{Q};\lambda)\bm{\Psi}(x,t;\lambda), \quad 
			\mathbf{V}(\mathbf{Q};\lambda) = -\frac{\mathrm{i}}{4}\bm{\sigma}_3\lambda^2 + \frac{\mathrm{i}}{2}\bm{\sigma}_3\mathbf{Q}\lambda - \frac{\mathrm{i}}{2}\bm{\sigma}_3\mathbf{Q}^2,
		\end{split}
	\end{equation}
	where \(\lambda\in\mathbb{C}\) is the spectral parameter, \(\bm{\sigma}_3 = \operatorname{diag}(1,-1)\) is the third Pauli matrix, and \(\mathbf{Q} = \begin{pmatrix} 0 & u \\ -u^* & 0 \end{pmatrix}\). The compatibility condition
	$
	\bm{\Psi}_{xt}(x,t;\lambda) = \bm{\Psi}_{tx}(x,t;\lambda)
	$
	of \eqref{Eq-Lax pair FL} gives rise to the zero-curvature equation
	\begin{equation}\label{Eq-zero-curvature-equation}
		\mathbf{U}_t - \mathbf{V}_x + [\mathbf{U}, \mathbf{V}] = 0,
	\end{equation}
	which yields the FL equation \eqref{Eq-FL}.
	
Consider the following stationary zero-curvature equations:
\begin{equation}\label{Eq-Lax equation FL}
	\mathbf{L}_x(x,t;\lambda) = [\mathbf{U}(\mathbf{Q};\lambda), \mathbf{L}(x,t;\lambda)],\quad 
	\mathbf{L}_t(x,t;\lambda) = [\mathbf{V}(\mathbf{Q};\lambda), \mathbf{L}(x,t;\lambda)].
\end{equation}
By using the Jacobi identity for commutators, one can verify that the compatibility condition
\[
\mathbf{L}_{xt}(x,t;\lambda) = \mathbf{L}_{tx}(x,t;\lambda)
\]
for the stationary zero-curvature equations \eqref{Eq-Lax equation FL} is guaranteed by the zero-curvature equation \eqref{Eq-zero-curvature-equation}, which also gives rise to the FL equation \eqref{Eq-FL}.

To obtain elliptic function solutions, we employ the following ansatz:
\begin{align}\label{Eq-L-ansatz}
	\mathbf{L}(x,t;\lambda)=s_0\mathbf{U}(\mathbf{Q};\lambda)+4\mathbf{V}(\mathbf{Q};\lambda)+\frac{\mathrm{i}}{2}s_1\bm{\sigma}_3,
\end{align}
where \(s_0,s_1\in\mathbb{R}\) are parameters to be determined. Substituting the explicit expressions of \(\mathbf{U}(\mathbf{Q};\lambda)\) and \(\mathbf{V}(\mathbf{Q};\lambda)\) from \eqref{Eq-Lax pair FL}, the matrix \(\mathbf{L}(x,t;\lambda)\) can be rewritten as
\begin{equation}\label{Eq-Lax-matrix}
	\mathbf{L}(x,t;\lambda) =
	\begin{pmatrix}
		-\mathrm{i}L_{11}(x,t;\lambda) & L_{12}(x,t;\lambda) \\
		L_{21}(x,t;\lambda) &  \mathrm{i}L_{11}(x,t;\lambda)
	\end{pmatrix},
\end{equation}
where the entries are given by
\begin{equation}\label{Eq-fgh}
	\begin{split}
		L_{11}(x,t;\lambda) &= \lambda^2 -2\nu-\frac{1}{2}s_1+s_0\lambda^{-2},\\    
		L_{12}(x,t;\lambda) &= s_0 u_x\lambda^{-1}+2\mathrm{i}\lambda u
		= 2\mathrm{i}u\lambda^{-1}(\lambda^2-\mu), \\
		L_{21}(x,t;\lambda) &= -s_0 u_x^*\lambda^{-1}+2\mathrm{i}\lambda u^*
		= 2\mathrm{i}u^*\lambda^{-1}(\lambda^2-\mu^*).
	\end{split}
\end{equation}
Here we have introduced the notations
\begin{equation}\label{Eq-nu-mu}
	\nu:=|u|^2,\qquad \mu:=\frac{\mathrm{i}s_0 u_x}{2u},
\end{equation}
where \(\nu\) denotes the squared modulus of \(u\), and \(\mu\) denotes the auxiliary spectrum.

Substituting the expressions in \eqref{Eq-fgh} into \eqref{Eq-Lax-matrix}, a direct computation yields the explicit form of the determinant:
	\begin{equation}\label{Eq-detL-explicit}
		\begin{split}
			\det(\mathbf{L}) =  \lambda^4 - s_1\lambda^2 + \Big(4\nu^2 - 4(\mu+\mu^*)\nu + 2s_1\nu + \frac{s_1^2}{4} + 2s_0\Big) + \bigl((4|\mu|^2 - 4s_0)\nu - s_0s_1\bigr)\lambda^{-2} + s_0^2\lambda^{-4}.
		\end{split}
	\end{equation}
	For later convenience, we introduce the constants \(s_2, s_3, s_4\) defined by
	\begin{equation}\label{Eq-s2s3s4}
		s_2 := 4\nu^2 - 4(\mu+\mu^*)\nu + 2s_1\nu + \frac{s_1^2}{4} + 2s_0,\qquad
		s_3 := -(4|\mu|^2 - 4s_0)\nu + s_0s_1,\qquad
		s_4 := s_0^2.
	\end{equation}
	Then \(\det(\mathbf{L})\) admits the compact form
	\begin{equation}\label{Eq-P-ansatz}
		\det(\mathbf{L}) = \lambda^4 - s_1\lambda^2 + s_2 - s_3\lambda^{-2} + s_4\lambda^{-4} := P(\lambda).
	\end{equation}
	Applying Abel's formula to the ordinary differential equations \eqref{Eq-Lax equation FL}, we obtain
	\begin{equation}
		(\det(\mathbf{L}))_x = \operatorname{tr}(\mathbf{L}^{-1}\mathbf{L}_x)\det(\mathbf{L}),\qquad 
		(\det(\mathbf{L}))_t = \operatorname{tr}(\mathbf{L}^{-1}\mathbf{L}_t)\det(\mathbf{L}),
	\end{equation}
	which, together with \eqref{Eq-Lax equation FL}, implies
	\begin{equation}
		(\det(\mathbf{L}))_x = (\det(\mathbf{L}))_t = 0.
	\end{equation}
	Consequently, the coefficients \(s_i,\ i=1,2,3,4\) are independent of \((x,t)\). Moreover, it follows from \eqref{Eq-Lax pair FL} that the matrices \(\mathbf{U}(\mathbf{Q};\lambda)\) and \(\mathbf{V}(\mathbf{Q};\lambda)\) obey the symmetry relations
	\begin{equation}\label{Eq-UV-symmetry}
		\mathbf{U}^\dagger(\mathbf{Q};\lambda^*) = -\mathbf{U}(\mathbf{Q};\lambda),\qquad 
		\mathbf{V}^\dagger(\mathbf{Q};\lambda^*) = -\mathbf{V}(\mathbf{Q};\lambda).
	\end{equation}
	Taking the Hermitian conjugate of both sides of \eqref{Eq-Lax equation FL} and using \eqref{Eq-UV-symmetry}, we deduce that \(\mathbf{L}^\dagger(x,t;\lambda^*)\) also satisfies the stationary zero-curvature equation \eqref{Eq-Lax equation FL}. Hence, by the existence and uniqueness of solutions to ordinary differential equations, we conclude that
	\begin{equation}
		\mathbf{L}^\dagger(x,t;\lambda^*) = \mathbf{L}(x,t;\lambda),
	\end{equation}
	which further implies \(s_i\in\mathbb{R}\).
	
	We now show that the squared modulus \(\nu\) propagates as a traveling wave in the variable 
	\begin{equation}\label{Eq-xi}
		\xi = -\frac{1}{2s_0}\Big(x - \frac{s_0}{4}t\Big).
	\end{equation}
	To this end, we establish that, regarded as a function of \(\xi\), \(\nu\) satisfies a fourth-order elliptic equation.
	
	\begin{prop}\label{Prop-ue}
		The squared modulus $\nu$ satisfies the linear equation
		\begin{equation}\label{Eq-nu-travelling-2}
			4\nu_t + s_0\nu_x = 0.
		\end{equation}
		Furthermore, in terms of the travelling coordinate $\xi $ defined in \eqref{Eq-xi}, the function $\nu$ satisfies the elliptic equation
		\begin{equation}\label{Eq-nu-x}
			\nu_\xi = -\frac{1}{2s_0}\sqrt{-R(\nu)},
		\end{equation}
		where $R(\nu)$ is the quartic polynomial
		\begin{equation}\label{Eq-R}
			R(\nu) = 16\nu^4 + 16s_1\nu^3 + \bigl( -48s_0 + 6s_1^2 - 8s_2 \bigr)\nu^2 + \bigl( s_1^3 - 8s_0s_1 - 4s_1s_2 + 16s_3 \bigr)\nu + \Big( \frac{s_1^2}{4} - s_2 + 2s_0 \Big)^{\!2}.
		\end{equation}
	\end{prop}
	
	\begin{proof}
		Substituting the ansatz \eqref{Eq-L-ansatz} into the $x$-part of the Lax pair \eqref{Eq-Lax equation FL} yields 
		\begin{equation}\label{Eq-relation-1}
			s_0\mathbf{U}_x + 4\mathbf{V}_x = 4[\mathbf{U},\mathbf{V}] + \frac{\mathrm{i}s_1}{2}[\mathbf{U},\boldsymbol{\sigma}_3].
		\end{equation}
		Applying the zero-curvature equation \eqref{Eq-zero-curvature-equation} to eliminate $\mathbf{V}_x$ gives
		\begin{equation}\label{Eq-prop1-1}
			4\mathbf{U}_t + s_0\mathbf{U}_x + \mathrm{i}s_1\boldsymbol{\sigma}_3\mathbf{Q}_x\lambda^{-1} = 0.
		\end{equation}
		From the explicit form of the Lax pair \eqref{Eq-Lax pair FL}, the diagonal and off-diagonal parts of $\mathbf{U}$ are
		\begin{equation}\label{Eq-U-diag-off}
			\mathbf{U}^{\mathrm{diag}} = -\mathrm{i}\boldsymbol{\sigma}_3\lambda^{-2}, \qquad 
			\mathbf{U}^{\mathrm{off}} = \mathbf{Q}_x\lambda^{-1},
		\end{equation}
		respectively. Taking the off-diagonal part of \eqref{Eq-prop1-1} and using \eqref{Eq-U-diag-off} leads to
		\begin{equation}\label{Eq-u-travelling-2}
			4u_t + s_0u_x + \mathrm{i}s_1u = 0.
		\end{equation}
		Multiplying \eqref{Eq-u-travelling-2} by $u^*$ and adding its complex conjugate yields exactly \eqref{Eq-nu-travelling-2}.
		Combining \eqref{Eq-s2s3s4}--\eqref{Eq-P-ansatz} with \eqref{Eq-nu-mu}, we arrive at
		\begin{equation}\label{Eq-mu-plus-mubar}
			|\mu|^2 = \frac{s_0(4\nu + s_1) - s_3}{4\nu}, \qquad
			\mu + \mu^* = \frac{\frac14(4\nu + s_1)^2 - s_2 + 2s_0}{4\nu}.
		\end{equation}
		Thus $\mu$ and $\mu^*$ are the two zeros of a quadratic polynomial. Consequently, we may set without loss of generality that 
		\begin{equation}\label{Eq-mu-nu}
			\mu = -\frac{1}{8\nu}\bigl( -\tfrac14(4\nu+s_1)^2 + s_2 - 2s_0 + \mathrm{i}\sqrt{-R(\nu)} \bigr),
		\end{equation}
		by solving the quadratic equation, where $R(\nu)$ is precisely the polynomial defined in \eqref{Eq-R}. On the other hand, from \eqref{Eq-nu-mu} we obtain
		\begin{equation}\label{Eq-Im-mu}
			\mu - \mu^* = \frac{\mathrm{i}s_0}{2}\,\frac{\nu_\xi}{\nu}.
		\end{equation}
		Equating the imaginary part extracted from \eqref{Eq-mu-nu} with the right-hand side of \eqref{Eq-Im-mu} gives \eqref{Eq-nu-x}, which completes the proof.
	\end{proof}
	
	Proposition~\ref{Prop-ue} reduces the dynamics of $\nu$ to a quartic elliptic equation, so the admissible profiles are governed entirely by the location of the zeros of $P(\lambda)$ and $R(\nu)$. We make this correspondence precise in the next proposition, which expresses the zeros $\nu_i$ of $R(\nu)$ directly in terms of the zeros $\pm\lambda^{(i)}$ of $P(\lambda)$.
	Denote by $\pm\lambda^{(i)}$ and $\nu_i, i=1,2,3,4$ the zeros of $P(\lambda)$ and $R(\nu)$, respectively. These values are related as follows.
	\begin{prop}\label{Prop-connection-P-R}
		The zeros of $P(\lambda)$ and $R(\nu)$ are related via
		\begin{equation}\label{Eq-connection-P-R+}
			\begin{aligned}
				\nu_1 &= -\frac{1}{4}(\lambda^{(1)} + \lambda^{(2)} + \lambda^{(3)} - \lambda^{(4)})^2, \quad
				\nu_2 = -\frac{1}{4}(\lambda^{(1)} + \lambda^{(2)} - \lambda^{(3)} + \lambda^{(4)})^2, \\
				\nu_3 &= -\frac{1}{4}(\lambda^{(1)} - \lambda^{(2)} + \lambda^{(3)} + \lambda^{(4)})^2, \quad
				\nu_4 = -\frac{1}{4}(-\lambda^{(1)} + \lambda^{(2)} + \lambda^{(3)} + \lambda^{(4)})^2,
			\end{aligned}
		\end{equation}
		for $s_0 = -\prod_{i=1}^4 \lambda^{(i)}$, and
		\begin{equation}\label{Eq-connection-P-R-}
			\begin{aligned}
				\nu_1 &= -\frac{1}{4}(\lambda^{(1)} + \lambda^{(2)} + \lambda^{(3)} + \lambda^{(4)})^2, \quad
				\nu_2 = -\frac{1}{4}(\lambda^{(1)} + \lambda^{(2)} - \lambda^{(3)} - \lambda^{(4)})^2, \\
				\nu_3 &= -\frac{1}{4}(\lambda^{(1)} - \lambda^{(2)} + \lambda^{(3)} - \lambda^{(4)})^2, \quad
				\nu_4 = -\frac{1}{4}(\lambda^{(1)} - \lambda^{(2)} - \lambda^{(3)} + \lambda^{(4)})^2,
			\end{aligned}
		\end{equation}
		for $s_0 = \prod_{i=1}^4 \lambda^{(i)}$.
	\end{prop}
	\begin{proof}
		When $\nu_i$ is a zero of $R(\nu)$, combining \eqref{Eq-Lax-matrix}--\eqref{Eq-fgh} and \eqref{Eq-P-ansatz} with \eqref{Eq-mu-nu}, we can factor the polynomial $P(\lambda)$ as
		\begin{equation}
			\begin{split}
				P(\lambda)=\big(\lambda^2 -2\nu_i-\frac{1}{2}s_1+s_0\lambda^{-2}\big)^2+4\lambda^{-2}\nu_i(\lambda^2-\mu_i)^2=\lambda^{-4}P_i(\lambda)P_i(-\lambda),\quad i=1,2,3,4.
			\end{split}
		\end{equation}
		where 
		\begin{equation}
			\begin{split}
				\mu_i= -\frac{1}{8\nu_i}\big(-\frac{1}{4}(4\nu_i+s_1)^2+s_2-2s_0\big),\quad 
				P_i(\lambda)=\lambda^4+2\mathrm{i}\sqrt{\nu_i}\lambda^3-\frac{s_1+4\nu_i}{2}\lambda^2-2\mathrm{i}\mu_i\sqrt{\nu_i}\lambda+s_0.
			\end{split}
		\end{equation}
		Applying Vieta's formulas under the constraint $s_0=-\prod_{i=1}^4\lambda^{(i)}$, we obtain \eqref{Eq-connection-P-R+} from the coefficients of $\lambda^3$ and $\lambda^0$ in $P_i(\lambda)$. Repeating the same argument gives \eqref{Eq-connection-P-R-}, which completes the proof.
	\end{proof}
	
	To ensure that \eqref{Eq-nu-x} admits non-degenerate, bounded, and non-negative solutions, certain constraints must be imposed on $\nu_i$ (and hence on $\lambda^{(i)}$). Based on Proposition \ref{Prop-connection-P-R}, we identify three admissible cases:
	
	\begin{itemize}
		\item[-] \textbf{Type A:} All $\lambda^{(i)}$ are purely imaginary, i.e. $\lambda^{(i)}=\mathrm{i}\tau_i$ with $\tau_1>\tau_2>\tau_3>\tau_4\geq0$. Then \eqref{Eq-connection-P-R-} yields four distinct real zeros of $R(\nu)$ satisfying $\nu_1>\nu_2>\nu_3>\nu_4\geq0$.
		
		\item[-] \textbf{Type B:} The parameters form conjugate pairs: $\lambda^{(1)} = -\lambda^{(2)*} = a-b\mathrm{i}$, $\lambda^{(3)} = -\lambda^{(4)*} = -c-d\mathrm{i}$ with $a,b,c,d>0$. Then \eqref{Eq-connection-P-R-} gives $\nu_1>\nu_2\geq0\geq\nu_3>\nu_4$.
		
		\item[-] \textbf{Type C:} Here $\lambda^{(1)} = -\lambda^{(2)*}=a+\mathrm{i}b$, $\lambda^{(3)}=\mathrm{i}c$, $\lambda^{(4)}=\mathrm{i}d$ with $a,b,c>0$, $d\geq0$ and $c>d$. Then \eqref{Eq-connection-P-R-} produces two real zeros $\nu_1>\nu_2\geq0$ and a complex conjugate pair $\nu_3=\nu_4^*$.
	\end{itemize}
	
	\begin{remark}
		For every admissible quadruple \( (\nu_1,\nu_2,\nu_3,\nu_4) \) generated by \eqref{Eq-connection-P-R-}, the same quadruple can also be obtained from \eqref{Eq-connection-P-R+} by changing the sign of \( \lambda^{(4)} \), which simultaneously reverses the sign of \( s_0 \). Consequently, the resulting elliptic function solution \( u \) given by \eqref{Eq-FL-elliptic-solution} remains invariant. Henceforth, we restrict our attention to \eqref{Eq-connection-P-R-}.
	\end{remark}

	Squaring both sides of equation \eqref{Eq-nu-x} yields a fourth-order elliptic equation. Consequently, real solutions exist if and only if \(R(\nu) \leq 0\), which restricts the oscillation of \(\nu(\xi)\) to either the interval \([\nu_2, \nu_1]\) or \([\nu_4, \nu_3]\). These two cases are essentially equivalent, so without loss of generality we restrict attention to oscillations within \([\nu_2, \nu_1]\). By the existence and uniqueness theorem for ordinary differential equations, the solution \(\nu(\xi)\) is uniquely determined once its initial value \(\nu_0 := \nu(0)\) is prescribed. Any other solution in \([\nu_2, \nu_1]\) can be obtained from the one with \(\nu(0)=\nu_1\) by a translation in \(\xi\). Hence we may further take \(\nu_0 = \nu_1\) without loss of generality.
	
	To express $\nu(\xi)$ in terms of standard Weierstrass functions, we now reduce the quartic curve underlying \eqref{Eq-nu-x} to its Weierstrass normal form. Consider the elliptic curve
	\begin{align}\label{Eq-K1}
		\mathcal{K}_1:=\{(\Lambda_1,Y_1) \mid Y_1^2 = R(\Lambda_1)\},
	\end{align}
	which arises naturally from equation \eqref{Eq-nu-x}. Following the classical procedure of reducing a quartic elliptic curve to a cubic one, we introduce the birational map
	\begin{equation}\label{Eq-birational-mapping}
		\big(\Lambda_1,	Y_1\big)=\left(\nu_0+ \frac{1}{c_1\Lambda+c_2},\frac{\mathrm{i}c_1Y}{(c_1\Lambda+c_2)^2}\right),\quad 	c_1=-\frac{4}{R'(\nu_0)},\quad c_2=-\frac{R''(\nu_0)}{6R'(\nu_0)},
	\end{equation}
	which maps $\mathcal{K}_1$ to the normalized Weierstrass elliptic curve
	\begin{align}\label{Eq-K}
		\mathcal{K}:=\{(\Lambda,Y)|Y^2=4\left(\Lambda-e_1\right)\left(\Lambda-e_2\right)\left(\Lambda-e_3\right)\},
	\end{align}
	with 
	\begin{equation}\label{Eq-ei}
		e_i=\begin{cases}
			\displaystyle-\frac{1}{24}R''(\nu_0)+\frac{R'(\nu_0)}{4(\nu_0-\nu_{i+1})},\quad \nu_0=\nu_1,\\[10pt]
			\displaystyle-\frac{1}{24}R''(\nu_0)+\frac{R'(\nu_0)}{4(\nu_0-\nu_{4-i})},\quad \nu_0=\nu_4.
		\end{cases}
	\end{equation}
	The normalized elliptic curve \eqref{Eq-K} can be parameterized as 
	\begin{equation}\label{Eq-parameterization-K}
		(\Lambda,Y)=\big(\wp(\xi),\wp'(\xi)\big).
	\end{equation}
	Combining \eqref{Eq-birational-mapping} with \eqref{Eq-parameterization-K} shows that the squared modulus admits the representation
	\begin{equation}\label{Eq-nu-parameterization}
		\nu(\xi)=\nu_0\frac{\wp(\xi)-\wp(\rho)}{\wp(\xi)-\wp(\kappa)},
	\end{equation}
	where the parameters $\kappa$ and $\rho$ are determined by
	\begin{equation}\label{Eq-kappa-rho}
		\begin{split}
			\wp(\kappa) &= -\frac{1}{24} R''(\nu_0), \quad 
			\wp(\rho)   = -\frac{1}{24} R''(\nu_0) + \frac{1}{4\nu_0} R'(\nu_0).
		\end{split}
	\end{equation}
	
	Each period parallelogram contains two candidates for $\kappa$ and two for $\rho$ satisfying \eqref{Eq-kappa-rho}. We now fix them uniquely. From \eqref{Eq-ei} and \eqref{Eq-kappa-rho}, one deduces that
	\begin{equation}\label{Eq-kappa-rho-prime}
		(\wp'(\kappa))^2 = -(R'(\nu_0))^2,\quad (\wp'(\rho))^2 = -\frac{\nu_1\nu_2\nu_3\nu_4 (R'(\nu_0))^2}{\nu_0^4},
	\end{equation}
	which in turn yields
	\begin{equation}
		\nu_1\nu_2\nu_3\nu_4 = \frac{\nu_0^4 (\wp'(\rho))^2}{(\wp'(\kappa))^2}.
	\end{equation}
	Combining \eqref{Eq-kappa-rho} with \eqref{Eq-kappa-rho-prime} leads to
	\begin{equation}\label{Eq-nu0}
		\nu_0 = \frac{\mathrm{i}\wp'(\kappa)}{4(\wp(\rho)-\wp(\kappa))}.
	\end{equation}
	Using \eqref{Eq-R}, \eqref{Eq-kappa-rho-prime} and \eqref{Eq-nu0}, we arrive at
	\begin{equation}\label{Eq-Csqr}
		\Big(\frac{s_1^2}{4}-s_2+2s_0\Big)^2 = 16\nu_1\nu_2\nu_3\nu_4 = \frac{(\wp'(\rho))^2 (\wp'(\kappa))^2}{16(\wp(\rho)-\wp(\kappa))^4}.
	\end{equation}
	Therefore, imposing the conditions
	\begin{equation}\label{Eq-restriction}
		\mathrm{i}\wp'(\kappa) > 0,\quad \frac{s_1^2}{4}-s_2+2s_0 = -\frac{\wp'(\rho)\wp'(\kappa)}{4(\wp(\rho)-\wp(\kappa))^2},
	\end{equation}
	fixes the sign and completes the characterization.
	
	Without loss of generality, we place $\kappa$ and $\rho$ inside the fundamental period parallelogram. For Type A and Type B solutions, this parallelogram is taken as the rectangle with vertices $\pm\omega_1\pm\omega_3$. For Type C solutions, it is the rhombus with vertices $\pm\omega_1$ and $\pm(\omega_1-\omega_3)$. By the monotonicity of the Weierstrass $\wp$-function, the following characterizations hold. For Type A and Type B solutions, one has for $m\in\mathbb{Z}$ that $\wp(z)\in(-\infty,e_3]$ if and only if $z = 2m\omega_1 + \mathrm{i}\mathbb{R}$, and $\wp(z)\in[e_2,e_1]$ if and only if $z = (2m+1)\omega_1 + \mathrm{i}\mathbb{R}$. For Type C solutions, one has $\wp(z)\in(-\infty,e_1]$ if and only if $z = m\omega_1 + \mathrm{i}\mathbb{R}$. A comparison between \eqref{Eq-ei} and \eqref{Eq-kappa-rho} shows that $\wp(\kappa)\in(-\infty,e_3]$ for all three solution types. Moreover, $\wp(\rho)\in(-\infty,e_3]$ for Type A and Type C solutions, while $\wp(\rho)\in[e_2,e_1]$ for Type B solutions. Consequently, $\mathrm{Re}(\kappa)=0$ holds for all three types. In addition, $\mathrm{Re}(\rho)=0$ for Type A and Type C solutions, whereas $\mathrm{Re}(\rho)=\omega_1$ for Type B solutions.
	
	With the parameters $\kappa$ and $\rho$ now pinned down, we are in a position to express the elliptic function solution itself in closed form. We do so in two stages: the auxiliary spectrum \(\mu\) is first written in terms of the Weierstrass $\zeta$-function, and the elliptic function solution \(u\) is then recovered through the relation \eqref{Eq-nu-mu}. These two steps are carried out in the following two propositions.

	\begin{prop}\label{Prop-mu}
		The auxiliary spectrum $\mu(\xi)$ can be written in terms of the Weierstrass $\zeta$-function as 
		\begin{align}\label{Eq-mu-parameterization}
			\mu(\xi)=\frac{\mathrm{i}}{4}\big(\zeta(\kappa+\xi) - \zeta(\xi+\rho) + \zeta(\rho+\kappa) - \zeta(2\kappa)\big).	
		\end{align}
	\end{prop}
	\begin{proof}
		Taking the derivative of $\nu(\xi)$ with respect to $\xi$, we obtain 
		\begin{equation}\label{Eq-nu-derivative}
			\nu'(\xi)\xlongequal{\eqref{Eq-nu-parameterization}}\frac{\nu_0\wp'(\xi)\big(\wp(\rho)-\wp(\kappa)\big)}{\big(\wp(\xi)-\wp(\kappa)\big)^2}\xlongequal{\eqref{Eq-nu0}}\frac{\mathrm{i}\wp'(\kappa)\wp'(\xi)}{4\big(\wp(\xi)-\wp(\kappa)\big)^2}.
		\end{equation}
		It then follows that
		\begin{equation}\label{Eq-nu-derivative-rho}
			\nu'(\rho)\xlongequal{\eqref{Eq-nu-derivative}}\frac{\mathrm{i}\wp'(\kappa)\wp'(\rho)}{4\big(\wp(\rho)-\wp(\kappa)\big)^2}\xlongequal{\eqref{Eq-restriction}}-\mathrm{i}\Big(	\frac{s_1^2}{4}-s_2+2s_0\Big).
		\end{equation}
		Consequently, the auxiliary spectrum can be expressed as 
		\begin{equation}\label{Eq-mu-rewrite}
			\begin{split}
				\mu(\xi)&\xlongequal[\eqref{Eq-nu-derivative-rho}]{\eqref{Eq-mu-nu}}
				\frac{\mathrm{i}\big(\nu'(\rho)-\nu'(\xi)\big)+2s_1\nu(\xi)+4\nu^2(\xi)}{8\nu(\xi)}.\\
			\end{split}
		\end{equation}
		From \eqref{Eq-nu-parameterization}, we see that $\nu(\xi)$ is a second-order elliptic function in $\xi$, with simple zeros at $\pm\rho$ and simple poles at $\pm\kappa$. The derivative $\nu'(\xi)$ is a fourth-order elliptic function, since \eqref{Eq-nu-derivative} shows that $\nu'(\xi)$ possesses simple zeros at $0$ and $\omega_i,\ i=1,2,3$, and second-order poles at $\pm\kappa$.
		
		According to \eqref{Eq-mu-rewrite}, the only possible poles of $\mu(\xi)$ are $\pm\kappa$ and $\pm\rho$. However, taking the limit
		\begin{equation}
			\lim_{\xi\rightarrow \rho}\mu(\xi)\xlongequal[\eqref{Eq-mu-rewrite}]{\eqref{Eq-nu-parameterization}}
			\frac{-\mathrm{i}(\xi-\rho)\nu''(\rho)}{8\nu_0\frac{\wp(\xi)-\wp(\rho)}{\wp(\xi)-\wp(\kappa)}}+\frac{1}{4}s_1
			=\frac{-\mathrm{i}\nu''(\rho)}{8\nu_0\frac{\wp'(\rho)}{\wp(\rho)-\wp(\kappa)}}+\frac{1}{4}s_1,
		\end{equation}
		reveals that $\rho$ is not a pole. Similarly, one finds that
		\begin{equation}\label{Eq-limit-1}
			\lim_{\xi\rightarrow \kappa} \frac{\mathrm{i}\nu'(\xi)}{\nu^2(\xi)}\xlongequal[\eqref{Eq-nu-derivative}]{\eqref{Eq-nu-parameterization}}4,
		\end{equation}
		and consequently,
		\begin{equation}\label{Eq-mu-zero-kappa}
			\lim_{\xi\rightarrow \kappa} \mu(\xi)\xlongequal[\eqref{Eq-mu-rewrite}]{\eqref{Eq-nu-parameterization}}
			\lim_{\xi\rightarrow \kappa}\frac{-\mathrm{i}\nu'(\xi)+4\nu^2(\xi)}{8\nu(\xi)}
			=\lim_{\xi\rightarrow \kappa}\frac{\frac{-\mathrm{i}\nu'(\xi)}{\nu^2(\xi)}+4}{8\nu^3(\xi)}\xlongequal{\eqref{Eq-limit-1}}0.
		\end{equation}
		Hence, $\kappa$ is a zero of $\mu(\xi)$. Moreover, from the limits
		\begin{equation}\label{Eq-mu-limit}
			\begin{split}
				&\lim_{\xi\rightarrow -\rho} \mu(\xi)\xlongequal[\eqref{Eq-mu-rewrite}]{\eqref{Eq-nu-parameterization}}
				\frac{1}{4}s_1+\lim_{\xi\rightarrow -\rho} \frac{\mathrm{i}\big(\nu'(\rho)-\nu'(\xi)\big)}{8\nu(\xi)},\\
				&\lim_{\xi\rightarrow -\kappa} \mu(\xi)\xlongequal[\eqref{Eq-mu-rewrite},\eqref{Eq-limit-1}]{\eqref{Eq-nu-parameterization}}
				\frac{1}{4}s_1+\lim_{\xi\rightarrow -\kappa} \nu(\xi),
			\end{split}
		\end{equation}
		it follows that both $-\kappa$ and $-\rho$ are poles of $\mu(\xi)$. Consequently, $\mu(\xi)$ is an elliptic function of order two. By the fundamental property of elliptic functions, the sum of its zeros equals the sum of its poles within each period parallelogram. This locates its remaining zero at $-2\kappa-\rho$. Following the general construction of elliptic functions, we may therefore write
		\begin{equation}\label{Eq-mu-C1}
			\mu(\xi)=C_1\big(\zeta(\kappa+\xi) - \zeta(\xi+\rho) + \zeta(\rho+\kappa) - \zeta(2\kappa)\big),
		\end{equation}
		where $C_1\in\mathbb{C}$ is a constant to be determined.
		Finally, $C_1$ is obtained by evaluating the residue at $\xi = -\kappa$. Specifically,
		\begin{equation}\label{Eq-Res-1}
			\begin{split}
				&\quad\quad\underset{\xi=-\kappa}{\mathrm{Res}}\,\mu(\xi)
				\xlongequal[\eqref{Eq-mu-rewrite}]{\eqref{Eq-nu-parameterization}} 
				\underset{\xi=-\kappa}{\mathrm{lim}}\mu(\xi)\sigma(\xi+\kappa)\xlongequal{\eqref{Eq-mu-limit}}
				\underset{\xi=-\kappa}{\mathrm{lim}}
				\nu(\xi)\sigma(\xi+\kappa)\\
				&\xlongequal[\eqref{Eq-wp-sigma-diff-theta}]{\eqref{Eq-nu-parameterization}\eqref{Eq-nu0}}	\underset{\xi=-\kappa}{\mathrm{lim}}
				\frac{\mathrm{i}\wp'(\kappa)}{4(\wp(\rho)-\wp(\kappa))}\frac{\wp(\xi)-\wp(\rho)}{-\frac{\sigma(\xi+\kappa)\sigma(\xi-\kappa)}{\sigma^2(\xi)\sigma^2(\kappa)}}=\frac{\mathrm{i}}{4}\frac{\wp'(\kappa)\sigma^4(\kappa)}{\sigma(2\kappa)}\xlongequal{\eqref{Eq-half-argument-1}}\frac{\mathrm{i}}{4}.
			\end{split}
		\end{equation}
		Comparing \eqref{Eq-mu-C1} with \eqref{Eq-Res-1} yields $C_1=\frac{\mathrm{i}}{4}$, which completes the proof of the proposition.
	\end{proof}
	
	Integrating the logarithmic derivative supplied by Proposition~\ref{Prop-mu} now yields the elliptic function solution $u$ in closed form.
	
	\begin{prop}
		The elliptic function solution of the FL equation \eqref{Eq-FL} can be expressed as 
		\begin{equation}\label{Eq-FL-elliptic-solution}
			\begin{split}
				u(\xi,t)=\frac{\sqrt{\nu_0}\sigma(\kappa)\sigma(\xi+\rho)}{\sigma(\rho)\sigma(\xi+\kappa)}e^{F(\xi,t)},
			\end{split}
		\end{equation}
		where the function $F(\xi,t)$ is defined by
		\begin{equation}\label{Eq-F} 
			F(\xi,t)=\big(-\zeta(\rho+\kappa)+\zeta(2\kappa)\big)\xi-\frac{\mathrm{i}s_1}{4}t.
		\end{equation}
	\end{prop}
	
	\begin{proof}
		From the second equation in \eqref{Eq-nu-mu}, the logarithmic derivative of $u$ is seen to be an elliptic function of $\xi$. Equation \eqref{Eq-mu-parameterization} shows that $-\kappa$ and $-\rho$ are simple poles of $\mu(\xi)$. Hence, the zeros and poles of $u$ are all located among these points. Furthermore, it is evident from \eqref{Eq-mu-parameterization} that 
		\begin{equation}
			\underset{\xi = -\kappa}{\mathrm{Res}}\,4\mathrm{i}\mu(\xi)=-1,\quad 		\underset{\xi = -\rho}{\mathrm{Res}}\,4\mathrm{i}\mu(\xi)=1,
		\end{equation}
		which implies that $-\kappa$ and $-\rho$ are, respectively, a first-order pole and a first-order zero of $u$. Therefore, we have 
		\begin{equation}
			u(\xi,0)=C_1\frac{\sigma(\xi+\rho)}{\sigma(\xi+\kappa)}e^{C_2\xi},
		\end{equation}
		where $C_1,C_2\in\mathbb{C}$ are constants to be determined. Combining the first equation in \eqref{Eq-nu-mu}, \eqref{Eq-nu-parameterization} and \eqref{Eq-nu0}, we may set 
		\begin{equation}
			C_1=\sqrt{\nu_0}\frac{\sigma(\kappa)}{\sigma(\rho)}.
		\end{equation}
		The constant $C_2$ corresponds precisely to the constant term of $\mu(\xi)$. From \eqref{Eq-mu-parameterization} we find that \begin{equation}
			C_2 = -\zeta(\rho + \kappa) + \zeta(2\kappa).
		\end{equation} 
		Consequently,
		\begin{equation}
			u(\xi,0)=\frac{\sqrt{\nu_0}\sigma(\kappa)\sigma(\xi+\rho)}{\sigma(\rho)\sigma(\xi+\kappa)}e^{(-\zeta(\rho+\kappa)+\zeta(2\kappa))\xi}.
		\end{equation}
		The time dependence of $u$ follows from \eqref{Eq-u-travelling-2}. This completes the proof.
	\end{proof}

	\subsection{The fundamental solution matrix of the Lax pair}
	Having obtained the elliptic seed solution, we now turn to the second objective of this section: the explicit construction of the fundamental solution matrix of the Lax pair \eqref{Eq-Lax pair FL} associated with it. The key step is to parameterize the spectral curve $\mathcal{K}_2$ below by the same uniform parameter $z$ used above, after which the two scalar entries of the eigenvector are obtained in closed form through Weierstrass functions. We begin with the parameterization of the elliptic curve:
	\begin{align} \mathcal{K}_2:=\{(\lambda,y)|y^2=P(\lambda)\}.
	\end{align}
	From \eqref{Eq-Lax equation FL} and \eqref{Eq-fgh}, we deduce that
	\begin{equation}\label{Eq-L12-derivative}
		L_{12,\xi}=-2\mathrm{i}\lambda^{-2}L_{12}+2\mathrm{i}\lambda^{-1}u_\xi L_{11}.
	\end{equation}
	Substituting $\lambda=\sqrt{\mu}$ into equation \eqref{Eq-L12-derivative} yields
	\begin{equation}\label{Eq-L12-xi-1}
		L_{12,\xi}\big|_{\lambda=\sqrt{\mu}}=-2\mathrm{i}\mu^{-1}L_{12}\big|_{\lambda=\sqrt{\mu}}+2\mathrm{i}\mu^{-\frac{1}{2}}u_\xi L_{11}\big|_{\lambda=\sqrt{\mu}}\xlongequal{\eqref{Eq-fgh}}2\mathrm{i}\mu^{-\frac{1}{2}}u_\xi L_{11}\big|_{\lambda=\sqrt{\mu}}.
	\end{equation}
	On the other hand, differentiating \( L_{12} \) directly yields
	\begin{equation}\label{Eq-L12-xi-2}
		L_{12,\xi}\big|_{\lambda=\sqrt{\mu}}\xlongequal{\eqref{Eq-fgh}}-2\mathrm{i}u\mu^{-\frac{1}{2}}\mu_\xi.
	\end{equation} 
	A comparison between \eqref{Eq-L12-xi-1} and \eqref{Eq-L12-xi-2} yields
	\begin{equation}\label{Eq-L11}
		L_{11}\big|_{\lambda=\sqrt{\mu}}=-\frac{u\mu_\xi}{u_\xi}.
	\end{equation}
	From \eqref{Eq-fgh} and \eqref{Eq-L11}, we obtain
	\begin{equation}\label{Eq-parameterization-lambda-y--1}
		P^2(\sqrt{\mu})\xlongequal{\eqref{Eq-fgh}}L_{11}^2\big|_{\lambda=\sqrt{\mu}}-L_{12}\big|_{\lambda=\sqrt{\mu}}L_{21}\big|_{\lambda=\sqrt{\mu}}\xlongequal{\eqref{Eq-fgh}}L_{11}^2\big|_{\lambda=\sqrt{\mu}}\xlongequal{\eqref{Eq-L11}}\frac{u^2\mu_\xi^2}{u_\xi^2}\xlongequal{\eqref{Eq-nu-mu}}-\frac{\mu_\xi^2}{16\mu^2}.
	\end{equation}
	Therefore, the elliptic curve $\mathcal{K}_2$ admits a parameterization of the form
	\begin{equation}\label{Eq-parameterization-lambda-y-0}
		\begin{split}
			&\quad\quad \big(\lambda(z),y(z)\big)=\Big(\sqrt{\mu(z)},\frac{\mathrm{i}\mu'(z)}{4\mu(z)}\Big)\\
			&\xlongequal{\eqref{Eq-mu-parameterization}}\left(\frac{\sigma(z-\kappa)\sigma(\rho)}{4\sqrt{\nu_0}\sigma(\rho+z)\sigma(\rho+\kappa)\sigma(\kappa)}\frac{d_0(z)}{d_0(\hat{z})},\frac{\mathrm{i}}{4}\frac{\sigma(2z+\rho+\kappa)\sigma(\rho+\kappa)\sigma(2\kappa)}{\sigma(z+\rho)\sigma(z+\kappa)\sigma(2\kappa+\rho+z)\sigma(z-\kappa)}\right),
		\end{split}
	\end{equation}
	where 
	\begin{equation}\label{Eq-nu0-d0}
		d_0(z)=\sqrt{\sigma(\rho+z)\sigma(z+2\kappa+\rho)},\quad \hat{z}=-\kappa-\rho-z.
	\end{equation}
	A direct consequence of \eqref{Eq-parameterization-lambda-y-0}--\eqref{Eq-nu0-d0} is that
	\begin{equation}\label{Eq-shift}
		\lambda(z)=\lambda(\hat{z}),\quad y(z)=-y(\hat{z}),
	\end{equation}
	indicating that the transformation $z \mapsto \hat{z}$ keeps $\lambda$ fixed while inverting the sign of $y$; this involution corresponds precisely to the hyperelliptic involution of $\mathcal{K}_2$.

	Clearly, $\pm \mathrm{i}y(z)$ are the two eigenvalues of $\mathbf{L}(\xi,t;z)$. Let $\big(1, r_\pm(\xi,t;z)\big)^{T}$ denote the eigenvectors corresponding to $\pm\mathrm{i}y(z)\mathbb{I}_2 - \mathbf{L}(\xi,t;z)$.
	In the $(\xi,t)$-coordinate, the Lax pair \eqref{Eq-Lax pair FL} is converted into 
	\begin{equation}\label{Eq-Lax-xi-version}
		\bm{\Psi}_\xi(\xi,t;\lambda) = -2s_0\mathbf{U}(\mathbf{Q};\lambda)\bm{\Psi}(\xi,t;\lambda),\quad \bm{\Psi}_t(\xi,t;\lambda)=\Big(\mathbf{V}(\mathbf{Q};\lambda)+\frac{1}{4}s_0\mathbf{U}(\mathbf{Q};\lambda)\Big)\bm{\Psi}.
	\end{equation}
	By the uniqueness of solutions to ordinary differential equations, we obtain a pair of linearly independent vector solutions $(\phi_{\pm}(\xi,t;z),\psi_{\pm}(\xi,t;z))^{T}$ of the Lax pair \eqref{Eq-Lax-xi-version} satisfying
	\begin{align}\label{Eq-rpm}
		r_{\pm}(\xi,t;z)=\frac{\psi_{\pm}(\xi,t;z)}{\phi_{\pm}(\xi,t;z)}=\frac{\mathrm{i}\big(L_{11}(\xi,t;z)\pm y(z)\big)}{L_{12}(\xi,t;z)}=\frac{\mathrm{i}L_{21}(\xi,t;z)}{L_{11}(\xi,t;z)\mp y(z)}.
	\end{align}
	With these preparations in place, we now establish two lemmas that supply, respectively, the first component $\phi_{+}$ and the ratio $r_{+}$ of the eigenvector in closed form; together they furnish all entries of the fundamental solution matrix.
	
	\begin{lemma}\label{Prop-psi+}
		The function $\phi_{+}(\xi,t;z)$ admits the explicit representation
		\begin{align}\label{Eq-phi-p}
			\phi_+(\xi,t;z)=\frac{\sigma(\xi-z)}{\sigma(\xi+\kappa)}E(\xi,t;z),
		\end{align}
		where 
		\begin{equation}\label{Eq-E}
			\begin{split}
				&E(\xi,t;z):=\exp\Big(\frac{1}{2}\Big(\zeta(2\kappa+\rho+z)-\zeta(\kappa+\rho)+\zeta(2\kappa)-\zeta(\kappa-z)\Big)\xi+\frac{\mathrm{i}}{4}\Big(y(z)-\frac{1}{2}s_1\Big)t\Big).
			\end{split}
		\end{equation}
		
		The corresponding function $\phi_{-}$ is obtained by replacing $z$ with $\hat{z}$ in \eqref{Eq-phi-p}.
	\end{lemma}
	
	\begin{proof}
		The proof proceeds by analyzing the logarithmic derivative of $\phi_+$ with respect to $\xi$, which turns out to be elliptic.  
		From the spatial part of the Lax pair \eqref{Eq-Lax-xi-version} together with \eqref{Eq-fgh}, \eqref{Eq-nu-mu}, \eqref{Eq-parameterization-lambda-y-0} and \eqref{Eq-rpm}, we obtain
		\begin{equation}\label{Eq-spatial}
			\begin{split}
				\phi_{+,\xi}\xlongequal[\eqref{Eq-parameterization-lambda-y-0}\eqref{Eq-Lax-xi-version}\eqref{Eq-rpm}]{\eqref{Eq-Lax pair FL}\eqref{Eq-fgh}\eqref{Eq-nu-mu}\eqref{Eq-parameterization-lambda-y--1}}\Big(2\mathrm{i}s_0\lambda^{-2}(z)-\frac{4\mathrm{i}\lambda^{-2}(z)\mu(\xi)\nu(\xi)(\lambda^2(z)-\mu^*(\xi))}{\nu(z)-\nu(\xi)}\Big)\phi_+,
			\end{split}
		\end{equation}
		which implies that $\big(\ln\phi_{+}(\xi,0;z)\big)_\xi$ is an elliptic function in $\xi$. Consequently, $\phi_{+}(\xi,0;z)$ is multiplicatively periodic and can be expressed as a product of a constant multiplier, a meromorphic factor, and an exponential factor.
		
		Since equations \eqref{Eq-Lax-xi-version} and \eqref{Eq-spatial} are linear, we may normalize the constant multiplier to $1$ without loss of generality. The zeros and poles of the meromorphic part are then encoded in the poles of $\big(\ln\phi_{+}(\xi,0;z)\big)_\xi$, while the exponential factor is determined by its constant term.
		
		From \eqref{Eq-mu-parameterization}, $\mu(\xi)$ has simple zeros at $\xi = \kappa$ and $\xi = -2\kappa - \rho$, and simple poles at $\xi = -\kappa$ and $\xi = -\rho$. Similarly, \eqref{Eq-nu-parameterization} shows that $\nu(\xi)$ has simple zeros at $\xi = \pm \rho$ and simple poles at $\xi = \pm \kappa$. Moreover, the symmetry $\mu^*(\xi) = \mu(-\xi)$ follows directly from \eqref{Eq-mu-parameterization}. Using \eqref{Eq-mu-parameterization} together with \eqref{Eq-shift}, one finds that $\lambda - \mu^*(\xi)$ has first-order zeros at $\xi = -z$ and $\xi = -\hat{z}$, and first-order poles at $\xi = \kappa$ and $\xi = \rho$. Meanwhile, \eqref{Eq-nu-derivative} indicates that $\nu(z) - \nu(\xi)$ has simple zeros at $\xi = \pm \rho$ and simple poles at $\xi = \pm \kappa$.
		
		Combining the above information and following the general construction principles of elliptic functions, we conclude that 
		\begin{equation}\label{Eq-C2-meromorphic}
			\frac{4\mathrm{i}\lambda^{-2}(z)\mu(\xi)\nu(\xi)(\lambda^2(z)-\mu^*(\xi))}{\nu(z)-\nu(\xi)}=C_2\big(\zeta(\kappa+\rho)+\zeta(\xi+\kappa)+\zeta(-2\kappa-\rho-z)-\zeta(\xi-z)\big).
		\end{equation}
		Comparing residues at $\xi=-\kappa$ in the above equation yields
		\begin{equation}\label{Eq-C2}
			C_2=4\mathrm{i}\lim_{\xi\rightarrow -\kappa}(1-\mu(-\xi)\lambda^{-2}(z))\lim_{\xi\rightarrow -\kappa}\frac{\nu(\xi)}{\nu(z)-\nu(\xi)}\operatorname*{Res}_{\xi=-\kappa}\mu(\xi)\xlongequal[\eqref{Eq-Res-1}]{\eqref{Eq-mu-parameterization}\eqref{Eq-mu-zero-kappa}} 1.
		\end{equation}
		It then follows from \eqref{Eq-spatial}--\eqref{Eq-C2} that the only poles of $\ln\big(\phi_{+}(\xi,0;z)\big)_\xi$ are simple ones at $\xi=z$ and $\xi=-\kappa$. Moreover, one has
		\begin{equation}
			\operatorname*{Res}_{\xi=-\kappa} \ln\big(\phi_{+}(\xi,0;z)\big)_\xi=-1,\quad \operatorname*{Res}_{\xi=z} \ln\big(\phi_{+}(\xi,0;z)\big)_\xi=1,
		\end{equation}
		which implies that $\xi=-\kappa$ and $\xi=z$ are, respectively, the only first-order pole and the only first-order zero of $\phi_+(\xi,0;z)$.
		
		The exponential factor of $\phi_+(\xi,0;z)$ is determined by the constant term of $\ln\big(\phi_{+}(\xi,0;z)\big)_\xi$, which evaluates to $\exp\big(\big(\zeta(2\kappa+\rho+z)-\zeta(\rho+\kappa)+2\mathrm{i}s_0\mu^{-1}(z)\big)\xi\big)$. For the time dependence, we compute
		\begin{equation*}
			\begin{split}
				&\phi_{+,t} \xlongequal[\eqref{Eq-rpm}]{\eqref{Eq-Lax-xi-version}, \eqref{Eq-fgh}} \frac{\mathrm{i}}{4}\Big(y(z)-\frac{1}{2}s_1\Big)\phi_+.
			\end{split}
		\end{equation*}
		
		To express $E(\xi,t;z)$ in terms of Weierstrass functions, we apply the addition formulas \eqref{Eq-addition-wp-theta}--\eqref{Eq-addition-zeta-theta} together with \eqref{Eq-kappa-rho}--\eqref{Eq-kappa-rho-prime} and \eqref{Eq-restriction}, obtaining
		\begin{equation}
			s_0=\frac{1}{16}(\wp(\kappa+\rho)-\wp(2\kappa)),
		\end{equation}
		which leads to
		\begin{equation}\label{Eq-2is0mu}
			\begin{split}
				2\mathrm{i}s_0\mu^{-1}(z)&\xlongequal[\eqref{Eq-zeta-sigma-relation-theta}\eqref{Eq-addition-sigma-theta}]{\eqref{Eq-mu-parameterization}}\frac{1}{2}\big(\zeta(\kappa+\rho)+\zeta(2\kappa)+\zeta(-\kappa+z)-\zeta(2\kappa+\rho+z)\big).
			\end{split}
		\end{equation}

		Finally, using  \eqref{Eq-2is0mu} yields the desired representation for $E(\xi,t;z)$. To complete the construction, note that converting $y$ to $-y$ in the expression for $\phi_-(\xi,t;z)$ suffices, as guaranteed by \eqref{Eq-shift}. Moreover, it remains only to replace $z$ with $\hat{z}$ in \eqref{Eq-spatial}. This finishes the proof.
	\end{proof}
	
	With the first component of the eigenvector in hand, we next determine the ratio $r_{+}$, which supplies the second component.
	
	\begin{lemma}\label{Prop-r+}
		The function $r_+(\xi,t;z)$ admits the explicit representation
		\begin{align}\label{Eq-para-r+}
			r_+(\xi,t;z)=\frac{\mathrm{i}\sigma(\hat{z}+\xi)\sigma(\xi+\kappa)d_0(\hat{z})}{\sigma(\xi-z)\sigma(\xi-\kappa)d_0(z)}e^{-F(\xi,t)}.
		\end{align}
		The corresponding function $r_-(\xi,t;z)$ is obtained through the substitution $z \mapsto \hat{z}$.
	\end{lemma}
	
	\begin{proof}
		By employing the expression \eqref{Eq-nu-parameterization} and analyzing its zeros and poles, we obtain
		\begin{equation}\label{Eq-nu-difference}
			\nu(z)-\nu(\xi)=\frac{C_3\sigma(\xi+z)\sigma(\xi-z)}{\sigma(\xi+\kappa)\sigma(\xi-\kappa)}.
		\end{equation}
		Evaluating \eqref{Eq-nu-difference} at \(\xi = 0\) yields
		\begin{equation}
			C_3\frac{\sigma^2(z)}{\sigma^2(\kappa)}=\nu(z)-\nu_0\xlongequal[\eqref{Eq-nu0}]{\eqref{Eq-nu-parameterization}}	\frac{\mathrm{i}\wp'(\kappa)}{4\big(\wp(\kappa)-\wp(z)\big)}\xlongequal[\eqref{Eq-half-argument-1}]{\eqref{Eq-wp-sigma-diff-theta}}\frac{\mathrm{i}\sigma(2\kappa)\sigma^2(z)}{4\sigma^2(\kappa)\sigma(\kappa+z)\sigma(\kappa-z)},
		\end{equation}
		and consequently 
		\begin{equation}\label{Eq-nu-difference-2}
			\nu(z)-\nu(\xi)=\frac{\mathrm{i}\sigma(2\kappa)\sigma(\xi+z)\sigma(\xi-z)}{4\sigma(\kappa+z)\sigma(\kappa-z)\sigma(\xi+\kappa)\sigma(\xi-\kappa)}.
		\end{equation}
		We can also obtain
		\begin{equation}\label{Eq-mu-difference}
			\begin{split}
				&\lambda^2(z)-\mu^*(\xi)\xlongequal[\eqref{Eq-parameterization-lambda-y-0}]{\eqref{Eq-mu-parameterization}}\mu(z)-\mu(-\xi)\xlongequal[\eqref{Eq-wp-sigma-diff-theta}]{\eqref{Eq-mu-parameterization}}\frac{\mathrm{i}}{4}\frac{\sigma(\kappa-\rho)\sigma(\xi+z)\sigma(\xi-\kappa-\rho-z)}{\sigma(\kappa+z)\sigma(-\rho-z)\sigma(-\kappa+\xi)\sigma(\xi-\rho)}.
			\end{split}
		\end{equation}	
		Moreover, employing \eqref{Eq-nu-parameterization} yields
		\begin{equation}\label{Eq-nu-C4}
			\nu(\xi)=C_4\frac{\sigma(\xi+\rho)\sigma(\xi-\rho)}{\sigma(\xi+\kappa)\sigma(\xi-\kappa)},
		\end{equation}
		by analyzing the zeros and poles. Evaluating \eqref{Eq-nu-C4} at \(\xi = 0\) gives
		\begin{equation}\label{Eq-C4}
			C_4=\nu_0\frac{\sigma^2(\kappa)}{\sigma^2(\rho)}\xlongequal[\eqref{Eq-wp-sigma-diff-theta}\eqref{Eq-half-argument-1}]{\eqref{Eq-nu0}}\frac{\mathrm{i}\sigma(2\kappa)}{4\sigma(\rho+\kappa)\sigma(\rho-\kappa)}.
		\end{equation}
		Therefore it follows from \eqref{Eq-nu-C4}--\eqref{Eq-C4} that 
		\begin{equation}\label{Eq-nu-representation-2}
			\nu(\xi)=\frac{\mathrm{i}\sigma(2\kappa)\sigma(\xi+\rho)\sigma(\xi-\rho)}{4\sigma(\rho+\kappa)\sigma(\rho-\kappa)\sigma(\xi-\kappa)\sigma(\xi+\kappa)}.
		\end{equation}
		A comparison between \eqref{Eq-FL-elliptic-solution} and \eqref{Eq-nu-representation-2} yields
		\begin{equation}\label{Eq-uc}
			u^*(\xi,t)=\sqrt{\nu_0}\frac{\sigma(\kappa)\sigma(\xi-\rho)}{\sigma(\rho)\sigma(\xi-\kappa)}e^{-F(\xi,t)}.
		\end{equation}
		From the above analysis, we conclude that
		\begin{equation}
			r_+(\xi,t;z)\xlongequal[\eqref{Eq-rpm}]{\eqref{Eq-fgh}}-\frac{u^*(\xi)\lambda^{-1}(z)(\lambda^2(z)-\mu^*(\xi))}{\nu(z)-\nu(\xi)}\xlongequal[\eqref{Eq-mu-difference}\eqref{Eq-uc}]{\eqref{Eq-parameterization-lambda-y-0}\eqref{Eq-nu-difference-2}}\frac{\mathrm{i}\sigma(-\hat{z}-\xi)\sigma(\xi+\kappa)d_0(\hat{z})}{\sigma(\xi-z)\sigma(-\xi+\kappa)d_0(z)}e^{-F(\xi,t)}.
		\end{equation}
		The representation for \(r_-(\xi,t;z)\) follows directly from \eqref{Eq-rpm} by substituting \(z\) with \(\hat{z}\), which completes the proof.
	\end{proof}
	With the two components now available in closed form, the fundamental solution matrix of the Lax pair \eqref{Eq-Lax pair FL} follows immediately, as summarized in the following theorem.
	
	\begin{theorem}	[The fundamental solution matrix of the Lax pair]\label{thm:solution to the Lax pair}
		The fundamental solution matrix of the Lax pair \eqref{Eq-Lax pair FL} with the elliptic function solution \eqref{Eq-FL-elliptic-solution} at $\lambda(z)$ can be written as
		\begin{align}\label{Eq-fundamental-solution-matrix}
			\mathbf{\Phi} \big(\xi,t;\lambda(z)\big) = \begin{pmatrix} 
				d(\xi;z)d_0(z)E(\xi,t;z) & d(\xi;\hat{z})d_0(\hat{z})E(\xi,t;\hat{z}) \\
				\mathrm{i}d(-\xi;\hat{z})d_0(\hat{z})e^{-F(\xi,t)}E(\xi,t;z) & \mathrm{i}d(-\xi;z)d_0(z)e^{-F(\xi,t)}E(\xi,t;\hat{z})
			\end{pmatrix},
		\end{align}
		where 
		\begin{equation}
			d(\xi;z) = \frac{\sigma(\xi-z)}{\sigma(\xi+\kappa)}.
		\end{equation}
	\end{theorem}
	
	\begin{proof}
		The proof follows directly from Lemma \ref{Prop-psi+} and Lemma \ref{Prop-r+}.
	\end{proof}

	\section{The $N$-elliptic localized solutions to the Fokas--Lenells equation}\label{sec:N-elliptic}
	This section constructs the $N$-elliptic localized solutions of the FL equation in an explicit, closed form, thereby providing the representation on which the asymptotic analysis of Section~\ref{sec:asymptotic} will rest. Building on the elliptic seed solution \eqref{Eq-FL-elliptic-solution} and its fundamental solution matrix \eqref{Eq-fundamental-solution-matrix}, we apply the $N$-fold Darboux--B\"acklund transformation to generate the $N$-elliptic localized solutions and then recast the resulting determinantal ratio in a compact sigma-function form. We begin with the transformation itself.
	
	\begin{prop}[$N$-fold Darboux--B\"acklund transformation]
		Let $\mathbf{\Phi}(\xi,t;\lambda(z))$ be a fundamental solution matrix of the Lax pair \eqref{Eq-Lax pair FL}. Select $N$ distinct uniform parameters $z_i$, $1\leq i\leq N$, and corresponding constant vectors $\mathbf{c}_i$, and define
		\begin{equation}\label{Eq-Phi-i}
			\mathbf{\Phi}_i(\xi,t) = \mathbf{\Phi}(\xi,t;\lambda_i)\mathbf{c}_i, \quad \lambda_i=\lambda(z_i),\quad  1\leq i\leq N.
		\end{equation}
		Then
		\begin{equation}\label{Eq-DT}
			u_N = u + 2\mathbf{\Psi}_1\mathbf{B}^{-1}\mathbf{\Psi}_2^{\dagger},
		\end{equation}
		is also a solution of the FL equation \eqref{Eq-FL} whenever $u$ is a solution. Here $\mathbf{\Psi}_i,i=1,2$ are the $i$-th rows of the matrix
		$[\mathbf{\Phi}_1,\dots,\mathbf{\Phi}_N]$, and the matrix $\mathbf{B}$ is defined by
		\begin{equation}\label{Eq-B}
			\mathbf{B} = \left(\frac{\mathbf{\Phi}_i^\dagger\bm{\sigma}_3\mathbf{\Phi}_j}{\lambda_j + \lambda_i^*} - \frac{\mathbf{\Phi}_i^\dagger\mathbf{\Phi}_j}{\lambda_j - \lambda_i^*}\right)_{1 \leq i,j \leq N}.
		\end{equation}
	\end{prop}
	To derive the $N$-elliptic localized solutions, we choose the fundamental solution matrix given by \eqref{Eq-fundamental-solution-matrix}, together with vectors $\mathbf{c}_i = (1,\alpha_i)^T$, where $\alpha_i \in \mathbb{C}$ for $1 \le i \le N$. Explicitly, we have
	
	\begin{align}\label{Eq-phi-l}
		\mathbf{\Phi}_i=\begin{pmatrix}\Phi_i^{(1)}\\\Phi_i^{(2)}
		\end{pmatrix}:=	\begin{pmatrix} d(\xi;z_i)d_0(z_i)E(\xi,t;z_i)  \\
			\mathrm{i}d(-\xi;\hat{z}_i)d_0(\hat{z}_i)e^{-F(\xi,t)}E(\xi,t;z_i) 	
		\end{pmatrix}
		+\alpha_i\begin{pmatrix}
			d(\xi;\hat{z}_i)d_0(\hat{z}_i)E(\xi,t;\hat{z}_i)\\
			\mathrm{i}d(-\xi;z_i)d_0(z_i)e^{-F(\xi,t)}E(\xi,t;\hat{z}_i)
		\end{pmatrix}.
	\end{align}
	Combining \eqref{Eq-DT}--\eqref{Eq-B} with the Sherman--Morrison--Woodbury identity 
	\begin{equation}\label{Eq-SMW-identity}
		\begin{aligned}
			a+\textbf{q}^{\dagger} \mathbf{A}^{-1} \textbf{p}=\frac{a^{1-N} \operatorname{det}\left(a \mathbf{A}+\textbf{p} \textbf{q}^{\dagger}\right)}{\operatorname{det} \mathbf{A}},
		\end{aligned}
	\end{equation}
	with $a\in\mathbb{C}$, $\mathbf{p},\mathbf{q}\in\mathbb{C}^{N\times 1}$ and $\mathbf{A}\in\mathbb{C}^{N\times N}$, we deduce that 
	\begin{equation}\label{Eq-uN-Sherman}
		\begin{split}
			u_N=\frac{u^{1-N}\det(u\textbf{B}+2\mathbf{\Psi}_2^\dagger\mathbf{\Psi}_1\big)}{\det(\textbf{B})},
		\end{split}
	\end{equation}
	is a solution of the FL equation \eqref{Eq-FL}, where $u$ is the elliptic function solution \eqref{Eq-FL-elliptic-solution}. The following theorem recasts this determinantal ratio in a fully explicit, sigma-function form, which is the representation best suited to the subsequent asymptotic analysis.
	\begin{theorem}[$N$-elliptic localized solution]\label{thm:N-soliton solution}
		The \(N\)-elliptic localized solution \eqref{Eq-uN-Sherman} of the FL equation \eqref{Eq-FL} can be expressed as
		\begin{equation}\label{Eq-uN}
			u_N(\xi,t)=\sqrt{\nu_0}\frac{\sigma(\kappa)}{\sigma(\rho)}\Big(\frac{\sigma(\xi+\kappa)}{\sigma(\xi+\rho)}\Big)^{N-1}\frac{\det(\mathbf{U}_{1,N})}{\det(\mathbf{U}_{2,N})}e^{F(\xi,t)},
		\end{equation}
		where the matrix entries are given by
		\begin{equation}\label{Eq-U1-U2}
			\begin{split}
				&\big(\mathbf{U}_{s,N}\big)_{ij}=\sum_{m,n=0}^{1}\big(-\alpha_i^*\big)^{m}\alpha_j^{n}I_i^{1-m}I_0^{m}(\xi)\\
				&\quad \Sigma^{(s)}\Big(\xi;\frac{z_i^*+\check{z}_i+(-1)^{m}\big(z_i^*-\check{z}_i\big)}{2},\frac{z_j+\hat{z}_j+(-1)^{n}\big(z_j-\hat{z}_j\big)}{2}\Big)\\
				&\quad d_0^*\Big(\frac{z_i+\hat{z}_i-(-1)^{m}\big(z_i-\hat{z}_i\big)}{2}\Big) d_0\Big(\frac{z_j+\hat{z}_j+(-1)^{n}\big(z_j-\hat{z}_j\big)}{2}\Big) \\
				&\quad E^*\Big(\frac{z_i+\hat{z}_i+(-1)^{m}\big(z_i-\hat{z}_i\big)}{2}\Big)
				E\Big(\frac{z_j+\hat{z}_j+(-1)^{n}\big(z_j-\hat{z}_j\big)}{2}\Big),\qquad s=1,2,
			\end{split}
		\end{equation}
		with
		\begin{equation}\label{Eq-z-check}
			\check{z}=\kappa+\rho-z^*,
		\end{equation}
		and the functions \(I_0(\xi)\) and \(I_i\) defined by
		\begin{equation}\label{Eq-Ii}
			\begin{split}
				I_i=\begin{dcases}
					-1,&   \mathrm{Re}(\rho)=0,  \\
					e^{-2\zeta(\omega_1)(\kappa+\rho-\omega_1-z_i^*)},& \mathrm{Re}(\rho)=\omega_1,\\   
				\end{dcases}\qquad
				I_0(\xi)=\begin{dcases}
					1,&   \mathrm{Re}(\rho)=0,\\
					e^{2\zeta(\omega_1)\xi},&    \mathrm{Re}(\rho)=\omega_1.\\
				\end{dcases}
			\end{split}
		\end{equation}
		Moreover, the \(\Sigma\)-functions are defined by
		\begin{equation}\label{Eq-Sigma-i}
			\begin{split}
				&\Sigma^{(1)}(\xi;z,w)=\frac{\sigma(\xi+\rho-z-w)\sigma(z-\kappa)\sigma(\rho+w)}{\sigma(z+w)},\\
				&\Sigma^{(2)}(\xi;z,w)=\frac{\sigma(\kappa+\xi-z-w)\sigma(\kappa+w)\sigma(-\rho+z)}{\sigma(z+w)}.\\
			\end{split}
		\end{equation}
	\end{theorem}
	\begin{remark}\label{rmk:rho-omega1-thmN}
		The case distinction $\mathrm{Re}(\rho)=0$ versus $\mathrm{Re}(\rho)=\omega_1$ enters only through the auxiliary factors $I_i$ and $I_0(\xi)$ in \eqref{Eq-Ii}. In the proof below we treat the case $\mathrm{Re}(\rho)=0$ in full; the case $\mathrm{Re}(\rho)=\omega_1$ is established by the identical argument, the only difference being the half-period shift of the $\sigma$-function dictated by its quasi-periodicity \eqref{Eq-quasi-periodic-sigma}, which produces exactly the exponential factors recorded in \eqref{Eq-Ii}. We therefore omit it.
	\end{remark}
	
	\begin{proof}
		We prove only the case \(\mathrm{Re}(\rho)=0\). As explained in Remark~\ref{rmk:rho-omega1-thmN}, the case \(\mathrm{Re}(\rho)=\omega_1\) follows verbatim once the quasi-periodicity \eqref{Eq-quasi-periodic-sigma} of the \(\sigma\)-function is taken into account. Based on \eqref{Eq-B}, the entries of the matrix \(\mathbf{B}\) are given by
		\begin{equation}\label{Eq-Bij-0}
			\begin{split}
				\mathbf{B}_{ij}
				=\frac{-2\lambda_j \Phi_{i}^{(2)*} \Phi_{j}^{(2)}
					-2\lambda_i^{*} \Phi_{i}^{(1)*} \Phi_{j}^{(1)}}
				{\lambda_j^{2}-(\lambda_i^{*})^{2}},\quad 1\leq i,j\leq N.
			\end{split}
		\end{equation}
		Using the parameterization \eqref{Eq-parameterization-lambda-y-0}, the denominator in \eqref{Eq-Bij-0} can be written as
		\begin{equation}
			\begin{split}
				\lambda_j^{2}-	(\lambda_i^*)^{2}=-\frac{\mathrm{i}\sigma(\kappa-\rho)\sigma(z_i^*+z_j)\sigma(-\kappa-\rho+z_i^*-z_j)}{4\sigma(\kappa+z_j)\sigma(\rho+z_j)\sigma(-\kappa+z_i^*)\sigma(-\rho+z_i^*)}.
			\end{split}
		\end{equation}
		Using \eqref{Eq-phi-l}, the numerator of \eqref{Eq-Bij-0} takes the form
		\begin{equation}\label{Eq-Mij-numerator-alt}
			\begin{split}
				&\quad\quad -2\lambda_j \Phi_{i}^{(2)*} \Phi_{j}^{(2)}
				-2\lambda_i^{*} \Phi_{i}^{(1)*} \Phi_{j}^{(1)} \\
				&= -2\lambda_j 
				\left( E^*(\xi,t;z_i) d^*(-\xi;\hat{z}_i) d_0^*(\hat{z}_i)
				+ \alpha_i^* E^*(\xi,t;\hat{z}_i) d^*(-\xi;z_i) d_0^*(z_i) \right) \\
				&\quad \left( E(\xi,t;z_j) d(-\xi;\hat{z}_j) d_0(\hat{z}_j)
				+ \alpha_j E(\xi,t;\hat{z}_j) d(-\xi;z_j) d_0(z_j) \right)e^{-(F(\xi,t)+F^*(\xi,t))} \\
				&\quad - 2\lambda_i^*
				\left( E^*(\xi,t;z_i) d^*(\xi;z_i) d_0^*(z_i)
				+ \alpha_i^* E^*(\xi,t;\hat{z}_i) d^*(\xi;\hat{z}_i) d_0^*(\hat{z}_i) \right) \\
				&\quad \left( E(\xi,t;z_j) d(\xi;z_j) d_0(z_j)
				+ \alpha_j E(\xi,t;\hat{z}_j) d(\xi;\hat{z}_j) d_0(\hat{z}_j) \right).
			\end{split}
		\end{equation}
		It follows from \eqref{Eq-Mij-numerator-alt} that the numerator of \eqref{Eq-Bij-0} is a sum of four terms, each of which can be expressed as follows:
		\begin{equation}\label{Eq-Mij-numerator-alt-merged}
			\begin{split}
				&\quad \lambda_j e^{-(F(\xi,t)+F^*(\xi,t))} d^*(-\xi;\hat{z}_i)  d(-\xi;\hat{z}_j)d_0^*(\hat{z}_i) d_0(\hat{z}_j)  +\lambda_i^* d^*(\xi;z_i)  d(\xi;z_j)d_0^*(z_i) d_0(z_j)\\
				&=-\frac{\sigma(\rho)\sigma(-\kappa+z_i^*+z_j-\xi)\sigma(\kappa-z_i^*+z_j+\rho)}{4\sqrt{\nu_0}\sigma(\kappa)\sigma(\rho+z_j)\sigma(\kappa-z_i^*)\sigma(-\xi-\kappa)}d_0(z_j)d_0^*(\hat{z}_i),\\
				&\quad \lambda_j e^{-(F(\xi,t)+F^*(\xi,t))}  d^*(-\xi;\hat{z}_i)  d(-\xi;z_j)d_0^*(\hat{z}_i) d_0(z_j) +\lambda_i^* d^*(\xi;z_i)  d(\xi;\hat{z}_j)d_0^*(z_i) d_0(\hat{z}_j)\\
				&=-\frac{\sigma(\rho)\sigma(2\kappa+\rho+z_j-z_i^*+\xi)\sigma(z_i^*+z_j)}{4\sqrt{\nu_0}\sigma(\kappa)\sigma(\kappa+z_j)\sigma(\kappa-z_i^*)\sigma(\xi+\kappa)}d_0(\hat{z}_j)d_0^*(\hat{z}_i),\\
				&\quad \lambda_j e^{-(F(\xi,t)+F^*(\xi,t))}  d^*(-\xi;z_i)  d(-\xi;\hat{z}_j)d_0^*(z_i) d_0(\hat{z}_j) +\lambda_i^* d^*(\xi;\hat{z}_i)  d(\xi;z_j)d_0^*(\hat{z}_i) d_0(z_j)\\
				&=-\frac{\sigma(\rho)\sigma(z_i^*+z_j)\sigma(-\rho+z_i^*-z_j+\xi)}{4\sqrt{\nu_0}\sigma(\kappa)\sigma(\xi+\kappa)\sigma(-\rho+z_i^*)\sigma(\rho+z_j)}d_0(z_j)d_0^*(z_i),\\
				&\quad \lambda_j e^{-(F(\xi,t)+F^*(\xi,t))}  d^*(-\xi;z_i)  d(-\xi;z_j)d_0^*(z_i) d_0(z_j) +\lambda_i^*  d^*(\xi;\hat{z}_i)  d(\xi;\hat{z}_j)d_0^*(\hat{z}_i) d_0(\hat{z}_j)\\
				&=-\frac{\sigma(\rho)\sigma(-z_i^*+z_j+\kappa+\rho)\sigma(\xi+z_i^*+z_j+\kappa)}{4\sqrt{\nu_0}\sigma(\kappa)\sigma(\xi+\kappa)\sigma(\kappa+z_j)\sigma(-\rho+z_i^*)}d_0(\hat{z}_j)d_0^*(z_i),
			\end{split}
		\end{equation}
		respectively. Therefore, combining \eqref{Eq-FL-elliptic-solution} with \eqref{Eq-Bij-0}--\eqref{Eq-Mij-numerator-alt-merged} yields
		\begin{equation}\label{Eq-denominator}
			\begin{split}
				&\quad u\mathbf{B}_{ij}=\frac{-2\mathrm{i}\sigma(\xi+\rho)e^{F(\xi,t)}}{\sigma(\kappa-\rho)\sigma^2(\xi+\kappa)}\\
				&\quad \quad\Big(\frac{\sigma(-\kappa+z_i^*+z_j-\xi)\sigma(\kappa+z_j)\sigma(-\rho+z_i^*)}{\sigma(z_i^*+z_j)}d_0(z_j)d_0^*(\hat{z}_i)E(\xi,t;z_j)E^*(\xi,t;z_i)\\
				&\quad\quad +\alpha_j\frac{\sigma(2\kappa+\rho+z_j-z_i^*+\xi)\sigma(\rho+z_j)\sigma(-\rho+z_i^*)}{\sigma(-\kappa-\rho+z_i^*-z_j)}d_0(\hat{z}_j)d_0^*(\hat{z}_i)E(\xi,t;\hat{z}_j)E^*(\xi,t;z_i)\\
				&\quad\quad +\alpha_i^*\frac{\sigma(-\rho+z_i^*-z_j+\xi)\sigma(\kappa+z_j)\sigma(\kappa-z_i^*)}{\sigma(-\kappa-\rho+z_i^*-z_j)} d_0(z_j)d_0^*(z_i)E(\xi,t;z_j)E^*(\xi,t;\hat{z}_i)\\
				&\quad\quad+\alpha_i^*\alpha_j\frac{\sigma(\xi+z_i^*+z_j+\kappa)\sigma(\rho+z_j)\sigma(-\kappa+z_i^*)}{\sigma(z_i^*+z_j)} d_0(\hat{z}_j)d_0^*(z_i)E(\xi,t;\hat{z}_j)E^*(\xi,t;\hat{z}_i)\Big).
			\end{split}
		\end{equation}
		On the other hand, using \eqref{Eq-phi-l} we obtain
		\begin{equation}\label{Eq-Psi-entries}
			\begin{split}
				&\quad\big(\mathbf{\Psi}_2^\dagger\mathbf{\Psi}_1\big)_{ij}\\
				&=-\mathrm{i}e^{F(\xi,t)}\big(d(\xi;z_j)d^*(-\xi;\hat{z}_i)d_0(z_j)d_0^*(\hat{z}_i)E(\xi,t;z_j)E^*(\xi,t;z_i)\\
				&\quad +\alpha_jd(\xi;\hat{z}_j)d^*(-\xi;\hat{z}_i)d_0(\hat{z}_j)d_0^*(\hat{z}_i)E(\xi,t;\hat{z}_j)E^*(\xi,t;z_i)\\
				&\quad +\alpha_i^*d(\xi;z_j)d^*(-\xi;z_i)d_0(z_j)d_0^*(z_i)E(\xi,t;z_j)E^*(\xi,t;\hat{z}_i)\\
				&\quad +\alpha_i^*\alpha_j d(\xi;\hat{z}_j)d^*(-\xi;z_i)d_0(\hat{z}_j)d_0^*(z_i)E(\xi,t;\hat{z}_j)E^*(\xi,t;\hat{z}_i)\big)\\
				&=-\frac{\mathrm{i}e^{F(\xi,t)}}{\sigma^2(\xi+\kappa)}\big(\sigma(\xi-z_j)\sigma(\xi+\kappa+\rho-z_i^*)d_0(z_j)d_0^*(\hat{z}_i)E(\xi,t;z_j)E^*(\xi,t;z_i)\\
				&\quad +\alpha_j\sigma(\xi+\kappa+\rho+z_j)\sigma(\xi+\kappa+\rho-z_i^*)d_0(\hat{z}_j)d_0^*(\hat{z}_i)E(\xi,t;\hat{z}_j)E^*(\xi,t;z_i)\\
				&\quad +\alpha_i^*\sigma(\xi-z_j)\sigma(\xi+z_i^*)d_0(z_j)d_0^*(z_i)E(\xi,t;z_j)E^*(\xi,t;\hat{z}_i)\\
				&\quad +\alpha_i^*\alpha_j\sigma(\xi+\kappa+\rho+z_j)\sigma(\xi+z_i^*) d_0(\hat{z}_j)d_0^*(z_i)E(\xi,t;\hat{z}_j)E^*(\xi,t;\hat{z}_i)\big).\\
			\end{split}
		\end{equation}
		Combining \eqref{Eq-denominator} with \eqref{Eq-Psi-entries} gives
		\begin{equation}\label{Eq-numerator}
			\begin{split}
				&\quad \big(u\mathbf{B}+2\mathbf{\Psi}_2^\dagger\mathbf{\Psi}_1\big)_{ij}\\
				&=\frac{-2\mathrm{i}e^{F(\xi,t)}}{\sigma(\xi+\kappa)\sigma(\kappa-\rho)}\Big(\frac{\sigma(-\xi-\rho+z_j+z_i^*)\sigma(z_i^*-\kappa)\sigma(\rho+z_j)}{\sigma(z_i^*+z_j)}d_0(z_j)d_0^*(\hat{z}_i)E(\xi,t;z_j)E^*(\xi,t;z_i)\\
				&\quad +\alpha_j\frac{\sigma(\kappa+z_j)\sigma(-\kappa+z_i^*)\sigma(\xi+2\rho+\kappa+z_j-z_i^*)}{\sigma(-\kappa-\rho+z_i^*-z_j)}d_0(\hat{z}_j)d_0^*(\hat{z}_i)E(\xi,t;\hat{z}_j)E^*(\xi,t;z_i)\\
				&\quad +\alpha_i^*\frac{\sigma(\rho-z_i^*)\sigma(\rho+z_j)\sigma(\xi-\kappa-z_j+z_i^*)}{\sigma(-\kappa-\rho+z_i^*-z_j)}d_0(z_j)d_0^*(z_i)E(\xi,t;z_j)E^*(\xi,t;\hat{z}_i)\\
				&\quad +\alpha_i^*\alpha_j \frac{\sigma(-\rho+z_i^*)\sigma(\kappa+z_j)\sigma(\xi+\rho+z_j+z_i^*)}{\sigma(z_i^*+z_j)}d_0(\hat{z}_j)d_0^*(z_i)E(\xi,t;\hat{z}_j)E^*(\xi,t;\hat{z}_i)\Big).
			\end{split}
		\end{equation}
		From \eqref{Eq-FL-elliptic-solution}, \eqref{Eq-denominator} and \eqref{Eq-numerator} we obtain the elliptic solution in the form \eqref{Eq-uN}, which completes the proof.
	\end{proof}

	\section{The asymptotic analysis of the $N$-elliptic localized solutions to the Fokas--Lenells equation}\label{sec:asymptotic}
	With the explicit sigma-function representation \eqref{Eq-uN}--\eqref{Eq-Sigma-i} of the $N$-elliptic localized solution now at hand, we turn to its long-time behavior. The goal of this section is to show that, as $t\to\pm\infty$, the solution disintegrates into $N$ individual first-order elliptic localized waves that travel at distinct velocities over a common elliptic background, and that the collisions between them are elastic. To this end we partition the exterior of the interaction region into the propagation directions $L_k^{\pm}$ and the intermediate regions $R_k^{\pm}$, and we examine the solution on each in turn: along $L_k^{\pm}$ it reduces to a single first-order elliptic localized wave, while in $R_k^{\pm}$ it collapses to a shifted copy of the elliptic seed \eqref{Eq-FL-elliptic-solution}. A symmetry characterization, which promotes the elastic collisions to strictly elastic ones, brings the section to a close. The entire analysis rests on an explicit evaluation of the structured determinants appearing in \eqref{Eq-uN}, for which the decisive tool is the sigma-function analogue of the classical Cauchy determinant, recorded in the next theorem.
	\begin{theorem}[Sigma-function version of Cauchy determinants]\cite{lingtang2026dnls}\label{thm:Cauchy matrix}
		Assume $\tau$, $m_1,\ldots m_N,n_1,\ldots n_N\in\mathbb{C}$ and $N\in\mathbb{Z}$. Then the determinant identity
		\begin{align*}
			\det\left(\frac{\sigma(\tau+m_i+n_j)}{\sigma(m_i+n_j)}\right)_{1\leq i,j\leq N}=D_N(\tau;m_1,\ldots,m_N,n_1,\ldots,n_N),
		\end{align*}
		holds, where $D_N(\tau;m_1,\ldots,m_N,n_1,\ldots,n_N)$ is explicitly defined by
		\begin{equation*}
			D_N(\tau;m_1,\ldots,m_N,n_1,\ldots,n_N):=\frac{\sigma^{N-1}(\tau)\sigma\big(\tau+\sum_{i=1}^{N}(m_i+n_i)\big)\prod_{1\leq i<j\leq N}\sigma(m_i-m_j)\sigma(n_i-n_j)}{\prod_{i,j=1}^{N}\sigma(m_i+n_j)}.	
		\end{equation*}
	\end{theorem}

	With this identity at our disposal, we proceed to the asymptotic analysis proper.
	For simplicity, we work throughout in the $(\xi,t)$-coordinate. 
	The domain exterior to the interaction region is partitioned into $2N$ regions by $N$ pairs of lines $L_i^{\pm},i=1,2,\ldots,N$. The regions bounded by $L_{i-1}^{\pm}$ and $L_{i}^{\pm}$ are denoted by $R_i^{\pm}$ for $i=2,\ldots,N$. Additionally, the regions bounded by $L_1^{\pm}$ and $L_N^{\mp}$ are labeled $R_1^{\pm}$, with the superscripts $\pm$ indicating positive and negative temporal directions. 
	
	It will be shown in
	Theorem \ref{thm:AB1} that the lines $L_i^{\pm}$ characterize the propagation directions of the $N$-elliptic localized solution \eqref{Eq-uN}, defined by
	\begin{align}\label{Eq-line}
		L_i^{\pm}:=\{\xi=v(z_i)t-c_i^{\pm}\},\quad i=1,2,\ldots,N,
	\end{align}
	where $c_i^{\pm}\in\mathbb{R}$ are constants. The velocities $v_i$ of each elliptic localized wave are determined by
	\begin{equation}\label{Eq-vi}
		\begin{split}
			v(z_i)&=-\frac{\mathrm{Im}(y(z_i))}{2\mathrm{Re}(\beta(z_i))},\quad 
			\beta(z_i)=(z_i-\hat{z}_i)\frac{\zeta(\omega_1)}{\omega_1}+\zeta(\kappa-z_i)-\zeta(\kappa-\hat{z}_i),
		\end{split}
	\end{equation}
	which will be verified in the proof of the following theorem.
	We conduct the asymptotic analysis along $L_k^{\pm}$ and in the regions $R_k^{\pm}$ as $t\rightarrow \pm\infty$ in the following theorems. We begin with the behavior along the propagation directions.
	
	\begin{theorem}[Asymptotic analysis along the propagation directions]\label{thm:AB1}
		Assume $\operatorname{Re}(\beta(z_i))>0$ for $i=1,2,\ldots,N$ and $v(z_i)<v(z_j)$ for $1\leq i<j\leq N$.
		As $t\rightarrow -\infty$, the asymptotic form of the $N$-elliptic localized solution $u_N$ given by \eqref{Eq-uN} along the line $L_k^{-}$ defined in \eqref{Eq-line}--\eqref{Eq-vi} reads
		\begin{equation}\label{Eq-uN-asym}
			u_{N,L_k^-}(\xi,t)=\sqrt{\nu_0}\,C_{N,L_k^-}\frac{\sigma(\kappa)}{\sigma(\rho)}\frac{V_{1,N,L_k^-}}{V_{2,N,L_k^-}}\,e^{F(\xi,t)},
		\end{equation}
		where
		\begin{equation}
			\begin{split}
				&	V_{s,N,L_k^-}=\sum_{m,n=0}^{1}\big(-\alpha_{N,L_k^-}^{*}\big)^{m}\big(\alpha_{N,L_k^-}\big)^{n}I_k^{1-m}I_0^{m}\big(\xi+z_{N,L_k^-}\big)\\
				&\quad \Sigma^{(s)}\Big(\xi+z_{N,L_k^-};\frac{z_k^*+\check{z}_k+(-1)^{m}\big(z_k^*-\check{z}_k\big)}{2},\frac{z_k+\hat{z}_k+(-1)^{n}(z_k-\hat{z}_k)}{2}\Big)\\
				&\quad d_0^*\Big(\frac{z_k+\hat{z}_k-(-1)^{m}\big(z_k-\hat{z}_k\big)}{2}\Big) d_0\Big(\frac{z_k+\hat{z}_k+(-1)^{n}(z_k-\hat{z}_k)}{2}\Big) \\
				&\quad E^*\Big(\frac{z_k+\hat{z}_k+(-1)^{m}(z_k-\hat{z}_k)}{2}\Big)
				E\Big(\frac{z_k+\hat{z}_k+(-1)^{n}(z_k-\hat{z}_k)}{2}\Big),\qquad s=1,2,
			\end{split}
		\end{equation}
		and the relevant parameters are defined by
		\begin{equation}\label{Eq-asym-Lkm-2}
			\begin{split} 
				&r(z)=\frac{\sigma(\rho+z)}{\sigma(\kappa+z)},\\[2mm]
				&C_{N,L_k^-}=\prod_{i=1}^{k-1}\frac{r(z_i)}{r(-z_i^*)}
				\prod_{i=k+1}^{N}\frac{r(-z_i^*)}{r(z_i)},\\[2mm]
				&z_{N,L_k^-}=-\sum_{i=1}^{k-1}(z_i+z_i^*)+\sum_{i=k+1}^{N}(z_i+z_i^*),\\[2mm]
				&S_{N,L_k^-}=\prod_{i=1}^{k-1}\frac{\sigma(z_i-\hat{z}_k)\sigma(z_k+z_i^*)}{\sigma(\hat{z}_k+z_i^*)\sigma(z_i-z_k)}
				\prod_{i=k+1}^{N}\frac{\sigma(\hat{z}_k-\hat{z}_i)\sigma(z_k+\check{z}_i)}{\sigma(\hat{z}_k+\check{z}_i)\sigma(z_k-\hat{z}_i)},\quad\alpha_{N,L_k^-}=\alpha_k S_{N,L_k^{-}}.
			\end{split}
		\end{equation}
		Moreover, the corresponding quantities for the case $L_k^+$ satisfy
		\begin{align}\label{Eq-z-Lk+}
			z_{N,L_k^+} = -z_{N,L_k^-},\quad 
			S_{N,L_k^+} = (S_{N,L_k^-})^{-1},\quad 
			\alpha_{N,L_k^+} = \alpha_k S_{N,L_k^+},\quad 
			C_{N,L_k^+} = (C_{N,L_k^-})^{-1}.
		\end{align}
		The statement covers both $\mathrm{Re}(\rho)=0$ and $\mathrm{Re}(\rho)=\omega_1$. The factors $I_k$ and $I_0(\xi)$ are those defined in \eqref{Eq-Ii}.
	\end{theorem}
	
	\begin{proof}
		We give the argument only for the line $L_k^-$ with $\operatorname{Re}(\rho)=0$. The remaining propagation directions are treated identically, and the case $\operatorname{Re}(\rho)=\omega_1$ requires no new idea: it differs only through the half-period shift of the $\sigma$-function, governed by the quasi-periodicity \eqref{Eq-quasi-periodic-sigma}, which is already encoded in the factors $I_k$ and $I_0(\xi)$ of \eqref{Eq-Ii}. We therefore omit the corresponding details.
		Equation \eqref{Eq-uN-Sherman} can be rewritten as 
		\begin{equation}\label{Eq-uN-Sherman-rewritten}
			u_N=\frac{u\det\big(u\mathbf{B}+2\mathbf{\Psi}_2^\dagger\mathbf{\Psi}_1\big)}{\det(u\mathbf{B})}.
		\end{equation}
		Based on \eqref{Eq-denominator} and \eqref{Eq-numerator}, the entries of $u\mathbf{B}+2\mathbf{\Psi}_2^\dagger\mathbf{\Psi}_1$ and $u\mathbf{B}$ can be rewritten as
		\begin{equation}\label{Eq-entries-numerator}
			\begin{split}
				\big(u\mathbf{B}+2\mathbf{\Psi}_2^\dagger\mathbf{\Psi}_1\big)_{ij}
				&= T_1(\xi,t)\tilde{\mathbf{E}}_j(\mathbf{\Delta}^{(1)})_{ij}\tilde{\mathbf{E}}_i^\dagger,\quad 
				\big(u\mathbf{B}\big)_{ij}
				= T_2(\xi,t) \tilde{\mathbf{E}}_j (\mathbf{\Delta}^{(2)})_{ij}\tilde{\mathbf{E}}_i^\dagger,\quad i,j=1,2,\ldots,N,
			\end{split}
		\end{equation}
		respectively, where 
		\begin{equation}\label{Eq-tildeE}
			\begin{aligned}[t]
				&T_1(\xi,t)=\frac{-2\mathrm{i}e^{\frac{\zeta(\omega_1)}{2\omega_1}(\xi^2+2\rho)\xi+F(\xi,t)}}{\sigma(\xi+\kappa)\sigma(\kappa-\rho)},\\
				&T_2(\xi,t)=\frac{-2\mathrm{i}\sigma(\xi+\rho)e^{\frac{\zeta(\omega_1)}{2\omega_1}(\xi^2+2\kappa)\xi+F(\xi,t)}}{\sigma(\kappa-\rho)\sigma^2(\xi+\kappa)},\\
				&\tilde{\mathbf{E}}_i=\begin{pmatrix}
					\tilde{E}(\xi,t;z_i) & \alpha_i \tilde{E}(\xi,t;\hat{z}_i)
				\end{pmatrix},\quad
				\tilde{E}(\xi,t;z)=E(\xi,t;z)e^{-\frac{\zeta(\omega_1)}{\omega_1}z\xi},\\
				&\big(\mathbf{\Delta}^{(l)}\big)_{ij}=- \begin{pmatrix}
					\Delta^{(l)}(\xi;z_i^*,z_j)d_0(z_j)d_0^*(\hat{z}_i) & \Delta^{(l)}(\xi;\check{z}_i,z_j)d_0(z_j)d_0^*(z_i) \\
					\Delta^{(l)}(\xi;z_i^*,\hat{z}_j)d_0(\hat{z}_j)d_0^*(\hat{z}_i) & \Delta^{(l)}(\xi;\check{z}_i,\hat{z}_j)d_0(\hat{z}_j)d_0^*(z_i)
				\end{pmatrix},\quad l=1,2,\\
				&\Delta^{(1)}(\xi;z,w)=\Sigma^{(1)}(\xi;z,w)e^{-\frac{\zeta(\omega_1)}{2\omega_1}(\xi^2-2(z+w-\rho)\xi)},\\
				&\Delta^{(2)}(\xi;z,w)=\Sigma^{(2)}(\xi;z,w)e^{-\frac{\zeta(\omega_1)}{2\omega_1}(\xi^2-2(z+w-\kappa)\xi)}.
			\end{aligned}
		\end{equation}
		The functions $\Delta^{(1,2)}(\xi;z,w)$ are bounded and periodic in $\xi$ for arbitrary $z,w\in\mathbb{C}$. From \eqref{Eq-entries-numerator}--\eqref{Eq-tildeE} we obtain
		\begin{equation}\label{Eq-asym-solution}
			\frac{\det\big(u\mathbf{B}+2\mathbf{\Psi}_2^\dagger\mathbf{\Psi}_1\big)}{\det(u\mathbf{B})}
			= \frac{\det\Big(\sum\limits_{m,n=1}^{2}T_1(\xi,t)\mathbf{X}_n^\dagger(\mathbb{I}_N\otimes \mathbf{e}_m^\mathsf{T})\mathbf{\Delta}^{(1)}(\mathbb{I}_N\otimes \mathbf{e}_n)\mathbf{X}_m\Big)}
			{\det\Big(\sum\limits_{m,n=1}^{2}T_2(\xi,t)\mathbf{X}_n^\dagger(\mathbb{I}_N\otimes \mathbf{e}_m^\mathsf{T})\mathbf{\Delta}^{(2)}(\mathbb{I}_N\otimes \mathbf{e}_n)\mathbf{X}_m\Big)},
		\end{equation}
		where $\mathbf{e}_1=(1,0)^\mathsf{T}$, $\mathbf{e}_2=(0,1)^\mathsf{T}$,
		\begin{equation}\label{Eq-Aij-entries}
			\begin{split}
				\mathbf{X}_1&=\operatorname{diag}\begin{pmatrix}
					1,\dots,1,e^{-\tau_{k+1}},\dots,e^{-\tau_N}
				\end{pmatrix},\quad 
				\mathbf{X}_2=\operatorname{diag}\begin{pmatrix}
					e^{\tau_{1}},\dots,e^{\tau_{k}},1,\dots,1
				\end{pmatrix},
			\end{split}
		\end{equation} 
		and $\otimes$ denotes the tensor product of matrices. For $1\leq i\leq N$,
		the functions 
		\begin{equation}\label{Eq-tau-i}
			\tau_i:=\ln\big(\alpha_i\tilde{E}(\xi,t;\hat{z}_i)\tilde{E}^{-1}(\xi,t;z_i)\big)
			=\operatorname{Re}(\beta(z_i))(\xi-v(z_i)t)+\mathrm{i}\Big(\operatorname{Im}(\beta(z_i))\xi-\frac{1}{2}\operatorname{Re}(y(z_i))t\Big)+\ln(\alpha_i),
		\end{equation}
		determine the propagation directions of each elliptic localized wave, where we have used \eqref{Eq-vi} and \eqref{Eq-tildeE}. Since that $\operatorname{Re}(\beta(z_i))>0$ for $i=1,2,\ldots,N$ and $v(z_i)<v(z_j)$ for $1\leq i<j\leq N$, along $L_k^{-}$ as $t\rightarrow -\infty$ we have
		\begin{equation}\label{Eq-tau-2}
			\mathbf{X}_1\rightarrow\operatorname{diag}\Big(
			\overbrace{1,\dots,1}^{k},0,\dots,0
			\Big),\qquad
			\mathbf{X}_2\rightarrow\operatorname{diag}\Big(
			0,\dots,e^{\tau_{k}},\overbrace{1,\dots,1}^{N-k}
			\Big).
		\end{equation}
		Combining \eqref{Eq-FL-elliptic-solution}, \eqref{Eq-uN-Sherman-rewritten}, \eqref{Eq-asym-solution}, \eqref{Eq-tau-2}, and expanding the determinants $\mathbf{X}_n^\dagger(\mathbb{I}_N\otimes \mathbf{e}_m^\mathsf{T})\mathbf{\Delta}^{(1,2)}(\mathbb{I}_N\otimes \mathbf{e}_n)\mathbf{X}_m$ with respect to the $k$-th row and column 
		yields the asymptotic representation
		\begin{equation}\label{Eq-N-breather-asymptotic}
			\begin{split}
				u_N &= u_{N,L_k^-}+\mathcal{O}\big(\exp\big(\min_{i\neq k}|\beta(z_i)(v(z_i)-v(z_k))|t\big)\big),\\
				u_{N,L_k^-} &= \sqrt{\nu_0}\frac{\sigma(\kappa)}{\sigma(\rho)}\Big(\frac{\sigma(\xi+\kappa)}{\sigma(\xi+\rho)}\Big)^{N-1}\frac{W_{N,L_k^{-}}^{(1)}}{W_{N,L_k^{-}}^{(2)}}e^{F(\xi,t)},
			\end{split}
		\end{equation}
		where
		\begin{equation}
			\begin{split}
				&	W_{N,L_k^{-}}^{(s)} = -\big(\det(\mathbf{W}_{N,L_k^{-}}^{[(s),1,1]})+\det(\mathbf{W}_{N,L_k^{-}}^{[(s),2,1]})e^{\tau_k}+\det(\mathbf{W}_{N,L_k^{-}}^{[(s),1,2]})e^{\tau_k^*}+\det(\mathbf{W}_{N,L_k^{-}}^{[(s),2,2]})e^{\tau_k+\tau_k^*}\big),\\
				&\big(\mathbf{W}_{N,L_k^-}^{[(s),m,n]}\big)_{ij} = \Sigma^{(s)}(\xi;\delta^{[m,n]}_i,\eta^{[m,n]}_j) d_0^*(\varepsilon_i^{[m,n]})d_0(\eta_j^{[m,n]})e^{\frac{\zeta(\omega_1)}{\omega_1}(\delta^{[m,n]}_i+\eta^{[m,n]}_j)\xi},\qquad s,m,n=1,2,
			\end{split}
		\end{equation}
		and 
		\begin{equation}\label{Eq-N-breather-asymptotic-3}
			\begin{split}
				&\big(\eta_1^{[1,n]}, \eta_2^{[1,n]}, \ldots, \eta_k^{[1,n]}, \eta_{k+1}^{[1,n]}, \ldots, \eta_N^{[1,n]}\big)=\big(z_1,z_2,\ldots,z_k,\hat{z}_{k+1},\ldots,\hat{z}_{N}\big),\\ 
				&\big(\eta_1^{[2,n]}, \eta_2^{[2,n]}, \ldots, \eta_{k-1}^{[2,n]}, \eta_k^{[2,n]}, \ldots, \eta_N^{[2,n]}\big)=\big(z_1,z_2,\ldots,z_{k-1},\hat{z}_k,\ldots,\hat{z}_{N}\big), \\	
				&\big(\delta_1^{[m,1]}, \delta_2^{[m,1]}, \ldots, \delta_k^{[m,1]}, \delta_{k+1}^{[m,1]}, \ldots, \delta_N^{[m,1]}\big)=\big(z_1^*,z_2^*,\ldots,z_k^*,\check{z}_{k+1},\ldots,\check{z}_{N}\big), \\
				&\big(\delta_1^{[m,2]}, \delta_2^{[m,2]}, \ldots, \delta_{k-1}^{[m,2]}, \delta_k^{[m,2]}, \ldots, \delta_N^{[m,2]}\big)=\big(z_1^*,z_2^*,\ldots,z_{k-1}^*,\check{z}_{k},\ldots,\check{z}_{N}\big),\\
				&\big(\varepsilon_1^{[m,1]}, \varepsilon_2^{[m,1]}, \ldots, \varepsilon_k^{[m,1]}, \varepsilon_{k+1}^{[m,1]}, \ldots, \varepsilon_N^{[m,1]}\big)=\big(\hat{z}_1,\hat{z}_2,\ldots,\hat{z}_k,z_{k+1},\ldots,z_N\big), \\
				&\big(\varepsilon_1^{[m,2]}, \varepsilon_2^{[m,2]}, \ldots, \varepsilon_{k-1}^{[m,2]}, \varepsilon_k^{[m,2]}, \ldots, \varepsilon_N^{[m,2]}\big)=\big(\hat{z}_1,\hat{z}_2,\ldots,\hat{z}_{k-1},z_k,\ldots,z_N\big).
			\end{split}
		\end{equation}
		Applying Theorem \ref{thm:Cauchy matrix}, we evaluate the determinants appearing in \eqref{Eq-N-breather-asymptotic}--\eqref{Eq-N-breather-asymptotic-3} as
		\begin{equation}\label{Eq-detV}
			\begin{split}
				&\det\big(\mathbf{W}_{N,L_k^-}^{[(1),m,n]}\big)= (-1)^{N} D_N(-\xi-\rho) \prod_{i=1}^N \sigma(\delta_i^{[m,n]}-\kappa)\sigma(\rho+\eta_i^{[m,n]})A_i^{[m,n]}(\xi),\\
				&\det\big(\mathbf{W}_{N,L_k^-}^{[(2),m,n]}\big)= (-1)^{N} D_N(-\xi-\kappa) \prod_{i=1}^N\sigma(\delta_i^{[m,n]}-\rho) \sigma(\kappa+\eta_i^{[m,n]})A_i^{[m,n]}(\xi),\\
				&A_i^{[m,n]}(\xi)=d_0^*(\varepsilon_i^{[m,n]})d_0(\eta_i^{[m,n]})e^{(\delta_i^{[m,n]}+\eta_i^{[m,n]})\frac{\zeta(\omega_1)}{\omega_1}\xi},
			\end{split}
		\end{equation}
		where we write $D_N(\bullet)=D_N(\bullet;\delta_1^{[m,n]},\ldots,\delta_N^{[m,n]},\eta_1^{[m,n]},\ldots,\eta_N^{[m,n]})$ for brevity.
		Furthermore, by removing the common factors of the numerator and denominator, we obtain
		\begin{equation}\label{Eq-N-breather-asymptotic-2}
			u_{N,L_k^-}(\xi,t)=\sqrt{\nu_0}C_{N,L_k^-}\frac{\sigma(\kappa)}{\sigma(\rho)}\frac{\tilde{W}_{N,L_k^{-}}^{(1)}}{\tilde{W}_{N,L_k^{-}}^{(2)}}e^{F(\xi,t)},
		\end{equation}
		where 
		\begin{equation}\label{Eq-U12-tilde}
			\begin{split}
				&	\tilde{W}_{N,L_k^{-}}^{(s)} = -\big(\tilde{W}_{N,L_k^-}^{[(s),1,1]}+\tilde{W}_{N,L_k^-}^{[(s),2,1]}e^{\tau_k}+\tilde{W}_{N,L_k^-}^{[(s),1,2]}e^{\tau_k^*} +\tilde{W}_{N,L_k^-}^{[(s),2,2]}e^{\tau_k+\tau_k^*}\big), \\
				&	\tilde{W}_{N,L_k^-}^{[(s),m,n]}=-\Sigma^{(s)}\Big(\xi-\sum_{i=1,i\neq k}^N(\delta_i^{[m,n]}+\eta_i^{[m,n]});\delta_k^{[m,n]},\eta_k^{[m,n]}\Big) A_k^{[m,n]}(\xi)B_k^{[m,n]},\quad s=1,2, \\
				&	B_k^{[m,n]} = \frac{ \prod_{i=1,i\neq k}^{N}\sigma(\delta_i^{[m,n]}-\delta_k^{[m,n]})\sigma(\eta_i^{[m,n]}-\eta_k^{[m,n]}) }{\prod_{i=1,\,i\neq k}^{N}\sigma(\delta_i^{[m,n]}+\eta_k^{[m,n]}) \sigma(\delta_k^{[m,n]}+\eta_i^{[m,n]})}.\\
			\end{split}
		\end{equation}
		Consequently, we obtain the representation for $C_{N,L_k^-}$ and $z_{N,L_k^-}$ presented in \eqref{Eq-asym-Lkm-2}. Moreover, we have 
		\begin{equation}
			\begin{split}
				&S_{N,L_K^-}=\frac{B_k^{[2,1]}}{B_k^{[1,1]}}=\frac{ \prod_{i=1,i\neq k}^{N}\sigma(\eta_i^{[2,1]}-\eta_k^{[2,1]}) }{\prod_{i=1,\,i\neq k}^{N}\sigma(\delta_i^{[2,1]}+\eta_k^{[2,1]})}\\
				&\xlongequal{\eqref{Eq-N-breather-asymptotic-3}}\prod_{i=1}^{k-1}\frac{\sigma(z_i-\hat{z}_k)\sigma(z_k+z_i^*)}{\sigma(\hat{z}_k+z_i^*)\sigma(z_i-z_k)}
				\prod_{i=k+1}^{N}\frac{\sigma(\hat{z}_k-\hat{z}_i)\sigma(z_k+\check{z}_i)}{\sigma(\hat{z}_k+\check{z}_i)\sigma(z_k-\hat{z}_i)}.
			\end{split}
		\end{equation}
		Therefore, the representation \eqref{Eq-N-breather-asymptotic-2} can be rewritten as \eqref{Eq-uN-asym}. The asymptotic analysis along $L_k^+$ as $t\rightarrow +\infty$ can be established analogously. Therefore we complete the proof. 
	\end{proof}
	
	The expression obtained along $L_k^{\pm}$ still carries the structure of an $N$-parameter determinant. The next theorem clarifies its true nature: each asymptotic profile is, in fact, a single first-order elliptic localized wave, obtained from the seed \eqref{Eq-FL-elliptic-solution} by a one-fold Darboux--B\"acklund transformation together with a spatial shift and a phase factor.
	
	\begin{theorem}[First-order reduction along the propagation directions $L_k^{\pm}$]\label{thm:first-order-reduction}
		The asymptotic solution along each propagation direction $L_{k}^{\pm}, k=1,2,\ldots,N$, is in fact a first-order elliptic localized wave.
	\end{theorem}
	
	\begin{proof}
		From \eqref{Eq-asym-Lkm-2}, we obtain
		\begin{equation}
			\begin{split}
				r^*(z_i)=
				\begin{dcases}
					r(-z_i^*), & \mathrm{Re}(\rho)=0,\\[6pt]
					-r(-z_i^*)e^{2\zeta(\omega_1)(-\rho+\omega_1+z_i^*)}, & \mathrm{Re}(\rho)=\omega_1.
				\end{dcases}
			\end{split}
		\end{equation}
		Consequently,
		\begin{equation}
			\begin{split}
				\big|C_{N,L_k^{\pm}}\big|
				=
				\begin{dcases}
					1, & \mathrm{Re}(\rho)=0,\\[6pt]
					e^{-\zeta(\omega_1)z_{N,L_k^{\pm}}},& \mathrm{Re}(\rho)=\omega_1.
				\end{dcases}
			\end{split}
		\end{equation}
		Hence we have
		\begin{equation}
			\big|C_{N,L_k^{\pm}}e^{(\zeta(\rho+\kappa)-\zeta(2\kappa))z_{N,L_k^{\pm}}}\big|=1.
		\end{equation}
		From \eqref{Eq-F}, it follows that
		\begin{equation}
			F(\xi,t)=F(\xi+z_{N,L_k^{\pm}},t)+\big(\zeta(\rho+\kappa)-\zeta(2\kappa)\big)z_{N,L_k^{\pm}}.
		\end{equation}
		Moreover, replacing
		$\alpha_{N,L_k^{\pm}}$, $E(\xi,t;z_k)$ and $E(\xi,t;\hat{z}_k)$ respectively by
		$\alpha_{N,L_k^{\pm}}e^{(\zeta(2\kappa+\rho+z_k)-\zeta(\kappa-z_k))z_{N,L_k^{\pm}}}$, 
		$E(\xi+z_{N,L_k^{\pm}},t;z_k)$ and $E(\xi+z_{N,L_k^{\pm}},t;\hat{z}_k)$, 
		leaves the asymptotic solution given by \eqref{Eq-uN-asym}--\eqref{Eq-asym-Lkm-2} invariant.
		
		Therefore, this solution is indeed a first-order elliptic localized solution derived from \eqref{Eq-FL-elliptic-solution} by means of the one-fold Darboux--B\"acklund transformation at $\lambda_k=\lambda(z_k)$, accompanied by the spatial shift $\xi\mapsto \xi+z_{N,L_k^{\pm}}$ and the phase factor $C_{N,L_k^{\pm}}e^{(\zeta(\rho+\kappa)-\zeta(2\kappa))z_{N,L_k^{\pm}}}$. This completes the proof.
	\end{proof}
	
	Having described the solution along the propagation directions, we now turn to its behavior in the intermediate regions $R_k^{\pm}$, where, as the next theorem shows, the solution reduces to a shifted copy of the elliptic seed.

	\begin{theorem}[The asymptotic solution in $R_{k}^{\pm}$]\label{thm:AB2}
		Assume $\operatorname{Re}(\beta(z_i))>0$ for $i=1,2,\ldots,N$ and $v(z_i)<v(z_j)$ for $1\leq i<j\leq N$. 
		Then the asymptotic solutions $u_{N,R_k^\pm}$ of the $N$-elliptic localized solution \eqref{Eq-uN} in the regions $R_k^{\pm}$ are given by
		\begin{equation}\label{Eq-asymptotic-Rk}
			u_{N,R_k^\pm}= \sqrt{\nu_0}C_{N,R_k^\pm}\frac{\sigma(\kappa)\sigma(\rho+\xi+z_{N,R_k^{\pm}})}{\sigma(\rho)\sigma(\kappa+\xi+z_{N,R_k^{\pm}})}e^{F(\xi,t)}+\mathcal{O}\big(\mp\exp\big(\min_{i\neq k}\big|\beta(z_i)\big(v(z_i)-v(z_k)\big)\big|t\big)\big),
		\end{equation}
		where
		\begin{equation}\label{Eq-asymptotic-Rk-para}
			\begin{split}
				C_{N,R_k^-}&=\prod_{i=1}^{k-1}\frac{r(z_i)}{r(-z_i^*)}\prod_{i=k}^{N}\frac{r(-z_i^*)}{r(z_i)},\qquad 
				C_{N,R_k^+}=(C_{N,R_k^-})^{-1},\\[2mm]
				z_{N,R_k^-}&=-\sum_{i=1}^{k-1}(z_i+z_i^*)+\sum_{i=k}^{N}(z_i+z_i^*),\qquad 
				z_{N,R_k^+}=-z_{N,R_k^-}.
			\end{split}
		\end{equation}
	\end{theorem}
	\begin{proof}
		We begin by rewriting the $N$-elliptic localized solution in the form \eqref{Eq-uN-Sherman-rewritten}. For clarity of exposition, we restrict the presentation to the case $\mathrm{Re}(\rho)=0$; the case $\mathrm{Re}(\rho)=\omega_1$ follows by an identical computation using the quasi-periodicity \eqref{Eq-quasi-periodic-sigma} of the $\sigma$-function, and is therefore omitted.
		
		Consider first the region $R_k^{-}$, where the velocity $v(z)$ satisfies $v(z_{k-1})<v(z)<v(z_k)$. In this region, the determinant-ratio representation \eqref{Eq-asym-solution}--\eqref{Eq-Aij-entries} remains valid, with $\mathbf{X}_{1,2}$ now replaced by
		\begin{align}
			\mathbf{X}_1 = \mathrm{diag}(1,\ldots,1,e^{-\tau_k},\ldots,e^{-\tau_N}), \quad 
			\mathbf{X}_2 = \mathrm{diag}(e^{\tau_1},\ldots,e^{\tau_{k-1}},1,\ldots,1).
		\end{align}
		Consequently, as $t\rightarrow -\infty$, we have
		\begin{align}
			\mathbf{X}_1\rightarrow\mathrm{diag}\Big(
			\overbrace{1,...,1}^{k-1},0,...,0
			\Big)
			,\quad  	\mathbf{X}_2\rightarrow\mathrm{diag}\Big(
			0,...,0,	\overbrace{1,...,1}^{N-k+1}
			\Big).
		\end{align}
		Substituting these limits into the representation yields the reduced asymptotic expression
		\begin{equation}\label{Eq-reduced-expression}
			\begin{split}
				u_N= u_{N,R_k^-}+\mathcal{O}\big(\exp\big(\mathrm{min}_{i\neq k}\big|\beta(z_i)\big(v(z_i)-v(z_k)\big)\big|t\big)\big),
			\end{split}
		\end{equation}
		where 
		\begin{equation}
			\begin{split}		
				u_{N,R_{k}^{-}}&=\sqrt{\nu_0}C_{N,L_k^-}\frac{\sigma(\kappa)}{\sigma(\rho)}\frac{\tilde{W}_{N,L_k^-}^{[(1),2,2]}}{\tilde{W}_{N,L_k^-}^{[(2),2,2]}}e^{F(\xi,t)}\\
				&\xlongequal{\eqref{Eq-U12-tilde}}\sqrt{\nu_0}C_{N,L_k^-}\frac{\sigma(\kappa)}{\sigma(\rho)}\frac{\Sigma^{(2)}\big(\xi-\sum_{i=1,i\neq k}^N(\delta_i^{[2,2]}+\eta_i^{[2,2]});\delta_i^{[2,2]},\eta_i^{[2,2]}\big)}{\Sigma^{(1)}\big(\xi-\sum_{i=1,i\neq k}^N(\delta_i^{[1,1]}+\eta_i^{[1,1]});\delta_i^{[1,1]},\eta_i^{[1,1]}\big)}e^{F(\xi,t)}\\
				&\xlongequal{\eqref{Eq-Sigma-i}}	\sqrt{\nu_0}C_{N,L_k^-}\frac{\sigma(\kappa)}{\sigma(\rho)}\frac{\sigma\big(\xi+\rho-\sum_{i=1}^N(\delta_i^{[2,2]}+\eta_i^{[2,2]})\big)\sigma(\check{z}_k-\kappa)\sigma(\hat{z}_k+\rho)}{\sigma\big(\xi+\kappa-\sum_{i=1}^N(\delta_i^{[2,2]}+\eta_i^{[2,2]})\big)\sigma(\check{z}_k-\rho)\sigma(\hat{z}_k+\kappa)}e^{F(\xi,t)},
			\end{split}
		\end{equation}
		which coincides with the representation \eqref{Eq-asymptotic-Rk} upon setting
		\begin{equation}
			\begin{split}
				z_{N,R_k^-}=-\sum_{i=1}^N(\delta_i^{[2,2]}+\eta_i^{[2,2]}),\quad C_{N,R_k^-}=C_{N,L_k^-}\frac{\sigma(\check{z}_k-\kappa)\sigma(\hat{z}_k+\rho)}{\sigma(\check{z}_k-\rho)\sigma(\hat{z}_k+\kappa)}.
			\end{split}
		\end{equation}
		The asymptotic formula in the region $R_k^+$ is obtained by an entirely analogous argument. This completes the proof.
	\end{proof}
	\begin{rmk}
		In particular, \eqref{Eq-asymptotic-Rk} shows that the asymptotic solutions in the regions $R_{k}^{\pm},k=1,2,\ldots,N$, are shifted copies of the elliptic solution \eqref{Eq-FL-elliptic-solution}.
	\end{rmk}

	Theorems~\ref{thm:AB1} and \ref{thm:AB2} together show that the $N$-elliptic localized solution separates, as $t\to\pm\infty$, into individual first-order elliptic localized waves over a shifted elliptic background, so that the collisions are elastic. Previous studies~\cite{ling2023multi,lingtang2026dnls} established that the $N$-elliptic localized solutions of the mKdV and DNLS equations are, in addition, symmetric with respect to the origin under specific conditions. We now extend this symmetry to the FL equation, thereby promoting the elastic collisions to strictly elastic ones whenever the solution is symmetric about the origin.

	\begin{theorem}[The symmetry property of the $N$-elliptic localized solutions]\label{thm:symmetry}
		When $\alpha_i=1,i=1,2,\ldots,N$, the $N$-elliptic localized solution $u_N$ satisfies the symmetry relation
		\begin{align}\label{Eq-uN-symmetry}
			u_N(\xi,t)= u_N^{*}(-\xi,-t).
		\end{align}
	\end{theorem}
	\begin{proof}
		Using \eqref{Eq-F} and \eqref{Eq-E}, we can verify that
		\begin{align}
			E(-\xi,-t;z)=E(\xi,t;\hat{z})e^{-F(\xi,t)}
		\end{align}
		holds for arbitrary $z$. Selecting $\mathbf{c}_i=(1,1)^\mathsf{T}$, we obtain the symmetry
		\begin{align}\label{Eq-symmetry-entry}
			\mathbf{\Phi}_i(-\xi,-t)\xlongequal[\eqref{Eq-Phi-i}]{\eqref{Eq-fundamental-solution-matrix}}\begin{pmatrix}
				0 & -\mathrm{i}\\
				\mathrm{i} & 0
			\end{pmatrix}\mathbf{\Phi}(\xi,t;\lambda_i)\begin{pmatrix}
				0 & 1\\
				1 & 0
			\end{pmatrix}\begin{pmatrix}
				1\\
				1
			\end{pmatrix}\xlongequal{\eqref{Eq-Phi-i}}\begin{pmatrix}
				0 & -\mathrm{i}\\
				\mathrm{i} & 0
			\end{pmatrix}	\mathbf{\Phi}_i(\xi,t),
		\end{align}
		which leads to
		\begin{equation}
			\begin{split}
				\mathbf{\Phi}_j^\dagger(\xi,t)\mathbf{\Phi}_i(\xi,t)&=\mathbf{\Phi}_j^\dagger(-\xi,-t)\mathbf{\Phi}_i(-\xi,-t),\\
				\mathbf{\Phi}_j^\dagger(\xi,t)\bm{\sigma}_3\mathbf{\Phi}_i(\xi,t)&=-\mathbf{\Phi}_j^\dagger(-\xi,-t)\bm{\sigma}_3\mathbf{\Phi}_i(-\xi,-t).
			\end{split}
		\end{equation}	
		Together with \eqref{Eq-B}, we obtain 
		\begin{equation}\label{Eq-M-symmetry}
			\mathbf{B}(\xi,t)=-\mathbf{B}^\dagger(-\xi,-t).
		\end{equation}
		Similarly, it follows directly from \eqref{Eq-symmetry-entry} that 
		\begin{equation}
			\mathbf{\Phi}_i^{(1)}\big(\mathbf{\Phi}_j^{(2)}\big)^*(\xi,t)=	-\mathbf{\Phi}_i^{(2)}\big(\mathbf{\Phi}_j^{(1)}\big)^*(-\xi,-t).
		\end{equation}
		As a consequence, we have 
		\begin{equation}\label{Eq-phi-symmetry}
			\big(\mathbf{\Phi}^{(1)}\big(\mathbf{\Phi}^{(2)}\big)^\dagger\big)(\xi,t)=-\big(\mathbf{\Phi}^{(1)}\big(\mathbf{\Phi}^{(2)}\big)^\dagger\big)^\dagger(-\xi,-t).
		\end{equation}
		Using the explicit form of the elliptic solution \eqref{Eq-FL-elliptic-solution}, we arrive at 
		\begin{equation}
			u(\xi,t)=u^*(-\xi,-t).
		\end{equation}
		Combining \eqref{Eq-uN-Sherman}, \eqref{Eq-M-symmetry} and \eqref{Eq-phi-symmetry}, we finally verify \eqref{Eq-uN-symmetry} and complete the proof.
	\end{proof}
\subsection{The asymptotic solutions for $N=1$ and $N=2$}

In the previous subsections, we have obtained the asymptotic formulae for the general
\(N\)-elliptic localized solution. In this subsection, we write out the cases \(N=1\)
and \(N=2\) explicitly. This allows us to see more directly how the phase factors and
the shifts of the elliptic backgrounds appear in the low-order cases.

\noindent\textit{The case \(N=1\).}
We first consider the first-order case. For \(N=1\), there is only one localized core,
and hence the solution has two asymptotic states in the two regions on its sides.
In terms of \(\tau_1\), we have
\begin{equation}\label{Eq-u1-R-pm-asymptotic}
	u_1(\xi,t)
	=
	u_{1,R_1^\pm}(\xi,t)
	+
	O\left(e^{\pm\operatorname{Re}\tau_1}\right),
	\qquad
	\operatorname{Re}\tau_1\to\mp\infty, 
\end{equation}
where
\begin{equation}\label{Eq-u1-R-pm-explicit} 
	\begin{split}
		u_{1,R_1^\pm}(\xi,t)
		={}&
		\sqrt{\nu_0}\,
		\left(
		\frac{
			\sigma(\kappa-z_1^*)\sigma(\rho+z_1)
		}{
			\sigma(\rho-z_1^*)\sigma(\kappa+z_1)
		}
		\right)^{\pm1}
		\frac{
			\sigma(\kappa)\sigma\big(\xi+\rho\mp(z_1+z_1^*)\big)
		}{
			\sigma(\rho)\sigma\big(\xi+\kappa\mp(z_1+z_1^*)\big)
		}
		e^{F(\xi,t)}. 
	\end{split} 
\end{equation} 
This shows that the one-elliptic localized solution connects two shifted elliptic
backgrounds. The two shifts are \(z_1+z_1^*\) and \(-z_1-z_1^*\), and the corresponding
phase factors are
$
\frac{
	\sigma(\rho-z_1^*)\sigma(\kappa+z_1)
}{
	\sigma(\kappa-z_1^*)\sigma(\rho+z_1)
}
$
and
$
\frac{
	\sigma(\kappa-z_1^*)\sigma(\rho+z_1)
}{
	\sigma(\rho-z_1^*)\sigma(\kappa+z_1)
},
$
respectively.

\noindent\textit{The case \(N=2\).}
We next consider the second-order case. Assume
$
\operatorname{Re}\beta(z_i)>0,\quad i=1,2,
$
and
$
v(z_1)<v(z_2).
$
Under this ordering of velocities, the two localized cores separate as
\(t\to\pm\infty\). There are two propagation directions in each temporal direction
and four intermediate regions.

We first write the asymptotic formulae along the propagation directions. In these
regions, one of \(\operatorname{Re}\tau_1\) and \(\operatorname{Re}\tau_2\) remains
bounded, while the other tends to infinity with a prescribed sign. The solution is
then approximated by a first-order elliptic localized wave with shifted parameters.
Thus we have
\begin{equation}\label{Eq-u2-L1-pm-asymptotic}
	u_2(\xi,t)
	=
	u_{2,L_1^\pm}(\xi,t)
	+
	O\left(e^{\pm\operatorname{Re}\tau_2}\right),
	\qquad
	\operatorname{Re}\tau_1=O(1),\quad
	\operatorname{Re}\tau_2\to\mp\infty,
\end{equation}
and
\begin{equation}\label{Eq-u2-L2-pm-asymptotic}
	u_2(\xi,t)
	=
	u_{2,L_2^\pm}(\xi,t)
	+
	O\left(e^{\mp\operatorname{Re}\tau_1}\right),
	\qquad
	\operatorname{Re}\tau_1\to\pm\infty,\quad
	\operatorname{Re}\tau_2=O(1).
\end{equation}
For \(k=1,2\), the asymptotic profiles are given by
\begin{equation}\label{Eq-u2-Lk-pm-profile}
	u_{2,L_k^\pm}(\xi,t)
	=
	\sqrt{\nu_0}\,
	C_{2,L_k^\pm}
	\frac{\sigma(\kappa)}{\sigma(\rho)}
	\frac{V_{1,2,L_k^\pm}(\xi,t)}
	{V_{2,2,L_k^\pm}(\xi,t)}
	e^{F(\xi,t)},
\end{equation}
where, for \(s=1,2\),
\begin{equation}\label{Eq-u2-Lk-pm-Vs}
	\begin{split}
		&V_{s,2,L_k^\pm}(\xi,t)
		={}
		I_k
		\Sigma^{(s)}(\xi+z_{2,L_k^\pm};z_k^*,z_k)
		d_0^*(\hat z_k)d_0(z_k)
		E^*(\xi,t;z_k)E(\xi,t;z_k)
		\\
		&+\alpha_{2,L_k^\pm} I_k
		\Sigma^{(s)}(\xi+z_{2,L_k^\pm};z_k^*,\hat z_k)
		d_0^*(\hat z_k)d_0(\hat z_k)
		E^*(\xi,t;z_k)E(\xi,t;\hat z_k)
		\\
		&-\alpha_{2,L_k^\pm}^* I_0(\xi+z_{2,L_k^\pm})
		\Sigma^{(s)}(\xi+z_{2,L_k^\pm};\check z_k,z_k)
		d_0^*(z_k)d_0(z_k)
		E^*(\xi,t;\hat z_k)E(\xi,t;z_k)
		\\
		&-|\alpha_{2,L_k^\pm}|^2 I_0(\xi+z_{2,L_k^\pm})
		\Sigma^{(s)}(\xi+z_{2,L_k^\pm};\check z_k,\hat z_k)
		d_0^*(z_k)d_0(\hat z_k)
		E^*(\xi,t;\hat z_k)E(\xi,t;\hat z_k).
	\end{split}
\end{equation}
Here the shifts, phase factors, and transformed parameters are given by
\begin{equation}\label{Eq-u2-L1-pm-data}
	\begin{split}
		z_{2,L_1^\pm}
		&=\mp(z_2+z_2^*),\\
		C_{2,L_1^\pm}
		&=
		\left(
		\frac{
			\sigma(\kappa-z_2^*)\sigma(\rho+z_2)
		}{
			\sigma(\rho-z_2^*)\sigma(\kappa+z_2)
		}
		\right)^{\pm1},\\
		\alpha_{2,L_1^\pm}
		&=
		\alpha_1
		\left(
		\frac{
			\sigma(z_1+z_2^*)\sigma(z_1+\kappa+\rho+z_2)
		}{
			\sigma(z_1-z_2)\sigma(z_1+\kappa+\rho-z_2^*)
		}
		\right)^{\pm1},
	\end{split}
\end{equation}
and
\begin{equation}\label{Eq-u2-L2-pm-data}
	\begin{split}
		z_{2,L_2^\pm}
		&=\pm(z_1+z_1^*),\\
		C_{2,L_2^\pm}
		&=
		\left(
		\frac{
			\sigma(\kappa+z_1)\sigma(\rho-z_1^*)
		}{
			\sigma(\rho+z_1)\sigma(\kappa-z_1^*)
		}
		\right)^{\pm1},\\
		\alpha_{2,L_2^\pm}
		&=
		\alpha_2
		\left(
		\frac{
			\sigma(z_1^*-\kappa-\rho-z_2)\sigma(z_1-z_2)
		}{
			\sigma(z_1+\kappa+\rho+z_2)\sigma(z_2+z_1^*)
		}
		\right)^{\pm1},
	\end{split}
\end{equation}
respectively.

We now turn to the intermediate regions. In these regions, both localized cores are
away from the observation point, and the leading term is a shifted elliptic background.
More precisely,
\begin{equation}\label{Eq-u2-R1-pm-asymptotic}
	u_2(\xi,t)
	=
	u_{2,R_1^\pm}(\xi,t)
	+
	O\left(e^{\pm\operatorname{Re}\tau_1}
	+e^{\pm\operatorname{Re}\tau_2}\right),
	\qquad
	\operatorname{Re}\tau_1\to\mp\infty,\quad
	\operatorname{Re}\tau_2\to\mp\infty,
\end{equation}
and
\begin{equation}\label{Eq-u2-R2-pm-asymptotic}
	u_2(\xi,t)
	=
	u_{2,R_2^\pm}(\xi,t)
	+
	O\left(e^{\mp\operatorname{Re}\tau_1}
	+e^{\pm\operatorname{Re}\tau_2}\right),
	\qquad
	\operatorname{Re}\tau_1\to\pm\infty,\quad
	\operatorname{Re}\tau_2\to\mp\infty.
\end{equation}
The corresponding regional asymptotic states are
\begin{equation}\label{Eq-u2-R1-pm-explicit}
	\begin{split}
		u_{2,R_1^\pm}(\xi,t)
		={}&
		\sqrt{\nu_0}
		\prod_{j=1}^{2}
		\left(
		\frac{
			\sigma(\kappa-z_j^*)\sigma(\rho+z_j)
		}{
			\sigma(\rho-z_j^*)\sigma(\kappa+z_j)
		}
		\right)^{\pm1}
		\frac{
			\sigma(\kappa)
			\sigma\Big(\xi+\rho\mp\sum_{j=1}^{2}(z_j+z_j^*)\Big)
		}{
			\sigma(\rho)
			\sigma\Big(\xi+\kappa\mp\sum_{j=1}^{2}(z_j+z_j^*)\Big)
		}
		e^{F(\xi,t)},
	\end{split}
\end{equation}
and
\begin{equation}\label{Eq-u2-R2-pm-explicit}
	\begin{split}
		u_{2,R_2^\pm}(\xi,t)
		={}&
		\sqrt{\nu_0}
		\prod_{j=1}^{2}
		\left(
		\frac{
			\sigma(\kappa-z_j^*)\sigma(\rho+z_j)
		}{
			\sigma(\rho-z_j^*)\sigma(\kappa+z_j)
		}
		\right)^{\pm(-1)^j}
		\frac{
			\sigma(\kappa)
			\sigma\Big(\xi+\rho+\sum_{j=1}^{2}\mp(-1)^j(z_j+z_j^*)\Big)
		}{
			\sigma(\rho)
			\sigma\Big(\xi+\kappa+\sum_{j=1}^{2}\mp(-1)^j(z_j+z_j^*)\Big)
		}
		e^{F(\xi,t)}.
	\end{split}
\end{equation}
Therefore, the four intermediate asymptotic states are shifted elliptic backgrounds.
Their shifts are \(\mp(z_1+z_1^*)\mp(z_2+z_2^*)\) in \(R_1^\pm\), and
\(\pm(z_1+z_1^*)\mp(z_2+z_2^*)\) in \(R_2^\pm\).
	\section{Dynamic behaviors of the one- and two-elliptic localized solutions}\label{sec:dynamics}
	To gain further insight into the solutions built above and into the analytical conclusions of the preceding sections, we now examine the concrete dynamics of the one- and two-elliptic localized solutions of the FL equation, complementing the analytical statements with graphical illustrations. We first portray the one-elliptic localized solution and confirm the reflection symmetry guaranteed by Theorem~\ref{thm:symmetry}; we then pass to the two-elliptic localized solution and verify, along the propagation directions as well as inside the intermediate regions, the long-time predictions furnished by Theorems~\ref{thm:AB1}, \ref{thm:first-order-reduction} and \ref{thm:AB2}.
	
	\subsection{The one- and two-elliptic localized solutions}
Putting $N=1$ in \eqref{Eq-uN} yields the one-elliptic localized solution. Figure~\ref{Fig-first-order-solution-1} presents the modulus $|u_1|$ of non-stationary one-elliptic localized solutions with parameters $\kappa=0.99\mathrm{i}$, $\rho=1.41-1.49\mathrm{i}$, $\omega_1=1.41$, $\omega_3=-1.69\mathrm{i}$, and $z_1=1+\mathrm{i}$, for which the corresponding spectral value obtained from \eqref{Eq-parameterization-lambda-y-0} is $\lambda_1=\lambda(z_1)=-0.53-0.22\mathrm{i}$. Using \eqref{Eq-vi}, both panels in Figure~\ref{Fig-first-order-solution-1} show an elliptic localized solution travelling at speed $v(z_1)=1.27$. With $\alpha_1=1$, the left panel illustrates a solution symmetric with respect to the origin. That is, it is invariant under the reflection $u_1(\xi,t)=u_1^{*}(-\xi,-t)$, as ensured by Theorem~\ref{thm:symmetry}. By contrast, the right panel, with $\alpha_1=0.1$, loses this invariance.
	
	\begin{figure}[H]
		\centering
		\begin{subfigure}{0.45\textwidth}
			\includegraphics[width=\textwidth]{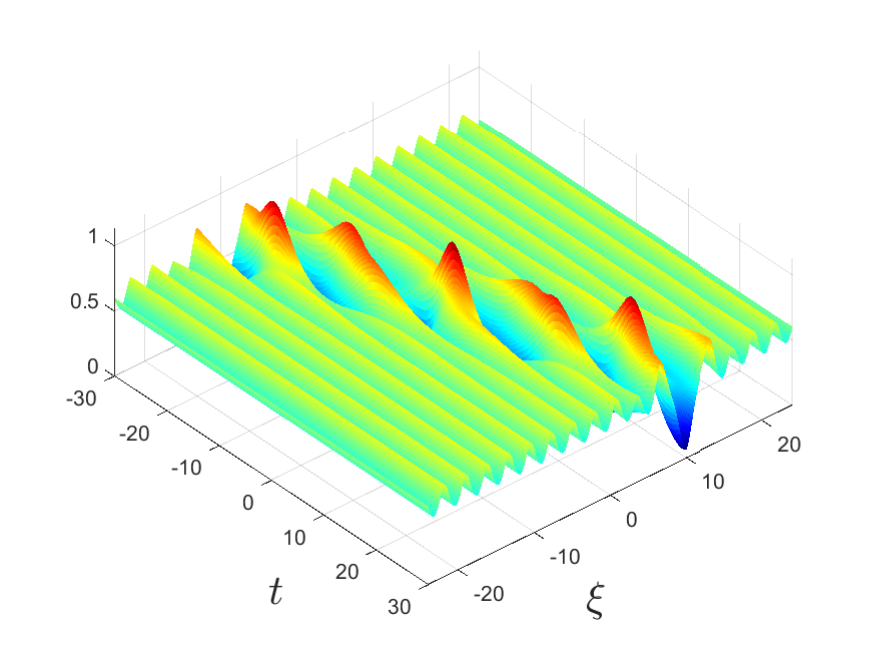}
		\end{subfigure}
		\begin{subfigure}{0.45\textwidth}
			\includegraphics[width=\textwidth]{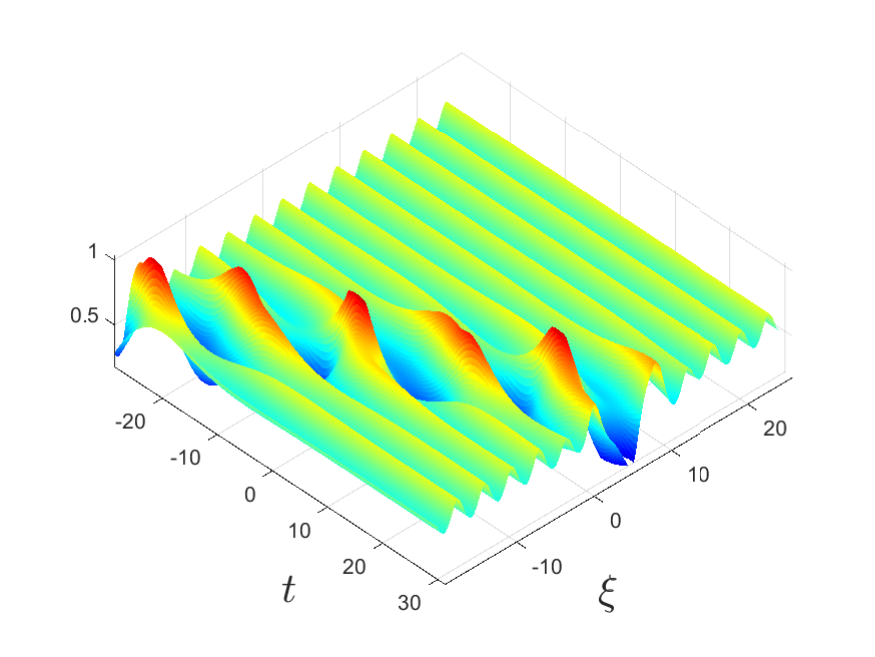}
		\end{subfigure}
		\caption{The non-stationary one-elliptic localized solution $|u_1|$ of the FL equation with $\kappa=0.99\mathrm{i}$, $\rho=1.41-1.49\mathrm{i}$, $\omega_1=1.41$, $\omega_3=-1.69\mathrm{i}$ and $z_1=1+\mathrm{i}$. Left: $\alpha_1=1$, for which the solution is symmetric about the origin as predicted by Theorem~\ref{thm:symmetry}. Right: $\alpha_1=0.1$, for which the symmetry is broken.}
		\label{Fig-first-order-solution-1}
	\end{figure}
	
Figure~\ref{Fig-first-order-solution-2} shows the modulus $|u_1|$ of stationary one-elliptic localized solutions for $\kappa=2.08\mathrm{i}$, $\rho=3.25-1.73\mathrm{i}$, $\omega_1=3.25$, $\omega_3=-3.31\mathrm{i}$, and $z_1=2.53+3\mathrm{i}$, whose stationary nature is verified by the direct evaluation $v(z_1)=0$ using \eqref{Eq-vi}. Using \eqref{Eq-tau-i}, the period of these stationary solutions along the temporal direction can be evaluated as
\begin{equation}
	T=\frac{4\pi}{|\mathrm{Re}(y(z_1))|}=60,
\end{equation}
which is consistent with Figure~\ref{Fig-second-order-solution}. In the left panel, we select $\alpha_1=1$, and therefore the profile is symmetric with respect to the origin in accordance with Theorem~\ref{thm:symmetry}. Moreover, $|u_1|$ in the left panel reaches its peak value $1.08$ at $(\xi,t)=\bigl(0,nT\bigr)$ and drops to its minimum $0.28$ at $(\xi,t)=\bigl(0,\frac{2n+1}{2}T\bigr)$, $n\in\mathbb{Z}$. Setting $\alpha_1=0.1$ in the right panel destroys the symmetry with respect to the origin.

Setting $N=2$ in \eqref{Eq-uN} yields the two-elliptic localized solution. By selecting $\kappa=2.08\mathrm{i}$, $\rho=3.25-1.73\mathrm{i}$, $\omega_1=3.25$, $\omega_3=-3.31\mathrm{i}$, $z_1=2.53+3\mathrm{i}$, and $\alpha_1=\alpha_2=1$, we display in the left panel of Figure~\ref{Fig-second-order-solution} a two-elliptic localized solution with one localized wave stationary. In the right panel, we obtain another two-elliptic localized solution with parameters $\kappa=1.02\mathrm{i}$, $\rho=1.29-0.87\mathrm{i}$, $\omega_1=1.29$, $\omega_3=-1.98\mathrm{i}$, $z_1=-0.24+1.2\mathrm{i}$, $z_2=-0.95+\mathrm{i}$, and $\alpha_1=\alpha_2=1$, for which both elliptic localized waves are non-stationary. Both solutions share symmetry with respect to the origin since $\alpha_1=\alpha_2=1$, as guaranteed by Theorem~\ref{thm:symmetry}.

	\begin{figure}[H]
		\centering
		\begin{subfigure}{0.45\textwidth}
			\includegraphics[width=\textwidth]{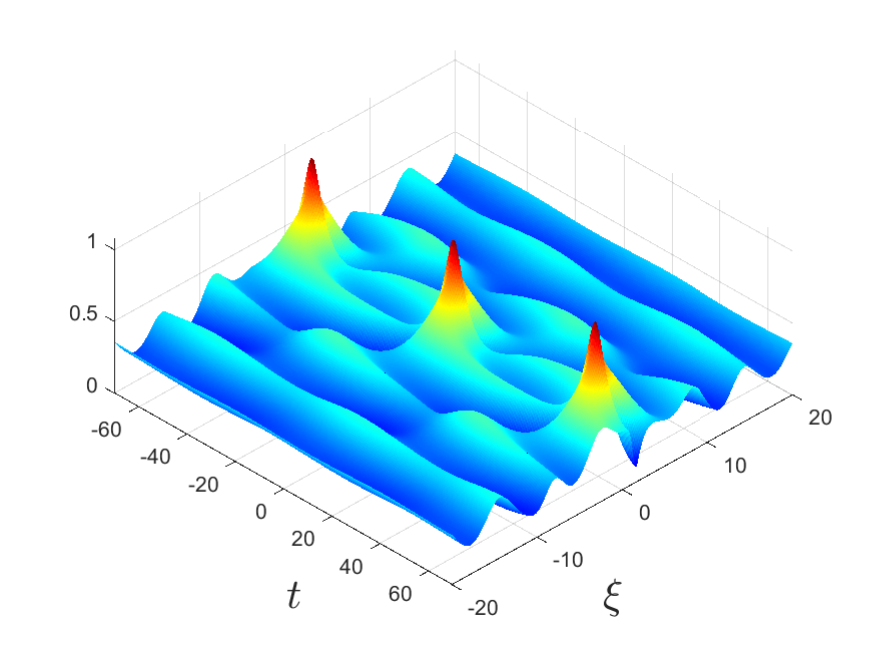}
		\end{subfigure}
		\begin{subfigure}{0.45\textwidth}
			\includegraphics[width=\textwidth]{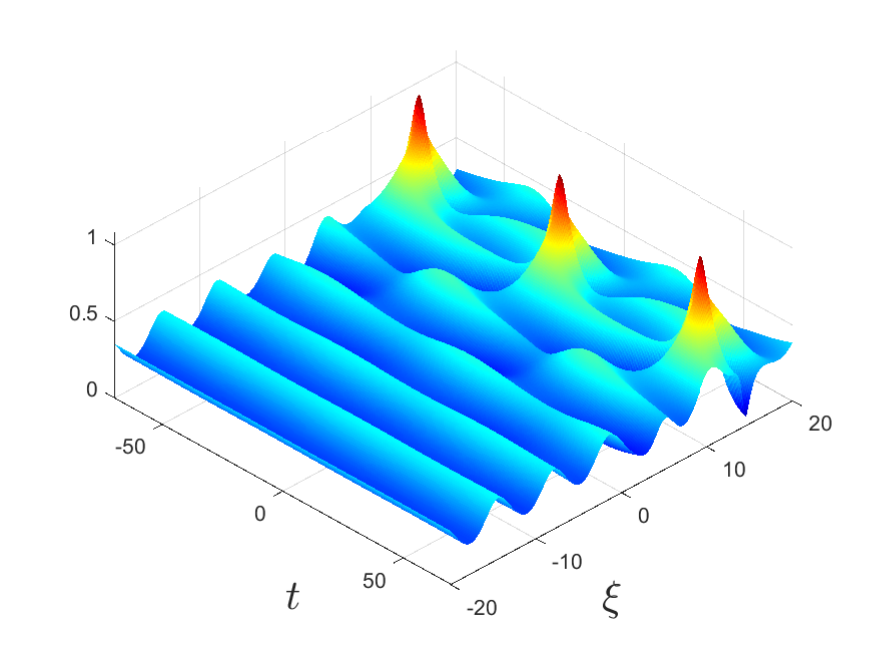}
		\end{subfigure}
		\caption{The stationary one-elliptic localized solution $|u_1|$ of the FL equation with $\kappa=2.08\mathrm{i}$, $\rho=3.25-1.73\mathrm{i}$, $\omega_1=3.25$, $\omega_3=-3.31\mathrm{i}$ and $z_1=2.53+3\mathrm{i}$. Left: $\alpha_1=1$, for which the solution is symmetric about the origin as predicted by Theorem~\ref{thm:symmetry}. Right: $\alpha_1=0.1$, for which the symmetry is broken.}
		\label{Fig-first-order-solution-2}
	\end{figure}

		\begin{figure}[H]
		\centering
		\begin{subfigure}{0.45\textwidth}
			\includegraphics[width=\textwidth]{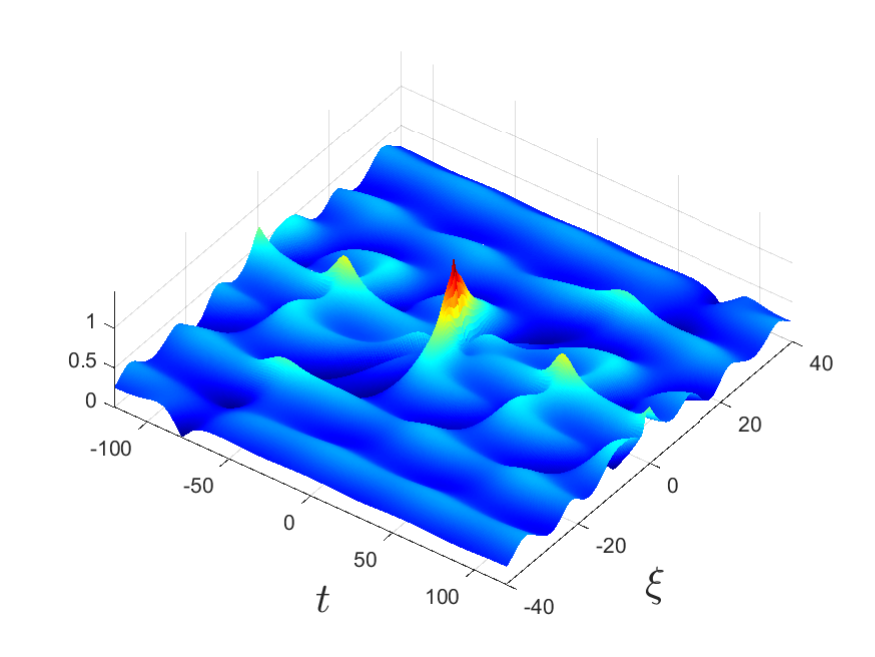}
		\end{subfigure}
			\begin{subfigure}{0.45\textwidth}
			\includegraphics[width=\textwidth]{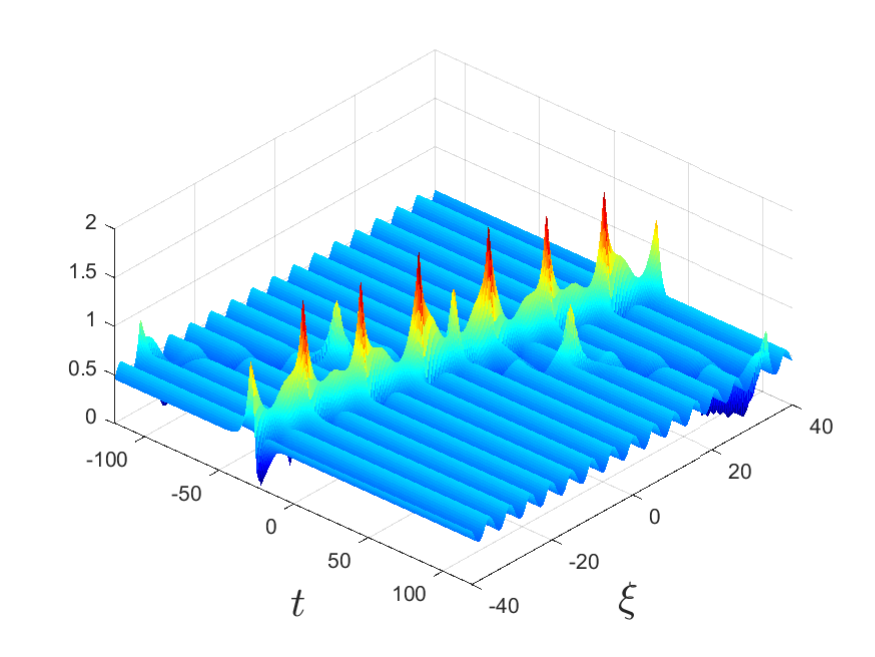}
		\end{subfigure}
		\caption{The two-elliptic localized solution $|u_2|$ of the FL equation.  Left: $\kappa=2.08\mathrm{i}$, $\rho=3.25-1.73\mathrm{i}$, $\omega_1=3.25$, $\omega_3=-3.31\mathrm{i}$ and $z_1=2.53+3\mathrm{i},\alpha_1=\alpha_2=1$. Right: $\kappa=1.02\mathrm{i}$, $\rho=1.29-0.87\mathrm{i}$, $\omega_1=1.29$, $\omega_3=-1.98\mathrm{i}$ and $z_1=-0.24+1.2\mathrm{i},z_2=-0.95+\mathrm{i},\alpha_1=\alpha_2=1$.}
		\label{Fig-second-order-solution}
	\end{figure}

	\subsection{The asymptotic analysis of the two-elliptic localized solution}
 Throughout this subsection, the parameters are fixed as
 $\kappa=1.02\mathrm{i}, \rho=1.29-0.87\mathrm{i}, \omega_1=1.29, \omega_3=-1.98\mathrm{i}, z_1=-0.34+0.95\mathrm{i}, z_2=-0.69+0.85\mathrm{i}$,
 together with $\alpha_1=\alpha_2=1$. By Theorem~\ref{thm:AB1}, the solution propagates along the two pairs of lines $L_1^{\pm}$ and $L_2^{\pm}$, whose slopes follow from the speeds $v(z_1)=0.15$ and $v(z_2)=0.3$ supplied by \eqref{Eq-vi}. The two localized waves thus separate and travel at distinct velocities, with the intermediate regions $R_1^{\pm}$ and $R_2^{\pm}$ lying between these directions. Since $\alpha_1=\alpha_2=1$, Theorem~\ref{thm:symmetry} forces the solution to be symmetric with respect to the origin, so that the propagation directions pass through the origin simultaneously and the collision becomes strictly elastic. We confirm the predicted asymptotic behaviours both inside the regions and along the lines.
	
	\paragraph{Asymptotics in the intermediate regions}
According to Theorem~\ref{thm:AB2}, inside each region $R_k^{\pm}$, the two-elliptic localized solution should degenerate into a shifted version of the background solution \eqref{Eq-FL-elliptic-solution}. This is illustrated in Figures~\ref{Fig-Asym-second-order-region1} and~\ref{Fig-Asym-second-order-region2}, recorded at $t=-70$ and $t=70$, respectively: in each panel, the modulus $|u_2|$ (blue solid line) of the exact solution is overlaid on the corresponding intervals from \eqref{Eq-asymptotic-Rk} together with the asymptotic solutions (red dash-dotted lines).
	
	\begin{figure}[H]
		\centering
		\begin{subfigure}{0.32\textwidth}
			\includegraphics[width=\textwidth]{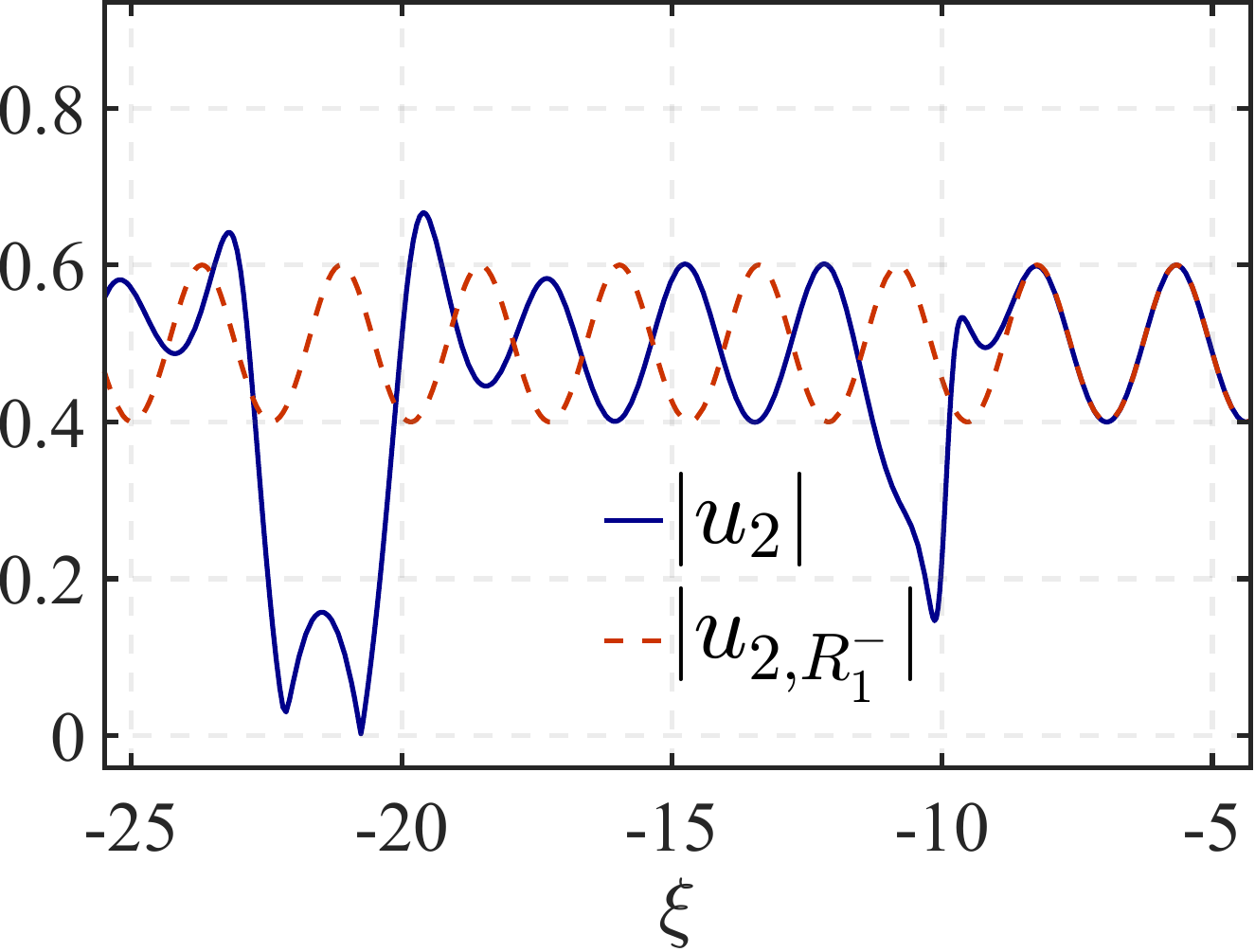}
			\caption{$|u_{2,R_1^+}|$}
		\end{subfigure}
		\begin{subfigure}{0.32\textwidth}
			\includegraphics[width=\textwidth]{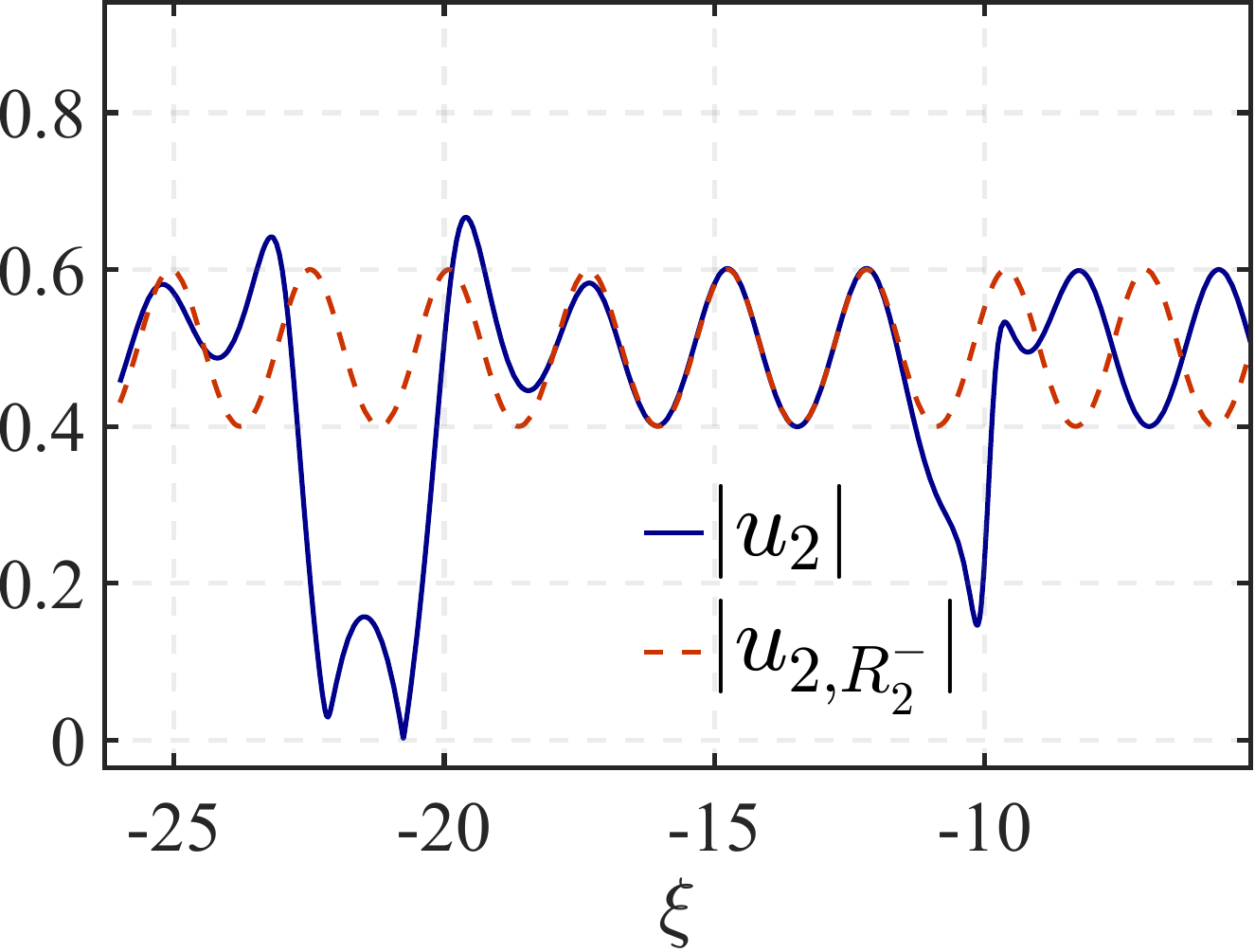}
			\caption{$|u_{2,R_2^-}|$}
		\end{subfigure}
		\begin{subfigure}{0.32\textwidth}
			\includegraphics[width=\textwidth]{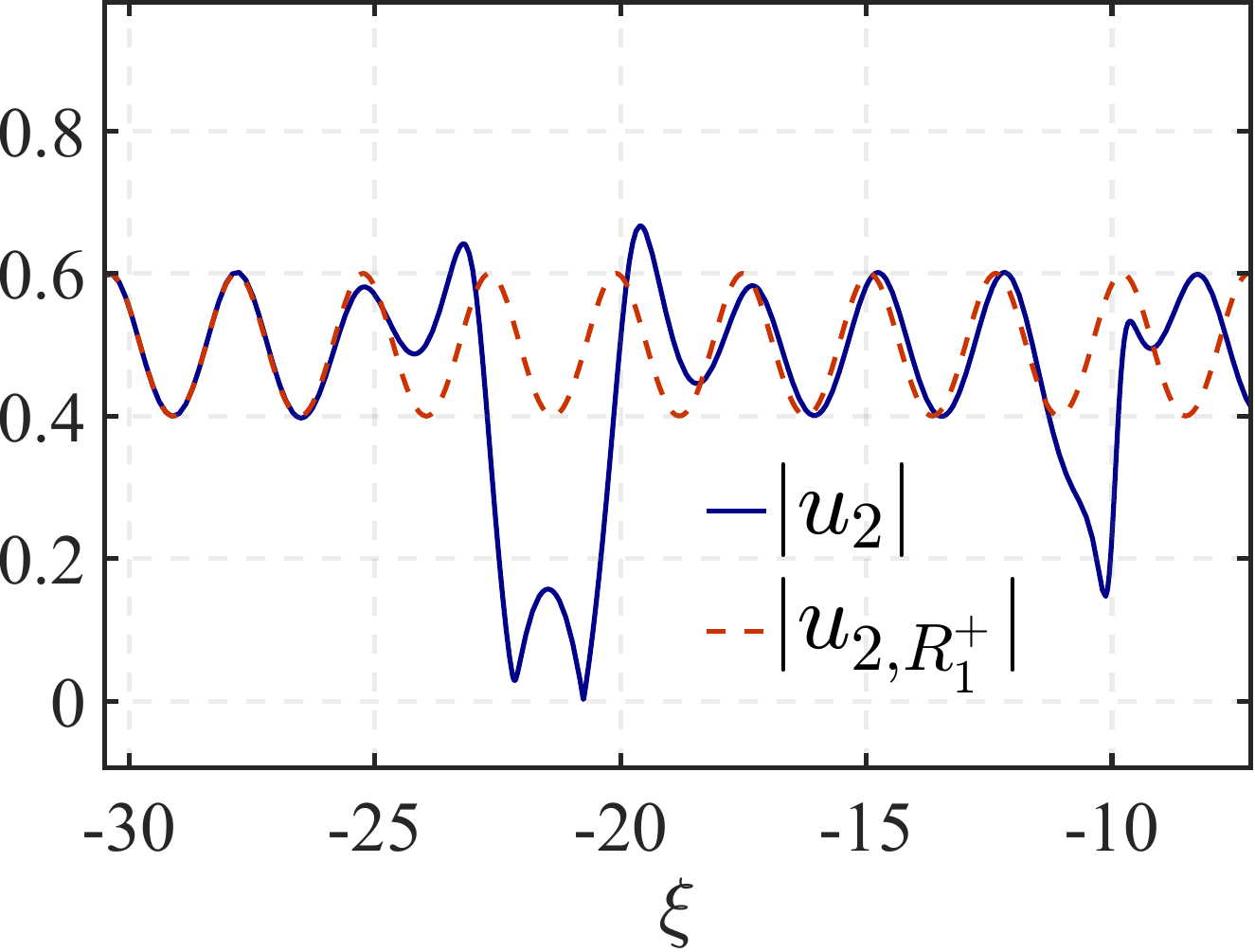}
			\caption{$|u_{2,R_1^+}|$}
		\end{subfigure}
		\caption{Comparison between the modulus $|u_2|$ of the two-elliptic localized solution (blue solid) and its asymptotic expression \eqref{Eq-asymptotic-Rk} (red dash-dotted) in the regions $R_1^{-}$, $R_2^{+}$ and $R_1^{+}$ at $t=-70$. The parameters are set as 	$
			\kappa=1.02\mathrm{i}, \rho=1.29-0.87\mathrm{i}, \omega_1=1.29, \omega_3=-1.98\mathrm{i}, z_1=-0.34+0.95\mathrm{i}, z_2=-0.69+0.85\mathrm{i},
			$. }
		\label{Fig-Asym-second-order-region1}
	\end{figure}

	\begin{figure}[H]
		\centering
		\begin{subfigure}{0.32\textwidth}
			\includegraphics[width=\textwidth]{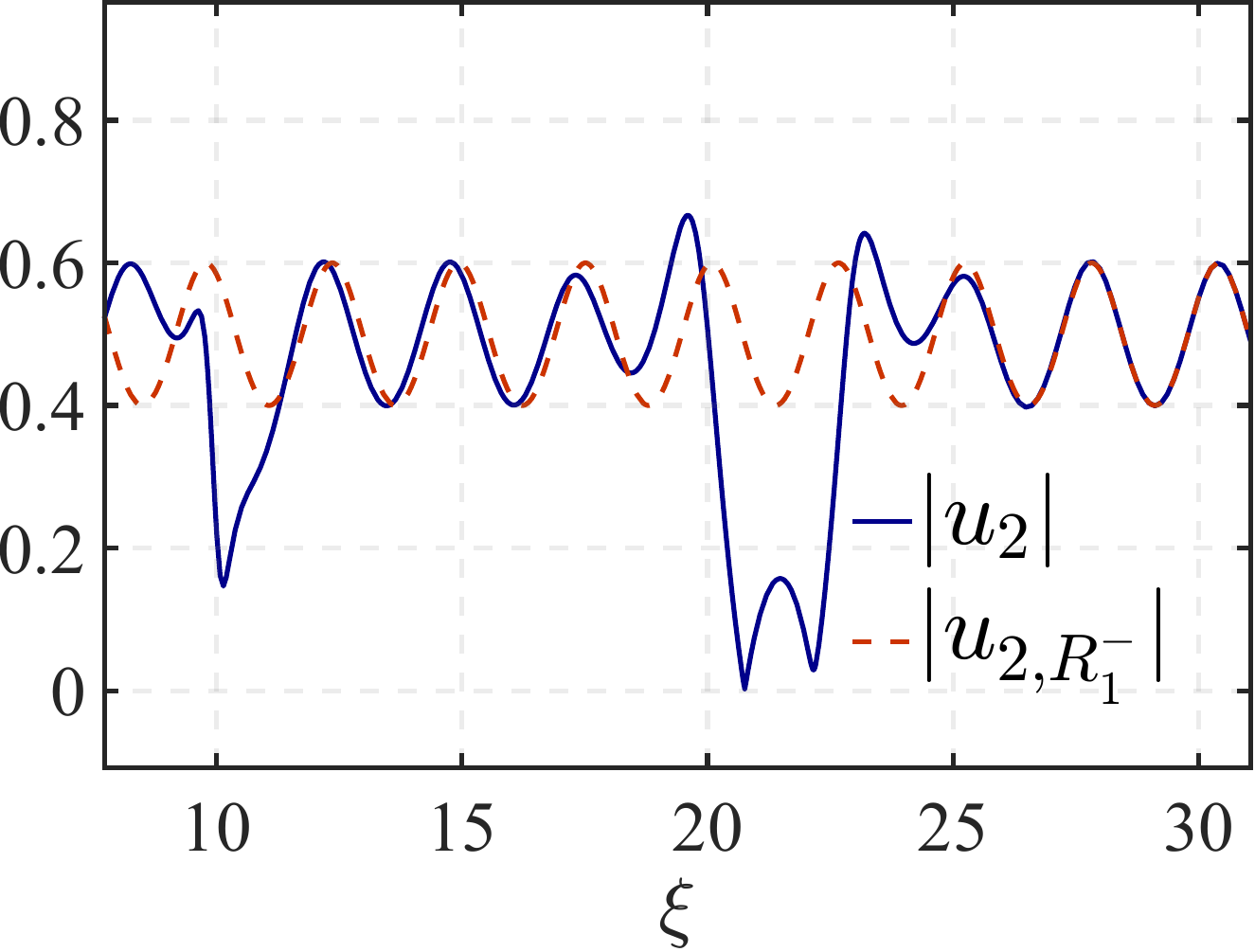}
			\caption{$|u_{2,R_1^+}|$}
		\end{subfigure}
		\begin{subfigure}{0.32\textwidth}
			\includegraphics[width=\textwidth]{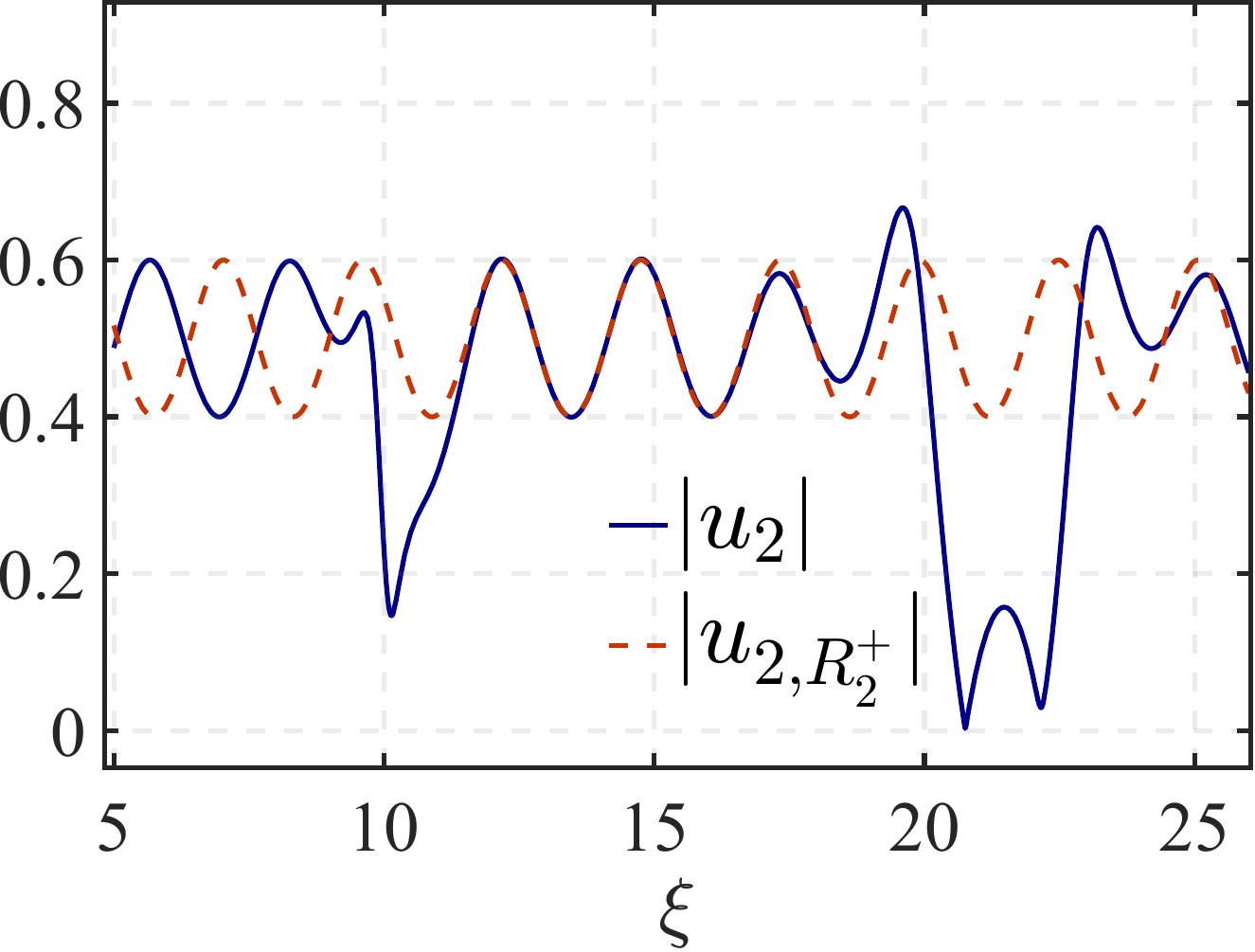}
			\caption{$|u_{2,R_2^-}|$}
		\end{subfigure}
		\begin{subfigure}{0.32\textwidth}
			\includegraphics[width=\textwidth]{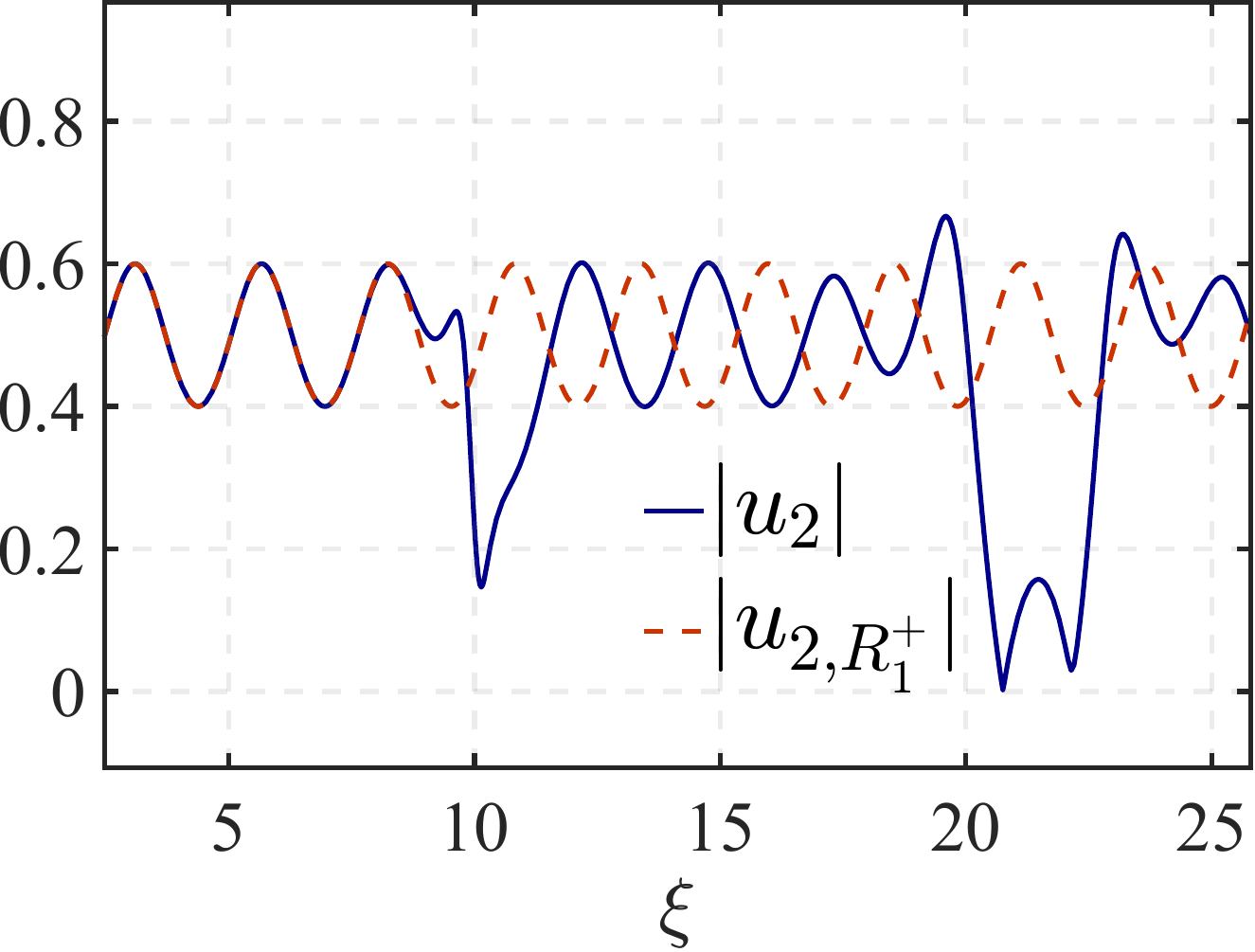}
			\caption{$|u_{2,R_1^+}|$}
		\end{subfigure}
		\caption{Comparison between the modulus $|u_2|$ of the two-elliptic localized solution (blue solid) and its asymptotic expression \eqref{Eq-asymptotic-Rk} (red dash-dotted) in the regions $R_1^{-}$, $R_2^{+}$ and $R_1^{+}$ at $t=-70$. The parameters are the same as in Figure \ref{Fig-Asym-second-order-region1}.}
		\label{Fig-Asym-second-order-region2}
	\end{figure}
	
	\paragraph{Asymptotics along the propagation directions}
By Theorems~\ref{thm:AB1} and~\ref{thm:first-order-reduction}, the solution degenerates into a first-order elliptic localized wave along each of the directions $L_s^{\pm}$, $s=1,2$. This is demonstrated in Figure~\ref{Fig-Asym-second-order-line}, which presents the profiles along $L_1^{+}$ and $L_2^{+}$ at $t=70$, and along $L_1^{-}$ and $L_2^{-}$ at $t=-70$. The elliptic localized waves preserve their shapes before and after the collision, indicating that the collision is elastic, exactly as predicted by Theorems~\ref{thm:AB1} and~\ref{thm:AB2}.
	
	\begin{figure}[H]
		\centering
		\begin{subfigure}{0.24\textwidth}
			\includegraphics[width=\textwidth]{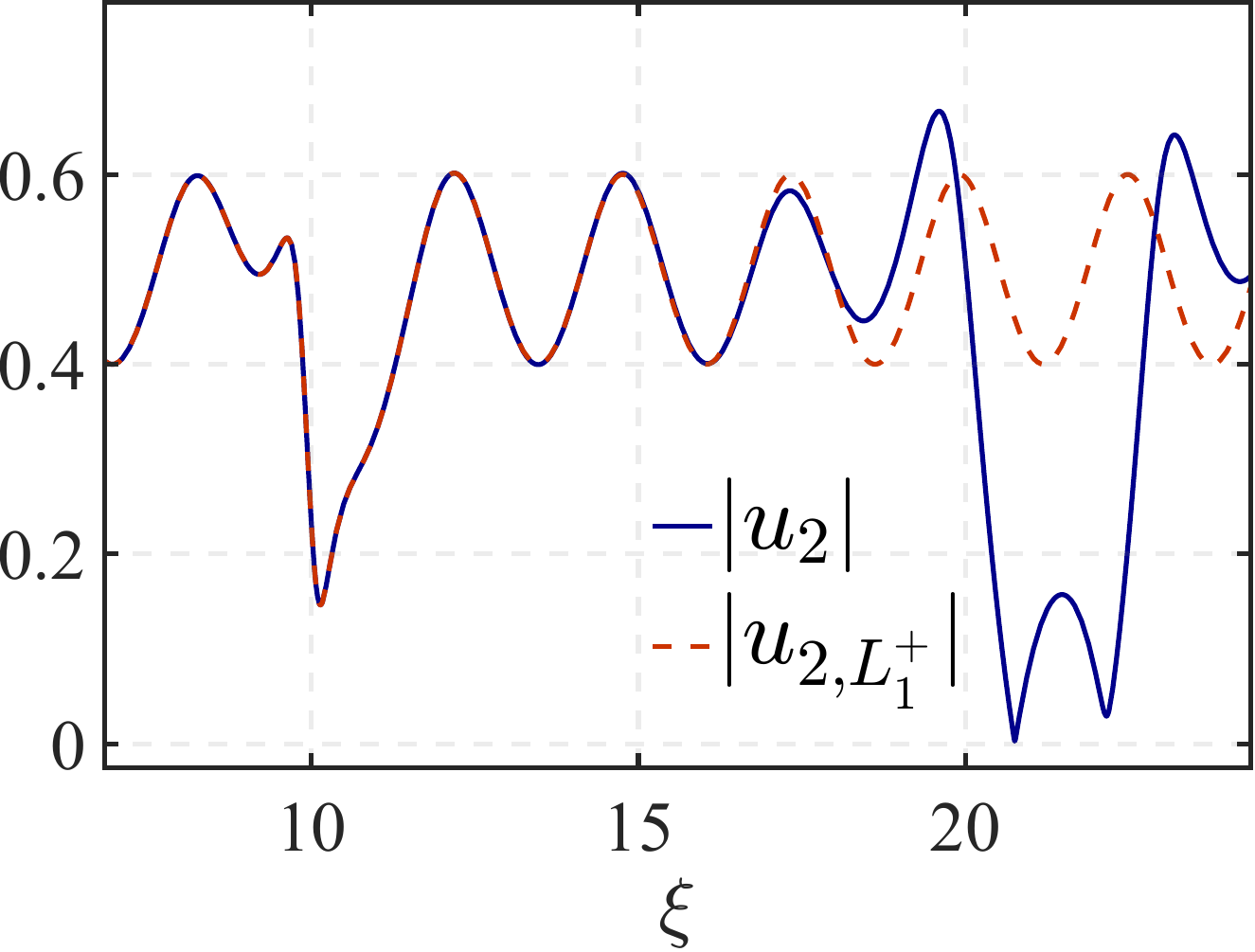}
			\caption{$|u_{2,L_1^+}|$, $t=70$}
		\end{subfigure}
		\begin{subfigure}{0.24\textwidth}
			\includegraphics[width=\textwidth]{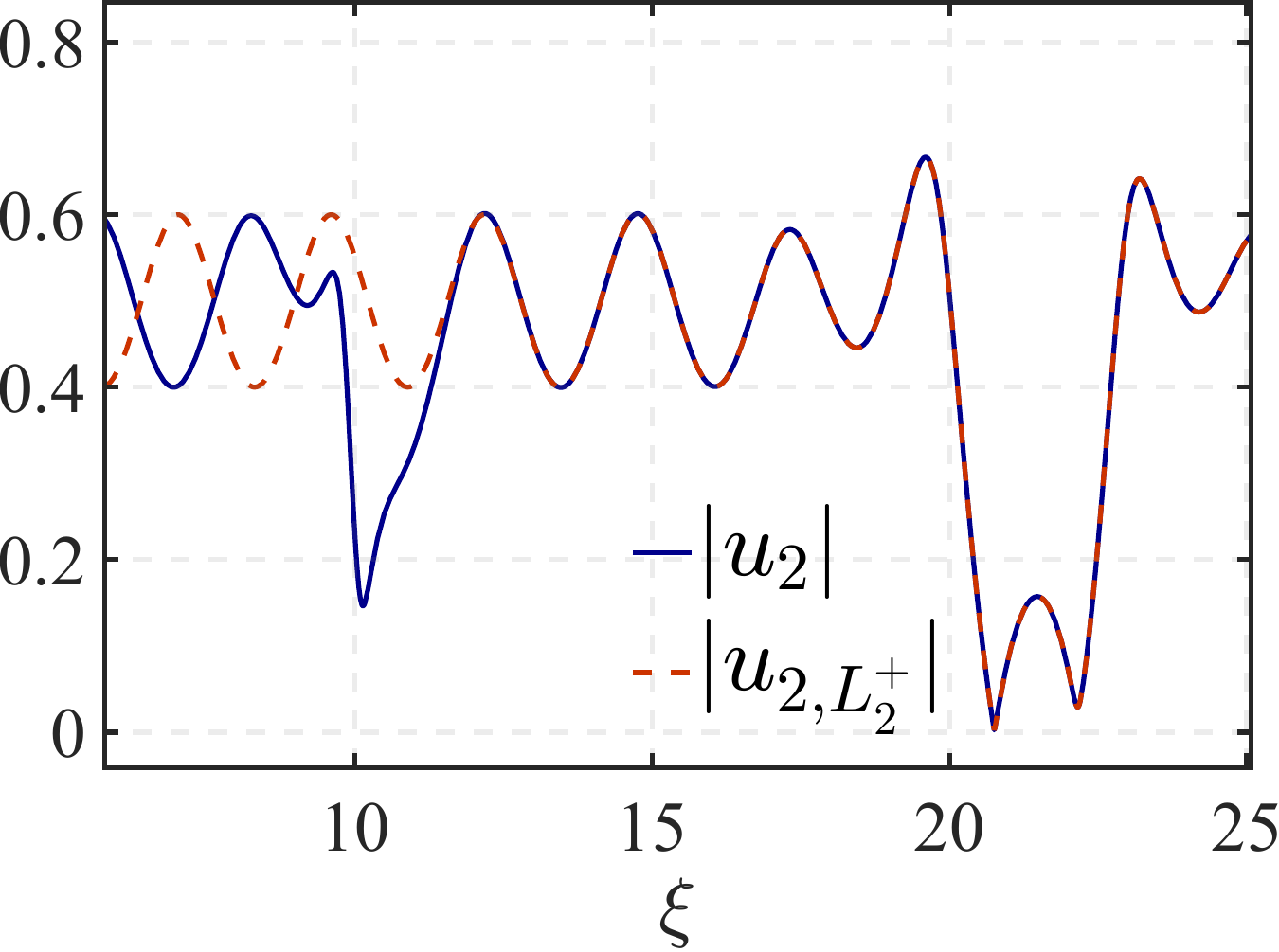}
			\caption{$|u_{2,L_2^+}|$, $t=70$}
		\end{subfigure}
		\begin{subfigure}{0.24\textwidth}
			\includegraphics[width=\textwidth]{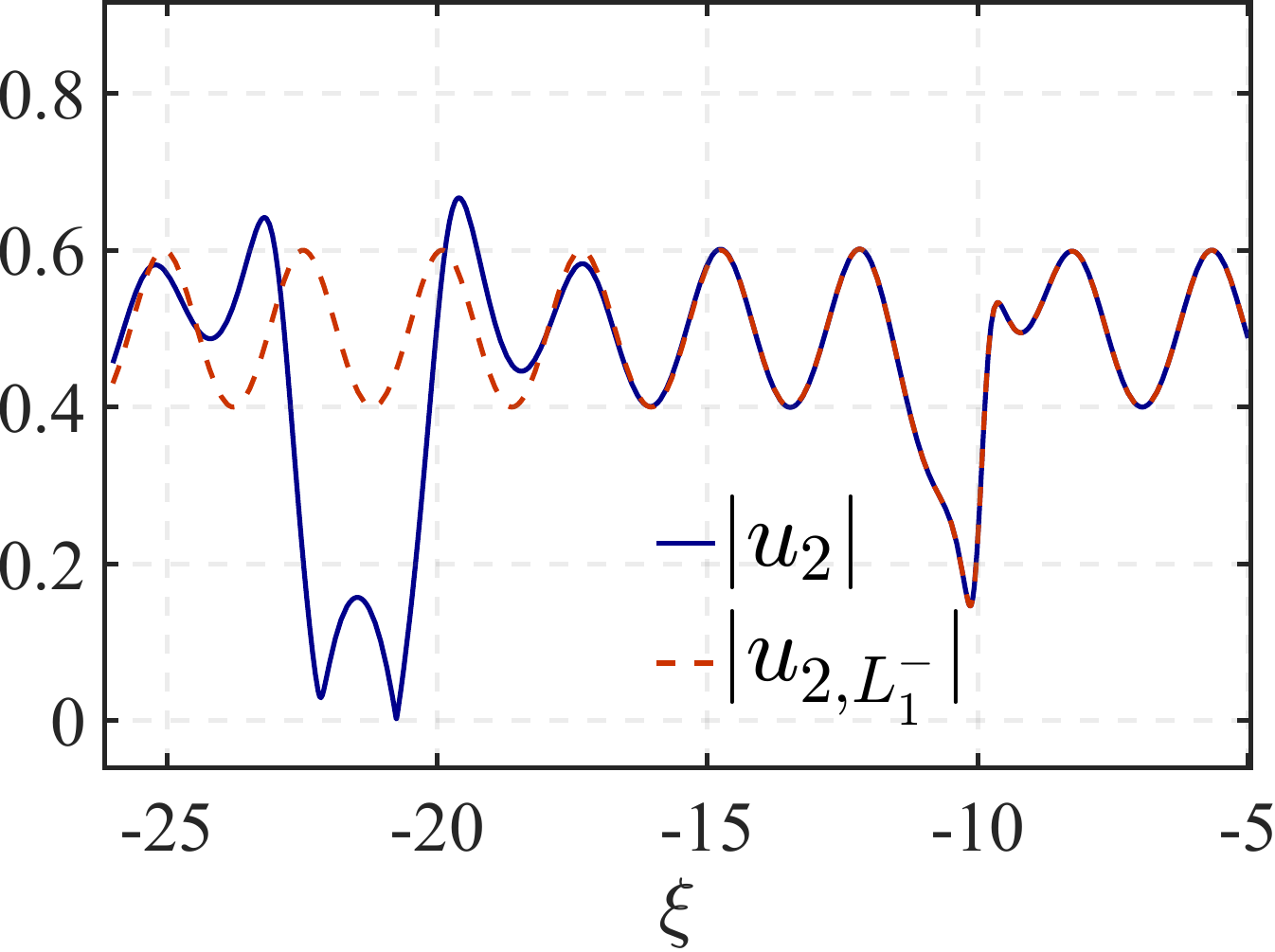}
			\caption{$|u_{2,L_1^-}|$, $t=-70$}
		\end{subfigure}
		\begin{subfigure}{0.24\textwidth}
			\includegraphics[width=\textwidth]{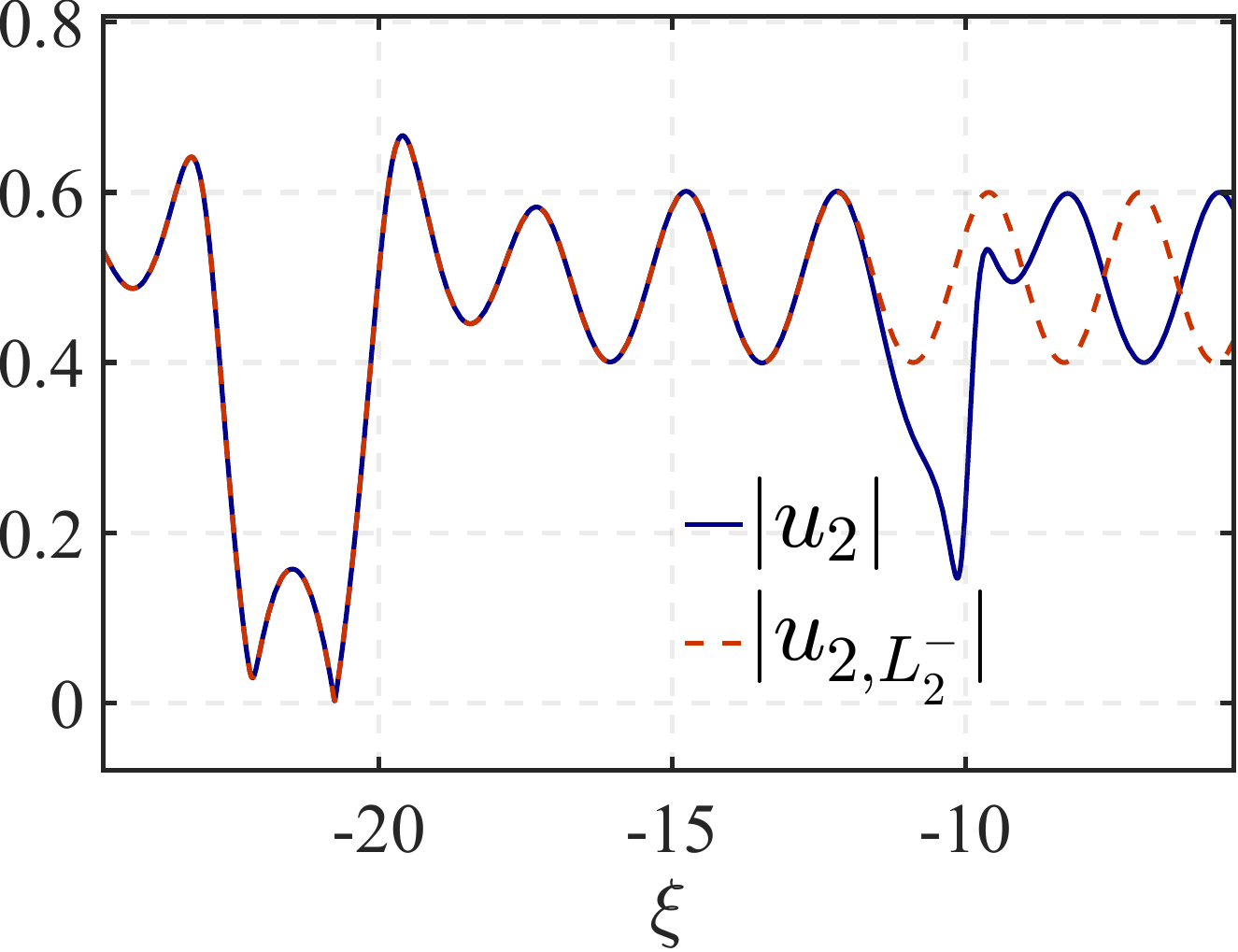}
			\caption{$|u_{2,L_2^-}|$, $t=-70$}
		\end{subfigure}
		\caption{Comparison between the modulus $|u_2|$ of the two-elliptic localized solution (blue solid) and its first-order asymptotic expression \eqref{Eq-uN-asym} (red dash-dotted) along the propagation directions $L_1^{\pm}$ and $L_2^{\pm}$. Panels (A) and (B) correspond to $L_1^{+}$ and $L_2^{+}$ at $t=70$; panels (C) and (D) correspond to $L_1^{-}$ and $L_2^{-}$ at $t=-70$. The parameters are the same as in Figure~\ref{Fig-Asym-second-order-region1}.}
		\label{Fig-Asym-second-order-line}
	\end{figure}

\section{Conclusion}\label{sec:conclusion}
In this paper, we have carried out a systematic study of the $N$-elliptic localized solutions of the Fokas--Lenells equation and of their long-time behavior. Starting from the Lax pair, we used Weierstrass elliptic functions to construct the elliptic function solutions of the equation together with the associated fundamental solution matrix. Applying the $N$-fold Darboux--B\"acklund transformation to this seed, we obtained the $N$-elliptic localized solutions and, with the help of a sigma-function version of the Cauchy determinant, recast them in a compact and fully explicit form. A detailed asymptotic analysis then showed that, as $t\to\pm\infty$, the $N$-elliptic localized solution decomposes into $N$ individual first-order elliptic localized waves that travel at distinct velocities over a shifted elliptic background: along each propagation direction the solution reduces to a single first-order elliptic localized wave, while between the propagation directions it collapses to a translated copy of the elliptic seed. The collisions between these constituent waves are elastic, and we further established a symmetry condition under which the solution is invariant under reflection through the origin, so that the collisions become strictly elastic. These results provide an explicit verification of the soliton resolution picture for exact solutions on elliptic backgrounds within the Kaup--Newell hierarchy, complementing the analogous constructions previously obtained for the modified Korteweg--de Vries and derivative nonlinear Schr\"odinger equations~\cite{ling2023elliptic,lingtang2026dnls}.

Several questions remain open. It would be of interest to extend the present uniform-parameter framework to higher-order elliptic localized solutions obtained through the generalized Darboux--B\"acklund transformation, to rogue waves on elliptic backgrounds, and to multi-component or discrete analogues of the Fokas--Lenells equation. A rigorous treatment of the soliton resolution conjecture on elliptic backgrounds, for instance through the nonlinear steepest-descent analysis of the corresponding Riemann--Hilbert problem~\cite{feng2020multi}, also appears to be a promising direction for future work.

	\section*{Acknowledgements}
Bao-Feng Feng was partially supported by the U.S. Department of Defense (DoD) and the Air Force Office of Scientific Research (AFOSR) under Grant No. W911NF2010276.
Guo-Fu Yu was supported by the National Natural Science Foundation of China under Grant Nos. 12371251 and 12175155.
	
	\section*{Conflict of interest}
	The authors have no conflicts of interest to declare.
	
	\appendix
	\section*{Appendix: Properties of Weierstrass Elliptic Functions}

This appendix summarizes the fundamental properties of the Weierstrass elliptic functions used in this work.

\subsection*{Definitions and fundamental parameters}
Let \(\Lambda = \{2m\omega_1 + 2n\omega_3 \mid m,n \in \mathbb{Z}\}\) be a lattice in the complex plane, with \(\omega_1 > 0\) and \(\operatorname{Im}(\omega_3) > 0\). The Weierstrass \(\wp\)-function is defined by
\begin{equation}\label{Eq-wp}
	\wp(\theta) = \frac{1}{\theta^2} + \sum_{\substack{\omega \in \Lambda \\ \omega \neq 0}} \left( \frac{1}{(\theta-\omega)^2} - \frac{1}{\omega^2} \right),
\end{equation}
which is doubly periodic with periods \(2\omega_1\) and \(2\omega_3\). The associated Weierstrass \(\zeta\)-function and \(\sigma\)-function are introduced via
\begin{equation}\label{Eq-derivative relation}
	\zeta'(\theta) = -\wp(\theta), \qquad \frac{\sigma'(\theta)}{\sigma(\theta)} = \zeta(\theta),
\end{equation}
with the normalizations
\[
\lim_{\theta\to 0}\Bigl(\zeta(\theta) - \frac{1}{\theta}\Bigr) = 0, \qquad \lim_{\theta\to 0}\frac{\sigma(\theta)}{\theta} = 1.
\]
The \(\wp\)-function satisfies the differential equation
\begin{equation}
	(\wp'(\theta))^2 = 4\wp(\theta)^3 - g_2\wp(\theta) - g_3 = 4\bigl(\wp(\theta)-e_1\bigr)\bigl(\wp(\theta)-e_2\bigr)\bigl(\wp(\theta)-e_3\bigr).
\end{equation}
The zeros \(e_1, e_2, e_3\) of the cubic satisfy
\begin{equation}\label{Eq-ei-conditions}
	e_1 + e_2 + e_3 = 0,
\end{equation}
and the values at the half-periods are \(\wp(\omega_i) = e_i\) for \(i=1,2,3\), where \(\omega_2 = -\omega_1 - \omega_3\). In the real case one has \(e_1 > e_2 > e_3\), while in the complex case \(e_1\) is real and \(e_2 = e_3^*\). The invariants \(g_2\) and \(g_3\) are defined by the lattice sums
\begin{equation}\label{Eq-g2g3-def}
	g_2 = 60 \sum_{\substack{\omega \in \Lambda \\ \omega \neq 0}} \frac{1}{\omega^4}, \qquad 
	g_3 = 140 \sum_{\substack{\omega \in \Lambda \\ \omega \neq 0}} \frac{1}{\omega^6}.
\end{equation}
In terms of the zeros \(e_1, e_2, e_3\), they are given explicitly by
\begin{equation}\label{Eq-g2g3-ei}
	g_2 =2(e_1^2 + e_2^2 + e_3^2), \qquad 
	g_3 = 4e_1e_2e_3.
\end{equation}
\subsection*{Symmetry, periodicity, expansions, and special values}

The \(\wp\)-function is even and doubly periodic, whereas the \(\zeta\)-function and \(\sigma\)-function are odd and quasi-periodic:
\begin{align}
	&\wp(\theta+2\omega_i) = \wp(\theta), \quad \wp(-\theta) = \wp(\theta), \label{Eq-wp-periodicity}\\
	&\zeta(\theta+2\omega_i) = \zeta(\theta) + 2\zeta(\omega_i), \quad \zeta(-\theta) = -\zeta(\theta), \label{Quasi-periodic-zeta}\\
	&\sigma(\theta+2\omega_i) = -\sigma(\theta)\,\exp\!\bigl(2\zeta(\omega_i)(\theta+\omega_i)\bigr), \quad \sigma(-\theta) = -\sigma(\theta). \label{Eq-quasi-periodic-sigma}
\end{align}
In a neighborhood of $\theta=0$, these functions admit the expansions
\begin{align}
	\wp(\theta) &= \theta^{-2} + \frac{g_2}{20}\theta^2 + \frac{g_3}{28}\theta^4 + O(\theta^6), \label{Eq-expansion-wp}\\
	\zeta(\theta) &= \theta^{-1} - \frac{g_2}{60}\theta^3 - \frac{g_3}{140}\theta^5 + O(\theta^7), \label{Eq-expansion-zeta}\\
	\sigma(\theta) &= \theta - \frac{g_2}{240}\theta^5 - \frac{g_3}{840}\theta^7 + O(\theta^9). \label{Eq-expansion-sigma}
\end{align}

\subsection*{Fundamental identities for Weierstrass functions}

The Weierstrass functions satisfy the following identities:
\begin{align}
	&\wp(\theta_1)-\wp(\theta_2) = -\frac{\sigma(\theta_1+\theta_2)\sigma(\theta_1-\theta_2)}{\sigma^2(\theta_1)\sigma^2(\theta_2)}, \label{Eq-wp-sigma-diff-theta}\\
	&\wp(\theta_1) + \wp(\theta_2) + \wp(\theta_1+\theta_2) = \frac{1}{4}\left( \frac{\wp'(\theta_1) - \wp'(\theta_2)}{\wp(\theta_1) - \wp(\theta_2)} \right)^2, \label{Eq-addition-wp-theta}\\
	&\zeta(\theta_1+\theta_2) - \zeta(\theta_1) - \zeta(\theta_2) = \frac{1}{2}\, \frac{\wp'(\theta_1) - \wp'(\theta_2)}{\wp(\theta_1) - \wp(\theta_2)}, \label{Eq-addition-zeta-theta}\\
	&\zeta(\theta_1)+\zeta(\theta_2)+\zeta(\theta_3)-\zeta(\theta_1+\theta_2+\theta_3) = \frac{\sigma(\theta_1+\theta_2)\sigma(\theta_1+\theta_3)\sigma(\theta_2+\theta_3)}{\sigma(\theta_1)\sigma(\theta_2)\sigma(\theta_3)\sigma(\theta_1+\theta_2+\theta_3)}.\label{Eq-zeta-sigma-relation-theta}
\end{align}
From the first identity, taking \(\theta_1\to\theta_2\) gives the half-argument formulas
\begin{align}
	\wp'(\theta) &= -\frac{\sigma(2\theta)}{\sigma^4(\theta)}, \label{Eq-half-argument-1}\\
	\zeta(2\theta) &= 2\zeta(\theta) + \frac{\wp''(\theta)}{2\wp'(\theta)}. \label{Eq-half-argument-2}
\end{align}
The half-period translation of the \(\wp\)-function also satisfies
\begin{equation}\label{Eq-wp-half-period}
	\bigl(\wp(\theta+\omega_i)-e_i\bigr)\bigl(\wp(\theta)-e_i\bigr) = (e_i-e_j)(e_i-e_k),
\end{equation}
where \(\{i,j,k\} = \{1,2,3\}\).

\subsection*{Addition formula for the sigma function}

A useful addition formula for the sigma function is
\begin{equation}\label{Eq-addition-sigma-theta}
	\begin{split}
		&\sigma(\theta_1+\theta_2)\sigma(\theta_1-\theta_2)\sigma(\theta_3+\theta_4)\sigma(\theta_3-\theta_4) \\
		&\quad + \sigma(\theta_1+\theta_3)\sigma(\theta_1-\theta_3)\sigma(\theta_4+\theta_2)\sigma(\theta_4-\theta_2) \\
		= &\sigma(\theta_3+\theta_2)\sigma(\theta_3-\theta_2)\sigma(\theta_1+\theta_4)\sigma(\theta_1-\theta_4),
	\end{split}
\end{equation}
valid for arbitrary \(\theta_i\in\mathbb{C}\).

\subsection*{Integration formulas}

The following integrals are required for constructing the elliptic localized solutions:
\begin{align}\label{Eq-integration-formulas}
	\int \frac{\wp^{\prime}(\theta)\,d\theta}{\wp(\theta)-\wp(\theta_1)} &= \ln\bigl(\wp(\theta)-\wp(\theta_1)\bigr),\\
	\int \frac{d\theta}{\wp(\theta)-\wp(\theta_1)} &= \frac{1}{\wp^{\prime}(\theta_1)}\left(\ln \frac{\sigma(\theta-\theta_1)}{\sigma(\theta+\theta_1)} + 2\theta\,\zeta(\theta_1)\right).
\end{align}
	
	\bibliographystyle{unsrt}
	\bibliography{Reference}
	
\end{document}